\documentclass[twosided]{easychair}

\usepackage{hyperref}%

\usepackage{amsmath}
\usepackage{amsthm}
\usepackage{mathtools,amssymb}
\usepackage{xcolor}

\usepackage{math-gen}
\usepackage{is}
\usepackage{sem} 
\usepackage{languages}
\usepackage{prooftree}
\usepackage{graphicx}

\title{Soundness conditions for big-step semantics \\[.5ex] {\small (Long version)} }
\titlerunning{Soundness conditions for big-step semantics} 

\author{Francesco Dagnino\inst{1}  \and Viviana Bono\inst{2} \and Elena Zucca\inst{1}\and \\ 
Mariangiola Dezani-Ciancaglini\inst{2}    }
\authorrunning{F. Dagnino et al.}
\institute{DIBRIS, University of Genova, Italy \and Computer Science Department, University of Torino, Italy} 

\begin{document}

\maketitle

\begin{abstract}
We propose a general proof technique to show that a predicate is \emph{sound}, that is, prevents stuck computation, \emph{with respect to a big-step semantics}. 
This result may look surprising, since in big-step semantics there is no difference between non-terminating and stuck computations, hence soundness cannot even be \emph{expressed}.
The key idea is to  define constructions yielding an extended version  of a given arbitrary big-step semantics, where the difference is made explicit. 
The extended semantics are exploited in the meta-theory, notably they are necessary to show that the proof technique works. 
However,  they remain \emph{transparent}  when using the proof technique, since   it  consists in checking three conditions on the original rules only,  as we illustrate by several examples.

\end{abstract}

\newcommand{\citet}[1]{\cite{#1}}  
\newtheorem{theorem}{Theorem}
\newtheorem{lemma}{Lemma}
\newtheorem{proposition}{Proposition}
\newtheorem{definition}{Definition}
\newtheorem{corollary}{Corollary}


\section{Introduction}

The semantics of programming languages or software systems specifies, for each program/system configuration, its final result, if any. 
In the case of non-existence of a final result,  there are  two possibilities:
\begin{itemize}
\item either the computation stops with no final result, and there is no means to compute further: \emph{stuck computation},
\item or the computation never stops: \emph{non-termination}.
\end{itemize}

There are two main styles to define operationally  a semantic relation:  the \emph{small-step} style \cite{Plotkin81,Plotkin04}, on top of a reduction relation representing single computation steps, or directly by a set of rules as in the \emph{big-step} style \cite{Kahn87}. 
Within a small-step semantics it is straightforward to make the distinction between stuck and non-terminating computations, while a typical drawback of the big-step style is that  they  are not distinguished (no judgement is derived in both cases).

For  this reason, even  though  big-step semantics is  generally more abstract, and sometimes  more intuitive to design and therefore to debug and extend, in the literature much more effort has been devoted to study the meta-theory of small-step semantics, providing properties, and related proof techniques. Notably, the \emph{soundness} of a type system (typing prevents stuck computation) can be proved by  \emph{progress} and \emph{subject reduction} \mbox{(also called \emph{type preservation}) \cite{WrightFelleisen94}. }

Our quest is then to provide a general proof technique to prove the soundness of a predicate with respect to an arbitrary big-step semantics.  How can we achieve this result,  given that in big-step formulation soundness cannot even be \emph{expressed}, since non-termination is modelled as the absence of a final result exactly like stuck computation? 
The  key idea is the following:  
\begin{enumerate}
\item We define constructions \emph{yielding an extended version of a given arbitrary big-step semantics}, where the difference  between stuckness and non-termination is made explicit.
In a sense, these constructions show that the distinction was ``hidden'' in the original semantics.
\item We provide a general proof technique by  identifying  \emph{three sufficient conditions} on the original big-step rules to prove soundness.
\end{enumerate}

Keypoint (2)'s three sufficient conditions are \emph{local preservation}, \emph{$\exists$-progress}, and \emph{$\forall$-progress}. For \emph{proving} the result that the three conditions actually ensure soundness, the setting up of the extended semantics from 
 the given  one is necessary, since otherwise, as said above, we could not even express the property. 

 \emph{However, the three conditions deal only with  the original rules of the given big-step semantics. }  
This means that, practically, in order to use the technique there is no need to deal with the extended semantics. This implies, in particular, that our approach does \emph{not} increase the original number of rules.  Moreover, the sufficient conditions are  checked only  
on  \emph{single rules}, which makes explicit the proof fragments typically needed in a proof of soundness.
Even though this is not exploited in this paper, this form of \emph{locality} means \emph{modularity}, in the sense that adding a new rule implies adding the corresponding proof fragment only. 

As an important by-product, in order to formally define and prove correct the keypoints (1) and (2), we propose a formalisation of  ``what is a big-step semantics"  which captures its essential features. Moreover, we support our approach by presenting  several  examples, demonstrating that: on the one hand, their soundness proof can be easily rephrased in terms of our technique, that is, by directly reasoning on big-step rules; on the other hand, our technique is essential when the property to be checked (for instance, the soundness of a type system) is \emph{not preserved} by intermediate computation steps, whereas it holds for the final result. On a side note, our examples concern type systems, but the meta-theory we present in this work holds for \mbox{any predicate.}

We describe now in more detail the constructions of keypoint (1).  
Starting from an arbitrary big-step judgment $\eval{\conf}{\res}$ that evaluates \emph{configurations} $\conf$ into \emph{results} $\res$, the \emph{first construction} produces an enriched judgement $\evaltr{\conf}{t}$ where $t$ is a \emph{trace}, that is, the (finite or infinite) sequence of all the (sub)configurations encountered during the evaluation. In this way, by interpreting coinductively the rules of the extended semantics,  an infinite trace models divergence (whereas no result corresponds to stuck computation). 
The \emph{second construction} is in a sense dual. It is the \emph{algorithmic} version of the well-known technique presented in Exercise 3.5.16 from the book \cite{Pierce02} of adding a special result $\Wrong$ explicitly modelling stuck computations (whereas no result corresponds to divergence). 

By trace semantics and $\Wrong$ semantics  we can express two flavours of soundness, \emph{soundness-may} and \emph{soundness-must}, respectively, and show the correctness of the corresponding proof technique. 
This achieves our original aim, and it should be noted that \emph{we define soundness  with respect to  a big-step semantics within a big-step formulation}, without resorting to a small-step style (indeed, the two extended semantics are themselves big-step). 


 Lastly,  
 we consider the issue of justifying on a formal basis that the two constructions are correct with respect to their expected meaning. 
For instance, for the $\Wrong$ semantics we would like to be sure that \emph{all} the cases are covered. 
To this end, we define  a \emph{third construction}, dubbed $\PEval$ for ``partial evaluation'', which makes explicit the \emph{computations} of a big-step semantics, intended as the sequences of execution steps of the naturally associated evaluation algorithm. 
Formally, we obtain a reduction relation on approximated proof trees, so termination, non-termination and stuckness can be defined as usual. 
Then, the correctness of traces and $\Wrong$ constructions is proved by showing they are equivalent to $\PEval$  for diverging and stuck computations, respectively.

In \refToSect{framework} we  illustrate the meta-theory on a running example. In \refToSect{constructions} we define the  trace and $\Wrong$  constructions. 
In  \refToSect{soundness} we express soundness in the \emph{must} and \emph{may} flavours,  introduce the proof technique, and prove its correctness. In \refToSect{examples} we show in detail how to apply the technique to the running example, and other significant examples. 
In \refToSect{PEV} we introduce the third construction and   prove 
that the three constructions are equivalent.  Finally, in \ref{sect:rw} and \ref{sect:conclu} we discuss related and further work and summarise our contribution. 

\section{A meta-theory for big-step semantics}\label{sect:framework}

We introduce a formalisation of  ``what is a big-step semantics''  that captures its essential features, subsuming 
a large class of examples (as testified in \refToSect{examples}).
This 
 enables a general formal reasoning on an arbitrary big-step semantics. 

A \emph{big-step semantics} is a triple $\Triple{\ConfSet}{\ResSet}{\RuleSet}$ where:
\begin{itemize} 
\item $\ConfSet$ is a set of \emph{configurations} $\conf$.
\item $\ResSet\subseteq \ConfSet$ is a set of \emph{results} $\res$. We define \emph{judgments} $\judg\equiv\eval{\conf}{\res}$, meaning that  configuration $\conf$ evaluates to result $\res$.
 Set  $\ConfSet(\judg) = \conf$ and $\ResSet(\judg) = \res$. 
\item $\RuleSet$ is a set of \emph{rules} $\rho$ of shape\\
\centerline{
\begin{tabular}{ll}
$\Rule{ \judg_1\ \ldots\ \judg_n \ \ \judg_{n+1} }{ \eval{\conf}{\ResSet(\judg_{n+1})}}$&\BigSpace also written in \emph{inline format}: $\inlinerule{\judg_1\ldots\judg_n}{\judg_{n+1}}{\conf}$\\
\end{tabular}
}
with $\conf\in\ConfSet{\setminus}\ResSet$, where $\judg_1\ldots\judg_n$ are the \emph{dependencies} and $\judg_{n+1}$ is the \emph{continuation}.  Set  $\ConfSet(\rho) {=}\conf$ and, for $i \in 1..n+1$, $\ConfSet(\rho, i) {=} \ConfSet(\judg_i)$ and $\ResSet(\rho, i) {=} \ResSet(\judg_i)$. 
\item For each result $\res \in \ResSet$, we implicitly assume a single axiom $\Rule{}{\eval{\res}{\res}}$. Hence, the only derivable judgment for $\res$ is $\eval{\res}{\res}$, which we will call a \emph{trivial} judgment.
\end{itemize}
 We will use the inline format, more concise and manageable, for the development of the meta-theory, e.g., in constructions.

A rule corresponds to the following evaluation process for a non-result configuration: first, dependencies {are evaluated} in the given order, then the continuation is evaluated and its result is returned as result of the entire computation. 

Rules as defined above specify an inference system \cite{Aczel77,LeroyGrall09}, whose inductive interpretation is, as usual, the semantic relation.  However, they carry slightly more structure with respect to standard inference rules. Notably, premises are a sequence rather than a set, and the last premise plays a special role. 
Such additional structure does not affect the semantic relation defined by the rules, but allows abstract reasoning about an arbitrary big-step semantics, in particular it is relevant for defining the three constructions.
In the following, we will write $\valid{\RuleSet}{\eval{\conf}{\res}}$ when the judgment $\eval{\conf}{\res}$ is derivable in $\RuleSet$. 

As customary, the (infinite) set of rules $\RuleSet$ is described by a finite set of meta-rules, each one with a finite number of premises.
As a consequence, the number of premises of rules is not only finite but \emph{bounded}. 
Since we have no notion of meta-rule, we model this feature (relevant in the following) as an explicit assumption:\\
\centerline{
\label{bp}
{\sf BP}\label{cond:bp} there exists $b \in \N$ such that, for each $\rho\equiv\inlinerule{\judg_1\ldots\judg_n}{\judg_{n+1}}{\conf}$, $n < b$.
}
We end this section illustrating the above definitions and conditions by a simple example: a $\lambda$-calculus with {natural constants,} successor and non-deterministic choice shown in \refToFigure{example}. 
\begin{figure}[t]
\begin{small}
\begin{math}
\begin{array}{rcll}
e &::=& x \mid v \mid e_1\ e_2 \mid \SuccExp e \mid \Choice{\e_1}{\e_2}& \text{expression}\\
v &::=& \const \mid \lambda x.e & \text{value}
\end{array}
\end{math}

\HSep

\begin{math}
\begin{array}{l}
\MetaRule
{val}
{}
{\eval{\val}{\val}}{}
\qquad
\MetaRule
{ app }
{\eval{\e_1}{\lambda x.e}\quad\eval{\e_2}{\val_2}\quad\eval{\subst{\e}{\val_2}{\x}}{\val}}
{\eval{\AppExp{\e_1}{\e_2}}{\val}}{}\Space
\MetaRule{succ}{\eval{\e}{\const}}{\eval{\SuccExp\e}{\const+1}}{}\\[3ex]
\MetaRule{choice}{\eval{\e_i}{v}}{\eval{\Choice{\e_1}{\e_2}}{v}}{i=1,2}
\end{array}
\end{math}

\HSep

\begin{math} 
\begin{array}{l}
\metainlinerule{app}{\eval{\e_1}{\LambdaExp{\x}{\e}}\ \eval{\e_2}{\val_2}}{\eval{\subst{\e}{\val_2}{\x}}{\val}}{\AppExp{\e_1}{\e_2}}\\
\metainlinerule{succ}{\eval{e}{n}}{\eval{n+1}{n+1}}{\SuccExp e} \\
\metainlinerule{choice}{\epsilon}{ \eval{e_i}{v} }{\Choice{\e_1}{\e_2}}\ i =1,2 
\end{array}
\end{math}
\end{small}
\caption{Example of big-step semantics}\label{fig:example}
\end{figure}
 We  present this example as an instance of our definition:
\begin{itemize}
\item Configurations and results are expressions, and values, respectively.\footnote{In general, configurations may include additional components, see \refToSect{fjl}.}
\item  To have the set of (meta-)rules in our required shape, abbreviated in inline format in the bottom section of the figure:
\begin{itemize}
\item axiom \rn{val} can be omitted (it is implicitly assumed)
\item in \rn{app} we consider premises as a sequence rather than a set (the third premise is the continuation)
\item in \rn{succ}, which has no continuation, we add a dummy continuation
\item on the contrary, in \rn{choice} there is only the continuation (dependencies are the empty sequence, denoted $\epsilon$ in the inline format).
\end{itemize}
\end{itemize}
Note that \rn{app}  corresponds to  the standard left-to-right evaluation order.  We could have chosen the right-to-left order instead:\\
\centerline{
\begin{math}
\metainlinerule{app-r}{\eval{e_2}{v_2}\ \eval{e_1}{\lambda x.e}\ }{\eval{e[v_2/x]}{v}}{e_1\ e_2}
\end{math}
}
or even opt for a non-deterministic approach by taking both rules \rn{app} and \rn{app-r}. {As said above, these different choices do} not affect the semantic relation $\eval{\conf}{\res}$ defined by the inference system, which is always the same. However, {they} will affect the way the extended semantics distinguishing stuck computation and non-termination is constructed. Indeed, if the evaluation of $\e_1$ and $\e_2$ is stuck and non-terminating, respectively, we should obtain stuck computation with rule  \rn{app}  and non-termination with  rule \rn{app-r}.
 
 In summary, to see a typical big-step semantics as an instance of our definition, it is enough to assume an order (or more than one) on premises, make implicit the axiom for results, and add a dummy continuation when needed. In the examples (\refToSect{examples}), we will assume a left-to-right order on premises, 
 we will present the rules in both styles. 
 In the technical part (\refToSect{constructions},  \refToSect{soundness} and \refToSect{partial-eval}) we will adopt the inline format.


\section{Extended semantics} \label{sect:constructions}

In the following, we assume a big-step semantics  $\Triple{\ConfSet}{\ResSet}{\RuleSet}$  and describe two  constructions which make the distinction between non-termination and stuck computation explicit.
In both cases, the approach is based on well-know ideas; the novel contribution is that, thanks to the meta-theory in \refToSect{framework}, we provide a \emph{general} construction working on an arbitrary big-step semantics.  
\subsection{Traces}\label{sect:traces}
We denote by $\ConfSet^\star$, $\ConfSet^\omega$, and $\ConfSet^\infty = \ConfSet^\star \cup \ConfSet^\omega$, respectively, the sets of finite, infinite, and possibly infinite \emph{traces}, that is, sequences of configurations. 
We write $t \cdot t'$ for concatenation of $t{\in} \ConfSet^\star$ with $t'{\in} \ConfSet^\infty$.

We  derive, from the judgement $\eval{\conf}{\res}$, an enriched big-step judgement $\evaltr{\conf}{t}$ with $t\in\ConfSet^\infty$.
Intuitively, $t$ keeps trace of all the configurations visited during the evaluation, starting from $\conf$ itself. To define the trace semantics, we construct, starting from $\RuleSet$, a new set of rules $\TrRuleSet$, 
which are of two kinds:
\begin{description}
\item[{trace introduction}] These rules enrich the standard semantics by finite traces: for each $\rho \equiv \inlinerule{\judg_1\ldots\judg_n}{\judg_{n+1}}{\conf}$ in $\RuleSet$, and finite traces $\ t_1, \ldots, t_{n+1} {\in} \List{\ConfSet}$, we add the rule\\ 
\centerline{$
\Rule{
	\evaltr{\ConfSet(\judg_1)}{t_1 \cdot \ResSet(\judg_1)} \Space \ldots \Space \evaltr{\ConfSet(\judg_{n+1})}{t_{n+1} \cdot \ResSet(\judg_{n+1})}
}{ \evaltr{\conf}{\conf\cdot t_1\cdot \ResSet(\judg_1)\cdot \ldots\cdot t_{n+1}\cdot \ResSet(\judg_{n+1})} }
$}
We denote this rule by $\tracerule{\rho}{t_1, \ldots, t_{n+1}}$, to highlight the relationship with the original rule $\rho$.
We also add one axiom $\Rule{}{ \evaltr{\res}{\res} }$ for each result $\res$.

Such rules derive judgements $\eval{\conf}{t}$ with $t{\in}\ConfSet^\star$,  for  convergent computations.

\item[divergence propagation]  
These rules propagate divergence{, that is, if a (sub)configuration in the premise of a rule diverges, then the subsequent premises are ignored and the configuration in the conclusion diverges as well:} for each $\rho \equiv \inlinerule{\judg_1 \ldots \judg_n}{\judg_{n+1}}{\conf}$ in $\RuleSet$, index $i {\in} 1..n+1$, finite traces $t_1, \ldots, t_{i-1} \in \List{\ConfSet}$,  and infinite trace $t$, we add the rule:\\
\begin{small}
\centerline{$
\Rule{
	\evaltr{\ConfSet(\judg_1)}{t_1\cdot \ResSet(\judg_1)}
	\Space \ldots \Space
	\evaltr{\ConfSet(\judg_{i-1})}{t_{i-1} \cdot \ResSet(\judg_{i-1})} 
	\Space  
	\eval{\ConfSet(\judg_i)}{t}
}{ \eval{\conf}{\conf\cdot t_1\cdot \ResSet(\judg_1) \cdot \ldots \cdot t_{i-1} \cdot \ResSet(t_{i-1}) \cdot t }}
$}
\end{small}
We denote this rule by $\divtracerule{\rho}{i}{t_1, \ldots, t_{i-1}}{t}$ to highlight the relationship with the original rule $\rho$. 
These rules derive judgements $\evaltr{\conf}{t}$ with $t \in \InfList{\ConfSet}$, modelling diverging computations. 
\end{description}

The  inference system $\TrRuleSet$ must be interpreted \emph{coinductively}, to properly model diverging computations.
Indeed, since there is no axiom introducing an infinite trace,  they  can be derived only by an infinite proof tree.
We write $\valid{\TrRuleSet}{\evaltr{\conf}{t}}$ when the judgment $\evaltr{\conf}{t}$ is derivable in $\TrRuleSet$.  In the following, 
 given a judgement $\judg = \evaltr{\conf}{t}$, we set $\Trace(\judg) = t$.

We show in \refToFigure{example-trace} the rules obtained starting from meta-rule \rn{app} of the example (for other meta-rules the outcome is analogous).
\begin{figure}[t]
\begin{small}
\begin{quote}
$\MetaRule{app-trace}{\evaltr{e_1}{t_1\cdot \lambda x.e}\Space\evaltr{e_2}{t_2\cdot v_2}\Space\evaltr{e[v_2/x]}{t\cdot v}}{ \evaltr{e_1\ e_2}{e_1\ e_2 \cdot t_1 \cdot \lambda x.e \cdot t_2\cdot v_2\cdot t\cdot v}}
{t_1,t_2,t{\in}\ConfSet^\star}$\\[2.5ex]
$\MetaRule{div-app-1}{\evaltr{e_1}{t}}{ \evaltr{e_1\ e_2}{e_1\ e_2\cdot t}}{t{\in}\InfList{\ConfSet}}$\Space
$\MetaRule{div-app-2}{\evaltr{e_1}{t_1\cdot \lambda x.e}\Space\evaltr{e_2}{t}}{ \evaltr{e_1\ e_2}{e_1\ e_2\cdot t_1\cdot \lambda x.e \cdot t}}{t_1{\in}\ConfSet^\star, t{\in}\InfList{\ConfSet}}$\\[2.5ex]
$\MetaRule{div-app-3}{\evaltr{e_1}{t_1\cdot \lambda x.e}\Space\evaltr{e_2}{t_2\cdot v_2}\Space\evaltr{e[v_2/x]}{t}}{ \evaltr{e_1\ e_2}{e_1\ e_2\cdot t_1\cdot \lambda x.e \cdot t_2\cdot v_2\cdot t}}{t_1, t_2\in\ConfSet^\star, t\in\InfList{\ConfSet}}$
\end{quote}
\end{small}
\caption{Trace semantics for application}\label{fig:example-trace}
\end{figure}

For instance,  set $\Omega=\AppExp{\omega}{\omega}=\AppExp{(\lambda \x.\AppExp{\x}{\x})}{(\lambda \x.\AppExp{\x}{\x})}$, and $t_\Omega$ the infinite trace $\Omega \cdot \omega \cdot \omega \cdot \Omega \cdot \omega\cdot\omega\cdot\ldots$, it is easy to see that the judgment $\evaltr{\Omega}{t_\Omega}$ can be derived by the following infinite tree:\footnote{To help the reader, we add  equivalent expressions with a grey background.}\\
\begin{small}
\centerline{$
\MetaRule{div-app3}{\MetaRule{trace-val}{}{\evaltr{\omega}{\omega}}{}\quad\MetaRule{trace-val}{}{\evaltr{\omega}{\omega}}{}\quad\MetaRule{div-app3}{\vdots}{\evaltr{\meta{\AppExp{\omega}{\omega}\equiv}\subst{(\AppExp{\x}{\x})}{\omega}{\x}}{t_\Omega}}{}}{\eval{\Omega}{\Omega\cdot \omega \cdot \omega \cdot t_\Omega\meta{\equiv t_\Omega}}}{}
$}
\end{small}
Note that \emph{only} the judgment $\evaltr{\Omega}{t_\Omega}$ can be derived, that is, the trace semantics of $\Omega$ is uniquely determined to be $t_\Omega$, since the infinite proof tree forces the equation $t_\Omega=\Omega\cdot\omega\omega\cdot t_\Omega$.  This example is a cyclic proof, but there are divergent  computations  
with no circular derivation. 

 The trace construction satisfies  the following property:

\begin{proposition}  \label{prop:evaltr-fin-tr}
If $\valid{\TrRuleSet}{\evaltr{\conf}{t}}$ holds, then the following are equivalent:
\begin{enumerate}
\item $\valid{\TrRuleSet}{\evaltr{\conf}{t}}$ holds by a finite derivation
\item $t = t' \cdot \res$ for some $t' \in \List{\ConfSet}$ and $\res \in \ResSet$
\item $t$ is finite.
\end{enumerate}
\end{proposition}
\begin{proof}
First note that, if $\evaltr{\conf}{t}$ is the conclusion of a rule $\rho^\trlb$, then $t$ is infinite iff $\rho^\trlb$ is a divergence propagation rule iff there is a premise $\judg$ (the last one) of $\rho^\trlb$ such that $\Trace(\judg)$ is infinite as well. 
Now we prove the following chain of implications: $1\Rightarrow 2 \Rightarrow 3 \Rightarrow 1$. 

To prove $1 \Rightarrow 2$, we proceed by induction on the derivation. 
If the last applied rule is an axiom (base case), then $\conf  = \res \in \ResSet$ and $ t = \res$. 
Otherwise, we have applied a rule $\rho^\trlb$ with premises $\judg_1, \ldots, \judg_{n+1}$, hence, for all $i \in 1..n+1$, $\valid{\TrRuleSet}{\judg_i}$ holds by a finite derivation, and so, by induction hypothesis, we have $\Trace(\judg_i) = t_i \cdot \res_i$ for some $t_i \in \List{\ConfSet}$ and $\res_i \in \ResSet$. 
Then, $\rho^\trlb$ is a trace introduction rule, since, if it were a divergence propagation rule, one of its premises would have an infinite trace, which contradicts the induction hypothesis. 
Hence, we have $t = \conf \cdot \Trace(\judg_1) \cdot \cdots \cdot \Trace(\judg_{n+1}) = t' \cdot \res_{n+1}$, as needed. 

The implication $2\Rightarrow 3$ is trivial. 
To prove $3 \Rightarrow 1$, we proceed by induction on (the length of) $t$, which is possible since $t$ is finite. 
The judgement $\evaltr{\conf}{t}$ is derivable by hypothesis, hence it has a (possibly infinite) derivation. Let us denote by $\rho^\trlb$ the last applied rule in this derivation. 
Since $t$ is finite, $\rho^\trlb$ is not a divergence propagation rule, hence we have only two cases:
$\rho^\trlb$ is an axiom, and then the thesis is trivial, or $\rho^\trlb$ is a trace introduction rule. 
In this second case, we have $\rho^\trlb \equiv \tracerule{\rho}{t_1, \ldots, t_{n+1}}$, with premises $\judg_1, \ldots, \judg_{n+1}$,  and $t = \conf \cdot \Trace(\judg_1) \cdot \cdots \cdot \Trace(\judg_{n+1})$;
hence, for all $i \in 1..n+1$, $\Trace(\judg_i)$ is finite and strictly shorter than $t$, thus, by induction hypothesis, $\valid{\TrRuleSet}{\judg_i}$ holds by a finite proof tree. 
Therefore, by applying $\rho^\trlb$ to the finite derivations for $\judg_1, \ldots, \judg_{n+1}$ we get a finite derivation for $\evaltr{\conf}{t}$, as needed. 
\end{proof}

The main consequence of \refToProp{evaltr-fin-tr} is that on judgements $\evaltr{\conf}{t}$ where $t$ is finite we can reason by induction on ({trace introduction}) rules, even  though  the set {$\TrRuleSet$} of rules is {interpreted} coinductively.
Furthermore, it ensures that if $\evaltr{\conf}{t}$ is derivable with $t$ finite, then $t$ terminates with a result. 

 The trace construction is \emph{conservative} with respect to the original semantics, that is, converging computations are not affected. 
\begin{theorem} \label{thm:eq-finite_0}
$\valid{\TrRuleSet}{\evaltr{\conf}{t\cdot \res}}$ for some $t \in \List{\ConfSet}$ iff $\valid{\RuleSet}{\eval{\conf}{\res}}$.
\end{theorem}
\begin{proof}
Thanks to \refToProp{evaltr-fin-tr}, in both directions it is a straightforward induction on rules. 
\end{proof}

\subsection{Wrong} \label{sect:wrong}

A well-known technique \cite{Pierce02}  (Exercise 3.5.16)  to distinguish between stuck and diverging computations, in a sense ``dual'' to the previous one, is to  add  a special result $\Wrong$, so that $\eval{\conf}{\Wrong}$ means that the evaluation of $\conf$ goes stuck.

 In this case, to define an ``automatic'' version of the construction, starting from $\Triple{\ConfSet}{\ResSet}{\RuleSet}$, is a non-trivial problem. Our solution is based on defining 
a relation on rules, modelling \emph{equality up to a certain index $i$},  also used for other aims  in the following.
Consider $\rho \equiv \inlinerule{\judg_1\ldots\judg_n}{\judg_{n+1}}{\conf}$,  
$\rho' \equiv \inlinerule{\judg'_1\ldots \judg'_m}{\judg'_{m+1}}{\conf'}$, and an index $i \in 1..\min (n+1, m+1)$, then 
 $\rho \sim_i \rho'$  if  
\begin{itemize}
\item $\conf = \conf'$
\item for all $k < i$, $\judg_k = \judg'_k$
\item $\ConfSet(\judg_i) = \ConfSet(\judg'_i)$
\end{itemize}
Intuitively, this means that rules $\rho$ and $\rho'$ model the same computation until the $i$-th premise. 
 Using this relation, we  derive, from the judgment $\eval{\conf}{\res}$, an enriched big-step judgement $\eval{\conf}{\res_\wrlb}$ where $\res_\wrlb \in \ResSet \cup \{\Wrong\}$, defined  by a set of rules $\WrRuleSet$ containing all rules in $\RuleSet$ and two other kinds of rules:

\begin{description}
\item [$\Wrong$ {introduction}]  {These rules derive $\Wrong$ whenever the (sub)configuration in a premise of a rule reduces to a result which is not admitted in such (or any equivalent) rule:
}
for each $\rho \equiv \inlinerule{\judg_1\ldots\judg_n}{\judg_{n+1}}{\conf}$ in $\RuleSet$, index $i \in 1..n+1$, and result $\res \in \ResSet$, 
if for all rules $\rho'$ such that $\rho \sim_i \rho'$, $\ResSet(\rho', i) \ne \res$, then we add the rule $\wrongrule{\rho}{i}{\res}$ as follows:\\
\centerline{$
\Rule{
	\judg_1 \ldots \judg_{i-1}
	\Space 
	\eval{\ConfSet(\judg_i)}{\res}
}{ \eval{\conf}{\Wrong} }
$}
We also add an axiom $\Rule{}{ \eval{\conf}{\Wrong} }$ for each configuration $\conf$ which is not the conclusion of any rule. 
\item [$\Wrong$ propagation] These rules propagate $\Wrong$ analogously to those for divergence propagation:
for each $\rho \equiv \inlinerule{\judg_1\ldots\judg_n}{\judg_{n+1}}{\conf}$ in $\RuleSet$, and index $i \in 1..n+1$, we add the rule $\proprule{\rho}{i}{\Wrong}$ as follows:\\
\centerline{$
\Rule{
	\judg_1 \ldots \judg_{i-1}
	\Space
	\eval{\ConfSet(\judg_i)}{\Wrong}
}{ \eval{\conf}{\Wrong} }
$}
\end{description}
We write $\valid{\WrRuleSet}{\eval{\conf}{\res_\wrlb}}$ when the judgment $\eval{\conf}{\res_\wrlb}$ is derivable in $\WrRuleSet$. 

We show in \refToFigure{example-wrong} the meta-rules for $\Wrong$ introduction and propagation constructed starting from  those for application and successor.
\begin{figure}
\begin{small}
\begin{quote}
$\MetaRule{wrong-app}{\eval{\e_1}{n}}{\eval{\AppExp{\e_1}{\e_2}}{\Wrong}}{}\Space\MetaRule{wrong-succ}{\eval{\e}{\lambda x.\e'}}{\eval{\SuccExp \e}{\Wrong}}{}$\\[2.5ex]
$\MetaRule{prop-app-1}{\eval{\e_1}{\Wrong}}{ \eval{\AppExp{\e_1}{\e_2}}{\Wrong}}{}$\Space
$\MetaRule{prop-app-2}{\eval{\e_1}{\lambda x.e}\Space\eval{e_2}{\Wrong}}{ \eval{\AppExp{\e_1}{\e_2}}{\Wrong}}{}$\\[2.5ex]
$\MetaRule{prop-app-3}{\eval{\e_1}{\lambda x.e}\Space\eval{e_2}{v_2}\Space\eval{e[v_2/x]}{\Wrong}}{ \eval{\AppExp{\e_1}{\e_2}}{\Wrong}}{}\Space\MetaRule{prop-succ}{\eval{\e}{\Wrong}}{\eval{\SuccExp \e}{\Wrong}}{}$
\end{quote}
\end{small}
\caption{Semantics with $\Wrong$ for application and successor}\label{fig:example-wrong}
\end{figure}
For instance, rule \rn{wrong-app} is introduced since in the original semantics there is rule \rn{app} with $\AppExp{\e_1}{\e_2}$ in the consequence and $\e_1$ in the first premise, but 
there is no equivalent rule (that is, with $\AppExp{\e_1}{\e_2}$ in the consequence and $\e_1$ in the first premise) such that the result in the first premise is $n$.

 The $\Wrong$ construction is conservative as well. 
\begin{theorem} \label{thm:eq-finite_1}
$\valid{\WrRuleSet}{\eval{\conf}{\res}}$ iff $\valid{\RuleSet}{\eval{\conf}{\res}}$.
\end{theorem}

\begin{proof}
The implication $\valid{\RuleSet}{\eval{\conf}{\res}}\Rightarrow \valid{\WrRuleSet}{\eval{\conf}{\res}}$ holds since $\RuleSet \subseteq \WrRuleSet$ by construction.
To prove the vice versa, we proceed by induction on rules.
The only relevant cases are rules in $\RuleSet$, because rules in $\WrRuleSet\setminus \RuleSet$ allow only to derive judgements of shape $\eval{\conf}{\Wrong}$.
Hence, the thesis is immediate. 
\end{proof}


\section{Expressing and proving soundness}\label{sect:soundness}

A predicate (for instance, a typing judgment) is \emph{sound} when, informally, a program satisfying the predicate (e.g., a well-typed program) cannot \emph{go wrong}, following Robin Milner's slogan \cite{Milner78}.
In small-step style, as firstly formulated in \citet{WrightFelleisen94}, this is naturally expressed as follows: well-typed programs  never reduce to 
terms which neither are values, nor can be further reduced  (called \emph{stuck} terms). The standard technique to ensure soundness is by subject reduction (well-typedness is preserved by reduction) and progress (a well-typed term is not stuck).

 We  discuss how soundness can be expressed  for  the two approaches previously presented and  we  introduce sufficient conditions. 
In other words, we provide a proof technique to show  the  soundness of a predicate with respect to a big-step semantics. 
 As  mentioned in the Introduction, the extended semantics  is only needed 
to prove the correctness of technique, whereas to \emph{apply} the technique for a given big-step semantics it is enough to reason on the original rules. 

\subsection{Expressing soundness} \label{sect:must-may}

In the following, we assume a big-step semantics  $\Triple{\ConfSet}{\ResSet}{\RuleSet}$, and an \emph{indexed  predicate on configurations}, that is,  a  family $  \Pred= (\Pred_\idx)_{\idx \in \IdxSet}$, for $\IdxSet$ set of \emph{indexes}, with $\Pred_\idx \subseteq \ConfSet$. 
A representative case is that, as in the examples of \refToSect{examples}, the predicate is a typing judgment and the indexes are types; however, the proof technique could be applied to other kinds of predicates.  When there is no ambiguity, we also denote by $\Pred$ the corresponding predicate $\bigcup_{\idx \in \IdxSet} \Pred_\idx$ on $\ConfSet$ (e.g., to be well-typed with an arbitrary type).

To discuss how to express soundness of $\Pred$, first of all note that, in the non-deterministic case (that is, there is possibly more than one computation for a configuration),  we can distinguish two flavours of soundness \cite{DNH84}:
\begin{description}
\item[soundness-must] (or simply soundness) no computation can be stuck
\item[soundness-may] at least one computation is not stuck
\end{description}
Soundness-must is the standard soundness in small-step semantics, and can be expressed in the $\Wrong$ 
extension 
as follows: 
\begin{description}
\item[soundness-must ($\Wrong$)]  If $\conf \in \Pred$, then $\notvalid{\WrRuleSet}{\eval{\conf}{\Wrong}}$
\end{description}
Instead, soundness-must \emph{cannot} be expressed in the trace extension. Indeed, stuck computations are not explicitly modelled. 
 Conversely,  soundness-may can be expressed in the trace extension as follows:
\begin{description}
\item[soundness-may (traces)] If $\conf \in \Pred$, then there is $t$ such that $\valid{\TrRuleSet}{\evaltr{\conf}{t}}$
\end{description}
\noindent  whereas cannot be expressed in  the $\Wrong$ semantics, since diverging computations are not modelled.

 Of course soundness-must and soundness-may coincide  in the deterministic case. 
Finally, note that indexes (e.g., the specific types of configurations) do not play any role in the above statements. However, they are relevant in the notion of \emph{strong soundness}, introduced by \citet{WrightFelleisen94}.
Strong soundness holds if, for configurations satisfying $\Pred_\idx$ (e.g., having a given type), computation cannot be stuck, and moreover, produces a result satisfying $\Pred_\idx$ (e.g., of the same type) if terminating. Note that soundness alone does not even guarantee to obtain a result satisfying $\Pred$ (e.g., a well-typed result). The three conditions introduced in the following section actually ensure strong soundness.

In 
\refToSect{sc} we provide sufficient conditions for soundness-must, showing that  they  actually ensure soundness  in  the $\Wrong$ semantics (\refToThm{sound-wrong}). Then, in \refToSect{smc}, we provide (weaker) sufficient conditions for soundness-may, and show that they actually ensure soundness-may  in  the trace semantics (\refToThm{sound-traces}).


\subsection{Conditions ensuring soundness-must} \label{sect:sc}

 The  three conditions which ensure the soundness-must property  are \emph{local preservation}, \emph{$\exists$-progress}, and $\forall$-progress. The names suggest that the former plays the role of the \emph{type preservation (subject reduction)} property, and the latter two of the \emph{progress} property in small-step semantics. However, as we will see, the correspondence is only rough, since the reasoning here is different. 

Considering the first condition more closely, we use the name \emph{preservation} rather than type preservation since, as already mentioned, the proof technique can be applied to arbitrary predicates. More importantly, \emph{local} means that the condition is \emph{on single rules} rather than on the semantic relation as a whole, as standard subject reduction. 
The same holds for the  other two conditions. 


\begin{definition}
[\refToSound{preservation}: Local Preservation] \label{sound:preservation}
For each $\rho {\equiv} \inlinerule{\judg_1\ldots\judg_n}{\judg_{n+1}}{\conf}$, if $\conf{\in}\Pred_\idx$, then there exist $\idx_1, \ldots, \idx_{n+1} \in \IdxSet$, with $\idx_{n+1} {=} \idx$,  such that, \mbox{for all $k \in 1..n+1$:}
\begin{quote}  
if, for all $h < k$, $\ResSet(\judg_h) \in \Pred_{\idx_{h}}$, 
then $\ConfSet(\judg_k) \in \Pred_{\idx_k}$.
\end{quote}
\end{definition}

 Thinking to the paradigmatic case where the indexes are types,  for each rule $\rho$,  if the configuration $\conf$ in the consequence has type $\idx$, 
we have to find types $\idx_1, \ldots, \idx_{n+1}$ which can be assigned to (the configurations in) the premises, in particular the same type as $\conf$ for the continuation. 
 More precisely,  we start finding type $\idx_1$, and successively find the type $\idx_k$ for (the configuration in) the $k$-th premise assuming that  the results of all the previous premises have the expected  types.   Indeed, if all such previous premises are derivable, then the expected type should be preserved by their results; if some premise is not derivable, the considered rule is ``useless''. For instance, considering (an instantiation of) meta-rule $\metainlinerule{app}{\eval{\e_1}{\LambdaExp{\x}{\e}}\ \eval{\e_2}{\val_2}}{\eval{\subst{\e}{\val_2}{\x}}{\val}}{\AppExp{\e_1}{\e_2}}$ in \refToSect{framework}, we prove that $\subst{\e}{\val_2}{\x}$ has the type $\T$ of $\AppExp{\e_1}{\e_2}$ under the assumption that $\LambdaExp{\x}{\e}$ has type $\funType{\tA'}{\tA}$, and $\val_2$ has type $\T'$ (see the proof example in \refToSect{simply-typed} for more details).\\
 A counter-example to condition \refToSound{preservation} is discussed at the beginning of \refToSect{iut}. 

The following lemma assures that local preservation actually implies \emph{preservation} of the semantic relation as a whole.  

\begin{lemma}[{Preservation}] \label{lem:sr}
 Let $\RuleSet$ and $\Pred$ satisfy condition {\em \refToSound{preservation}}. 
If $\valid{\RuleSet}{\eval{\conf}{\res}}$ and $\conf \in \Pred_\idx$, then $\res \in \Pred_\idx$.
\end{lemma}
\begin{proof}
The proof is by  a double induction. We denote by $RH$  and $IH$ the first and the second induction hypothesis, respectively. 
The first 
induction is on big-step rules.  
 Axioms  have conclusion $\eval{\res}{\res}$, hence the thesis holds since $\res \in \Pred_\idx$ by hypothesis. 
 Other rules have shape $\inlinerule{\judg_1\ldots\judg_n}{\judg_{n+1}}{\conf}$ with $\conf \in \Pred_\idx$. 
We prove by complete induction on $k \in 1..n+1$ that \mbox{$\ConfSet(\judg_k) \in \Pred_{\idx_k}$}, for all $k \in 1..n+1$ and for some $\idx_1, \ldots, \idx_{n+1} \in \IdxSet$.
By \refToSound{preservation}, there are $\idx_1, \ldots, \idx_{n+1} \in \IdxSet$ and $\ConfSet(\judg_1) \in \Pred_{\idx_1}$. 
For $k > 1$, by $IH$ we know that $\ConfSet(\judg_h) \in \Pred_{\idx_h}$, for all $h < k$. Then, by  $RH$,  
we get that $\ResSet(\judg_h) \in \Pred_{\idx_h}$.  Moreover, 
by \refToSound{preservation}, $\ConfSet(\judg_k) \in \Pred_{\idx_k}$, as needed. 
In particular, we have just proved that $\ConfSet(\judg_{n+1}) \in \Pred_{\idx_{n+1}}$ and, since by \refToSound{preservation} $\idx_{n+1} = \idx$, we get $\ConfSet(\judg_{n+1}) \in \Pred_\idx$. 
Then, by  $RH$,  
we  conclude  that $\res = \ResSet(\judg_{n+1}) \in \Pred_\idx$, as needed. 
\end{proof}
%

The following proposition  is  a form of local preservation  where indexes (e.g., specific types) are not relevant, simpler to use in the 
 proofs of  
Theorems \ref{thm:sound-wrong} \mbox{and 
 \ref{thm:sound-traces}}.  

\begin{proposition} \label{prop:preservation}
 Let $\RuleSet$ and $\Pred$ satisfy condition {\em \refToSound{preservation}}. 
 For each  
$\inlinerule{\judg_1\ldots\judg_n}{\judg_{n+1}}{\conf}$ and \mbox{$k{\in} 1..n+1$}, 
if $\conf \in \Pred$ and, for all $h < k$, $\valid{\RuleSet}{\judg_h}$, then $\ConfSet(\judg_k) \in \Pred$.
\end{proposition}
\begin{proof}
The proof is by complete induction on $k$. 
Assume the thesis for all $h < k$,
then, since by hypothesis we have $\valid{\RuleSet}{\judg_h}$ for all $ h < k$, we  get, by induction hypothesis,  $\ConfSet(\judg_h) \in \Pred$ for all $h < k$. 
By \refToLem{sr}, we also get $\ResSet(\judg_h) \in \Pred$, hence by condition \refToSound{preservation}, we get the thesis. 
\end{proof}

The second condition, named \emph{$\exists$-progress}, ensures that, for configurations satisfying the predicate $\Pred$  (e.g., well-typed),
we can \emph{start constructing} a {proof} tree. 

\begin{definition}
[\refToSound{progress-ex}: $\exists$-progress] \label{sound:progress-ex}
For each 
$\conf \in \Pred {\setminus} \ResSet$,  $\ConfSet(\rho) = \conf$ \mbox{for some rule $\rho$}.
\end{definition}

The third condition, named \emph{$\forall$-progress}, ensures that, for configurations satisfying 
$\Pred$, we can \emph{continue constructing} the proof tree. This condition uses the notion of rules \emph{equivalent up-to an index} introduced at the 
beginning of \refToSect{wrong}.

\begin{definition}[\refToSound{progress-all}: $\forall$-progress] \label{sound:progress-all} 
For each $\rho\equiv\inlinerule{\judg_1\ldots\judg_n}{\judg_{n+1}}{\conf}$, if  $\conf \in \Pred$, then, for each $k \in 1..n+1$:
\begin{quote}
if, for all $h < k$, $\valid{\RuleSet}{\judg_h}$ and $\valid{\RuleSet}{\eval{\ConfSet(\judg_k)}{\res}}$, for some $\res \in \ResSet$, 
then there is a rule $\rho' \sim_k \rho$ such that $\ResSet(\rho', k) = \res$. 
\end{quote}
\end{definition}

We have to check, for each rule $\rho$, the following: if the configuration $\conf$ in the consequence satisfies the predicate (e.g., is well-typed), then, for each $k$, if the configuration in premise $k$ evaluates to some result $\res$ (that is, $\valid{\RuleSet}{\eval{\ConfSet(\judg_k)}{\res}}$), then there is a rule ($\rho$ itself or another rule with the same  configuration in the consequence  and the first $k-1$ premises) with such judgment as $k$-th premise.  This check  can be done  under the assumption that all the previous premises are derivable. For instance, consider again (an instantiation of) the meta-rule $\metainlinerule{app}{\eval{\e_1}{\LambdaExp{\x}{\e}}\ \eval{\e_2}{\val_2}}{\eval{\subst{\e}{\val_2}{\x}}{\val}}{\AppExp{\e_1}{\e_2}}$. Assuming that $\e_1$ evaluates to some $\val_1$,  we have to check that there is a rule  with first premise  $\eval{\e_1}{\val_1}$, in pratice,  that $\val_1$ is a $\lambda$-abstraction. In general,  in the common case where for each configuration in the consequence there is only one applicable meta-rule, checking \refToSound{progress-all} amounts to show that results obtained in the premises satisfy the side conditions of such meta-rule, in particular have the required shape (see also the proof example in \refToSect{simply-typed}). 
If there is more than one applicable meta-rule, (sub)configurations in the premises should only evaluate to results which satisfy the side conditions of one of them, for an example see the proof of \refToSound{progress-all} in \refToThm{sdfjl}.
\label{progress-all}

\paragraph{Soundness-must in $\Wrong$ semantics}
Recall that $\WrRuleSet$ is the extension of $\RuleSet$ with $\Wrong$ (\refToSect{wrong}). 
We prove the claim of {soundness-must}  with respect to  $\WrRuleSet$. 

\begin{theorem} \label{thm:sound-wrong}
 Let $\RuleSet$ and $\Pred$ satisfy conditions {\em \refToSound{preservation}}, {\em \refToSound{progress-ex}} and {\em \refToSound{progress-all}}.   If $\conf \in \Pred$, then \mbox{{\em $\notvalid{\WrRuleSet}{\eval{\conf}{\Wrong}}$.}}
\end{theorem}
\begin{proof}
To prove the statement, we assume $\valid{\WrRuleSet}{\eval{\conf}{\Wrong}}$  and look for a contradiction. 
The proof is by induction on the derivation of $\eval{\conf}{\Wrong}$.\\  
 If the last applied rule is an axiom,  then,  by construction, there is no rule $\rho \in \RuleSet$ such that $\ConfSet(\rho) = \conf$, and this violates  condition \refToSound{progress-ex}, since $\conf \in \Pred$.\\ If the last applied rule is $\wrongrule{\rho}{i}{\res}$, with $\rho \equiv \inlinerule{\judg_1\ldots\judg_n}{\judg_{n+1}}{\conf}$,  then,  by hypothesis, for all $k < i$, $\valid{\WrRuleSet}{\judg_k}$, and $\valid{\WrRuleSet}{\eval{\ConfSet(\judg_i)}{\res}}$, and these judgments can also be derived in $\RuleSet$ by conservativity (\refToThm{eq-finite_1}). 
Furthermore, by construction of this rule, we know that there is no other rule $\rho' \sim_i \rho$ such that $\ResSet(\rho', i) = \res$, and this violates condition \refToSound{progress-all}, since $\conf \in \Pred$.\\ If the last applied rule is $\proprule{\rho}{i}{\Wrong}$, with $\rho \equiv \inlinerule{\judg_1\ldots\judg_n}{\judg_{n+1}}{\conf}$,  then,  by hypothesis, for all $k < i$, $\valid{\WrRuleSet}{\judg_k}$, 
 and these judgments 
can also be derived in $\RuleSet$ by conservativity. 
Then, by \refToProp{preservation} (which requires condition \refToSound{preservation}), since $\conf \in \Pred$, we have
$\ConfSet(\judg_i) \in \Pred$, hence we get the thesis by induction hypothesis. 
\end{proof}
\refToSect{simply-typed} ends with examples not satisfying properties \refToSound{progress-ex} and \refToSound{progress-all}.

\subsection{Conditions ensuring soundness-may}\label{sect:smc}
As discussed in \refToSect{must-may}, in the trace semantics we can only express a weaker form of soundness: at least one computation is not stuck (\emph{soundness-may}). As the reader can expect, to ensure this property weaker sufficient conditions are enough: namely, condition \refToSound{preservation}, and another condition 
 named \emph{progress-may}  and defined below. 

 We   write $\notvalid{\RuleSet}{\eval{\conf}{}}$ if $\conf$ \emph{does not converge} (there is no $\res$ such that $\valid{\RuleSet}{\eval{\conf}{\res}}$). 

\begin{definition}
[\refToSound{progress-may}: progress-may] \label{sound:progress-may}
For each $\conf \in \Pred{\setminus}\ResSet$, there is $\rho \equiv \inlinerule{\judg_1\ldots\judg_n}{\judg_{n+1}}{\conf}$ such that:
\begin{quote}
if there is a (first) $k \in 1..n+1$ 
such that $\notvalid{\RuleSet}{\judg_k}$ and, 
for all $h < k$, 
$\valid{\RuleSet}{\judg_h}$,
then $\notvalid{\RuleSet}{\eval{\ConfSet(\judg_k)}{}}$. 
\end{quote}
\end{definition}

This condition can be informally understood as follows: we have to show that there is an either finite or infinite computation for $\conf$. If we find a rule where all premises are derivable (no $k$), then there is a finite computation.  Otherwise, $\conf$ does not converge. In this case, we should find a rule where the configuration in the first non-derivable premise $k$ does not converge as well. Indeed, 
 by  coinductive reasoning (use of \refToLem{div-consistency} below), we obtain that $\conf$ diverges. 
The following proposition shows that this condition is indeed a weakening of \refToSound{progress-ex} and \refToSound{progress-all}.

\begin{proposition} \label{prop:progress-to-may}
Conditions {\em \refToSound{progress-ex}} and {\em \refToSound{progress-all}} imply condition {\em \refToSound{progress-may}}.
\end{proposition}
\begin{proof}
For each $\conf \in \ConfSet$, let us define $b_\conf  \in \N$ as $\max \{ \#\rho  \mid \ConfSet(\rho) = \conf \}$, which is finite by the boundedness condition,  see condition {\sf BP} at page \pageref{bp}.
For each rule $\rho$, with $\ConfSet(\rho) = \conf$,  let us denote by $nd(\rho)$ the index of the first premise of $\rho$ which is not derivable, if any, 
otherwise set  $nd(\rho) = b_\conf$. 
For each $\conf \in \Pred$, we first prove the following fact:\\ 
\centerline{
$(\star)$ for each rule $\rho$, with $\ConfSet(\rho) = \conf$, there exists a rule $\rho'$ such that}
\centerline{$\ConfSet(\rho') = \conf$,  $nd(\rho') \ge nd(\rho)$ and, if $nd(\rho') \le b_\conf$, then, for all $\res \in \ResSet$, $\notvalid{\RuleSet}{\eval{\ConfSet(\rho', nd(\rho'))}{\res}}$. 
}
Note that the requirement in $(\star)$ is the same as that of condition \refToSound{progress-may}. 
The proof is by complete induction on $h(\rho) = b_\conf + 1 - nd(\rho)$. 
If $h(\rho) = 0$, hence $nd(\rho) = b_\conf + 1$, then the thesis follows by taking $\rho' = \rho$. 
Otherwise, we have two cases: 
if there is no $\res \in \ResSet$ such that $\valid{\RuleSet}{\eval{\ConfSet(\rho, nd(\rho))}{\res}}$, then we have the thesis taking $\rho' = \rho$;
otherwise, by condition \refToSound{progress-all}, there is a rule $\rho'' \sim_{nd(\rho)} \rho$ such that $\ResSet(\rho'', nd(\rho)) = \res$, hence $nd(\rho'') > nd(\rho)$. 
Then, we have $h(\rho'') < h(\rho)$, hence we get the thesis by induction hypothesis.\\
Now, by condition \refToSound{progress-ex}, there is a rule $\rho$ with $\ConfSet(\rho) = \conf$, and applying $(\star)$ to $\rho$ we get \mbox{condition \refToSound{progress-may}.}
\end{proof}

\paragraph{Soundness-may in trace semantics}
Recall that $\TrRuleSet$ is the extension of $\RuleSet$ with traces, defined in \refToSect{traces}, where judgements have shape $\evaltr{\conf}{t}$, with $t \in \FIList{\ConfSet}$. 

The following lemma provides a proof principle useful  to coinductively show that a property ensures the existence of an infinite trace, in particular to show  \refToThm{sound-traces}.  It is a slight variation of an analogous principle presented in \citet{AnconaDZ@OOPSLA17}. 

\begin{lemma} \label{lem:div-consistency}
Let $\Spec \subseteq \ConfSet$ be a set. 
If, for all $\conf \in \Spec$, there are $\rho \equiv \inlinerule{\judg_1\ldots\judg_n}{\judg_{n+1}}{\conf}$ and $ k  \in 1..n+1$
such that 
\begin{enumerate}
\item for all  $h < k$, $\valid{\RuleSet}{\judg_h}$, and 
\item  $\ConfSet(\judg_k) \in \Spec$ 
\end{enumerate}
then, for all $\conf \in \Spec$, there is $t \in \InfList{\ConfSet}$ such that $\valid{\TrRuleSet}{\evaltr{\conf}{t}}$. 
\end{lemma}
\begin{proof}
First of all, for each $\conf \in \Spec$, we construct a trace $t_\conf \in \FIList{\ConfSet}$, which will be the candidate trace to prove the thesis. 
By hypothesis, there is a rule $\rho_\conf \equiv \inlinerule{\judg_1^\conf\ldots\judg_{n_\conf}^\conf}{\judg_{n_\conf+1}^\conf}{\conf}$ and an index $i_\conf \in 1..n_\conf+1$ such that, for all $k < i_\conf$, we have $\valid{\RuleSet}{\judg_k^\conf}$. 
Therefore, by  \refToThm{eq-finite_1}, there are finite traces $t_1^\conf, \ldots, t_{i_\conf}^\conf \in \List{\ConfSet}$ such that,  for all $k < i_\conf$, we have $\valid{\TrRuleSet}{\evaltr{\ConfSet(\judg_k^\conf)}{t_k^\conf \cdot \ResSet(\judg_k^\conf)}}$,
and, in addition, we know that $\ConfSet(\judg_{i_\conf}^\conf) \in \Spec$. 
Then, for each $\conf \in \Spec$, we can introduce a variable $X_\conf$ and define an equation $X_\conf = \conf \cdot t_1^\conf \cdot \cdots \cdot t_{i_\conf - 1}^\conf \cdot X_{\ConfSet(\judg_{i_\conf}^\conf)}$. 
The set of all such equations is a guarded system of equations, which thus has a unique solution function $\fun{s}{\Spec}{\InfList{\ConfSet}}$, that is, for each $\conf \in \Spec$ we have $s(\conf) = \conf \cdot t_1^\conf \cdot \cdots \cdot t_{i_\conf -1}^\conf \cdot s(\ConfSet(\judg_{i_\conf}^\conf))$. 

We now have to prove that, for all $\conf \in \Spec$, we have $\valid{\TrRuleSet}{\evaltr{\conf}{s(\conf)}}$. 
To this end, 
consider the set $\Spec' = \{\Pair{\conf}{s(\conf)} \mid \conf \in \Spec \} \cup \{\Pair{\conf}{t\cdot \res} \mid \valid{\TrRuleSet}{\evaltr{\conf}{t\cdot \res}} \}$, 
then the proof is by coinduction. 
Let $\Pair{\conf}{t} \in \Spec'$, then we have to find a rule $\Rule{\judg_1\Space\ldots\Space\judg_n}{\evaltr{\conf}{t}} \in \TrRuleSet$ such that, for all $k \in 1..n$, $\Pair{\ConfSet(\judg_k)}{\Trace(\judg_k)} \in \Spec'$. 
We have two cases:
\begin{itemize}
\item if $t = s(\conf)$, then the needed rule is $\divtracerule{\rho_\conf}{i_\conf}{t_1^\conf,\,\ldots,\,t_{i_\conf-1}^\conf}{s(\ConfSet(\judg_{i_\conf}^\conf))}$, and 
\item if $t = t' \cdot \res$ is finite and $\valid{\TrRuleSet}{\evaltr{\conf}{t}}$, then $\evaltr{\conf}{t}$ is the consequence of a trace introduction rule, where all premises are derivable.
\end{itemize}
\end{proof}

We end this section with the proof of soundness-may for the trace semantics.

\begin{theorem} \label{thm:sound-traces}  Let $\RuleSet$ and $\Pred$ satisfy conditions {\em \refToSound{preservation}} and {\em \refToSound{progress-may}}.  
If $\conf \in \Pred$, then there is $t$ such that $\valid{\TrRuleSet}{\evaltr{\conf}{t}}$.
\end{theorem}
\begin{proof}
 First note that, thanks to  \refToThm{eq-finite_0}, the statement is equivalent to the following:\\ 
\centerline{
 If $\conf \in \Pred$ and $\notvalid{\RuleSet}{\eval{\conf}{}}$, then there is $t \in \InfList{\ConfSet}$ such that $\valid{\TrRuleSet}{\evaltr{\conf}{t}}$. 
}
Then, the proof follows from \refToLem{div-consistency}. 
We define  $\Spec = \{ \conf \mid \conf {\in} \Pred\ \mbox{and}\  \notvalid{\RuleSet}{\eval{\conf}{}} \}$,  and show that, for all $\conf \in \Spec$, there are $\rho \equiv \inlinerule{\judg_1\ldots\judg_n}{\judg_{n+1}}{\conf}$ and $k \in 1..n+1$ such that, for all $h < k$, $\valid{\RuleSet}{\judg_h}$, and $\ConfSet(\judg_k) \in \Spec$. 

Consider $\conf \in \Spec$, then, by \refToSound{progress-may}, there is $\rho \equiv \inlinerule{\judg_1\ldots\judg_n}{\judg_{n+1}}{\conf}$. 
By definition of $\Spec$, we have $\notvalid{\RuleSet}{\eval{\conf}{}}$, 
hence there exists a  (first)  $k \in 1..n+1$ such that $\notvalid{\RuleSet}{\judg_k}$, 
since, otherwise, we would have $\valid{\RuleSet}{\eval{\conf}{\ResSet(\judg_{n+1})}}$. 
Then, since $k$ is the  first  index with such property, for all $h < k$, we have $\valid{\RuleSet}{\judg_h}$, hence, again by condition \refToSound{progress-may}, we have that  $\notvalid{\RuleSet}{\eval{\ConfSet(\judg_k)}{}}$. 
Finally, since for all $h < k$  we have $\valid{\RuleSet}{\judg_h}$, by \refToProp{preservation} we get  $\ConfSet(\judg_k) \in \Pred$,  hence $\ConfSet(\judg_k) \in \Spec$, as needed. 
\end{proof}


\section{Examples}\label{sect:examples}
\refToSect{simply-typed} explains in detail how a typical soundness proof can be rephrased in terms of our technique,  by reasoning directly  on big-step rules. 
\refToSect{fjl} shows a case where  this  is advantageous, since the property to be checked is \emph{not preserved} by intermediate computation steps, whereas it holds for the final result.
\refToSect{iut}  considers  a more sophisticated type system, with intersection and union types. \refToSect{fjos} shows another example where subject reduction is  not preserved,  whereas soundness can be proved  with our technique. This example is intended as a preliminary step towards a more challenging case. Finally, \refToSect{ifj} shows how our approach can also easily deal with memory.

For reader's convenience, we provide the reduction rules also in inline format, where the dummy continuation $\eval\res\res$, if any, is made explicit. 

\subsection{Simply-typed $\lambda$-calculus with recursive types}\label{sect:simply-typed}
As  a first example, we take the $\lambda$-calculus with natural constants, successor,  and choice used in \refToSect{framework} (\refToFigure{example}). 
We consider a standard simply-typed version with recursive types, obtained by interpreting the production in \refToFigure{lambda-typesystem}
coinductively.  Introducing  recursive types makes the calculus non-normalising and  permits  
to write interesting programs such as $\Omega$ (see \refToSect{traces}).

The typing rules are recalled in \refToFigure{lambda-typesystem}. Type environments, written $\Gamma$, are finite maps from variables to types, and $\SubstFun{\Gamma}{\T}{\x}$ denotes the map which returns $\T$ on $\x$  and  coincides with $\Gamma$ elsewhere.  We write $\HasType{\es}{\e}{\T}$ for $\HasType{\emptyset}{\e}{\T}$.

\begin{figure}[h]
\begin{center}
\begin{small}
\begin{grammatica}
\produzione{\T}{\natType\mid\funType{\T_1}{\T_2}}{type}
\end{grammatica}
\HSep
\begin{math}
\begin{array}{c}
\MetaRule
{t-var}
{}
{\HasType{\Gamma}{\x}{\T}}{\Gamma(\x)=\T}
\BigSpace
\MetaRule{t-const}{}{\HasType{\Gamma}{\const}{\natType}}{}
\\[4ex]
\MetaRule
{t-abs}
{\HasType{\SubstFun{\Gamma}{\T'}{\x}}{\e}{\T}}
{\HasType{\Gamma}{\LambdaExp{\x}{\e}}{\funType{\T'}{\T}}}
\BigSpace
\MetaRule
{t-app}
{\HasType{\Gamma}{\e_1}{\funType{\T'}{\T}}\Space\HasType{\Gamma}{\e_2}{\T'}}
{\HasType{\Gamma}{e_1\ e_2}{\T}}
\\[4ex]
\MetaRule
{t-succ}{\HasType{\Gamma}{\e}{\natType}}
{\HasType{\Gamma}{\SuccExp{\e}}{\natType}}{}
\BigSpace
\MetaRule{t-choice}{\HasType{\Gamma}{\e_1}{\T}\Space\HasType{\Gamma}{\e_2}{\T}}{\HasType{\Gamma}{\Choice{\e_1}{\e_2}}{\T}}{}
\end{array}
\end{math}
\end{small}
\end{center}
\caption{$\lambda$-calculus: type system}\label{fig:lambda-typesystem}
\end{figure}

 Let  $\RuleSet_1$  be  the big-step semantics defined in \refToFigure{example}, and let 
 $\Pred1_T(\e)$ hold if $\HasType{\es}{\e}{\T}$, for $\T$ defined in \refToFigure{lambda-typesystem}. 
  To prove the three conditions \refToSound{preservation}, \refToSound{progress-ex} and \refToSound{progress-all} of \refToSect{sc}, we need lemmas of inversion, substitution and canonical  forms, as in the standard technique.
\begin{lemma}[Inversion]\label{lem:ilr}
\begin{enumerate}
\item \label{lem:ilr:1}If $\HasType{\Gamma} \x\tA$, then $\Gamma(\x)= \tA$.
\item \label{lem:ilr:2} If $\HasType{\Gamma}{\const}\T$, then $\T=\natType$.
\item \label{lem:ilr:3} If $\HasType\Gamma{\LambdaExp{\x}{\e}}\tA$, then $\tA=\funType{\tA_1}{\tA_2}$ and  $\HasType{\SubstFun\Gamma{\tA_1}\x} \e {\tA_2}$.
\item \label{lem:ilr:4} If $\HasType\Gamma{\e_1\appop\e_2}\tA$, then $\HasType\Gamma{\e_1} {\funType{\tA'}{\tA}}$,  and $\HasType\Gamma{\e_2} {\tA'}$.
\item \label{lem:ilr:5} If $\HasType{\Gamma}{\SuccExp{\e}}{\T}$, then $\T=\natType$ and $\HasType{\Gamma}{\e}{\natType}$.
\item \label{lem:ilr:6} If $\HasType\Gamma {\e_1\oplus\e_2}\tA$, then $\HasType\Gamma{\e_i}\tA$ with $i\in1,2$.
\end{enumerate}
\end{lemma}

\begin{lemma}[Substitution]\label{lem:s}
If $\HasType{\SubstFun\Gamma{\tA'}\x} \e {\tA}$ and $\HasType\Gamma{\e'} {\tA'}$, then \mbox{$\HasType\Gamma{\subst\e{\e'}\x} {\tA}$}.
\end{lemma}

\begin{lemma}[Canonical Forms]\label{lem:cf}\
\begin{enumerate}
\item \label{lem:cf:1}
If $\HasType{\es} \val {\funType{\tA'}{\tA}}$, then $\val=\LambdaExp{\x}{\e}$.
\item \label{lem:cf:2}
If $\HasType{\es} \val \natType$, then $\val=\const$.
\end{enumerate}
\end{lemma}

\begin{theorem}[Soundness]\label{thm:sd}
The big-step semantics $\RuleSet_1$ and the indexed predicate $\Pred1$ satisfy the conditions {\em \refToSound{preservation}}, {\em \refToSound{progress-ex}} and {\em \refToSound{progress-all}} of \refToSect{sc}.
 \end{theorem}
 
Since the aim of this first example is to illustrate the proof technique, we provide a proof where we explain the reasoning  in detail.\\[1ex] 
{\em Proof of {\em \refToSound{preservation}}.}
We should prove this condition for each \mbox{(instantiation of meta-)rule.}\\
\rn{app}: Assume that $\HasType{\es}{\AppExp{\e_1}{\e_2}}{\T}$ holds.  We have to find types for the premises, notably $\T$ for the last one. We  proceed as follows:
\begin{enumerate}
\item First premise: by \refToLemItem{ilr}{4}, $\HasType{\es}{\e_1} {\funType{\tA'}{\tA}}$.
\item Second premise:  again by \refToLemItem{ilr}{4}, $\HasType{\es}{\e_2}{\T'}$  (without needing the assumption $\HasType{\es}{\LambdaExp{\x}{\e}}{\funType{\tA'}{\tA}}$).   
\item Third premise:  $\HasType{\es}{\subst{\e}{\val_2}{\x}}{\T}$ should hold (assuming $\HasType{\es}{\LambdaExp{\x}{\e}}{\funType{\tA'}{\tA}}$, $\HasType{\es}{\val_2}{\T'}$).  Since $\HasType{\es}{\LambdaExp{\x}{\e}}{\funType{\tA'}{\tA}}$, by \refToLemItem{ilr}{3} we have $\HasType{\x{:}\T'} \e {\tA}$, so by \refToLem{s} and $\HasType{\es}{\val_2}{\T'}$ we have $\HasType{\es}{\subst\e{\val_2}\x} {\tA}$. 
\end{enumerate}
\rn{succ}: This rule has an implicit continuation $\eval{n+1}{n+1}$.  Assume that $\HasType{\es}{\SuccExp{\e}}{\T}$ holds. By \refToLemItem{ilr}{5}, $\T=\natType$, and $\HasType{\es}{\e}{\natType}$, hence we find $\natType$ as type for the first premise. Moreover, $\HasType{\es}{n+1}{\natType}$ holds by rule \rn{t-const}.\\ 
\rn{choice}: Assume that $\HasType{\es}{\Choice{\e_1}{\e_2}}{\T}$ holds. By \refToLemItem{ilr}{6}, we have $\HasType{\es}{\e_i}{\T}$, with $i\in 1,2$. Hence we  find  $\T$ as type for the  premise.\\[1ex]  
{\em Proof of {\em \refToSound{progress-ex}}.}
We should prove that, for each non-result configuration (here, expression $\e$ which is not a value) such that $\HasType{\es}{\e}{\T}$ holds for some $\T$, there is a rule with this configuration in the consequence.  The expression $\e$ cannot be a variable, since a variable cannot be typed in the empty environment. 
Application, successor and choice appear as consequence in the \mbox{reduction rules.}\\[1ex] 
{\em Proof of \em \refToSound{progress-all}.}
 We should prove this condition for each \mbox{(instantiation of meta-)rule.}\\
\rn{app}: Assuming $\HasType{\es}{\AppExp{\e_1}{\e_2}}{\T}$, again  by \refToLemItem{ilr}{4}  we get  $\HasType\Gamma{\e_1} {\funType{\tA'}{\tA}}$. 
\begin{enumerate}
\item First premise: if $\eval{\e1}{\val}$ is derivable, then there should be a rule with $\AppExp{\e_1}{\e_2}$ in the consequence and $\eval{\e1}{\val}$ as first premise. Since we proved \refToSound{preservation}, by preservation (\refToLem{sr}) 
$\HasType{\es}{\val}{\funType{\tA'}{\tA}}$ holds. Then, by \refToLemItem{cf}{1}, $\val$ has shape $\LambdaExp{\x}{\e}$, hence the required rule exists. As noted at page \pageref{progress-all}, in practice checking \refToSound{progress-all} for a (meta-)rule amounts to show that (sub)configurations in the premises  only evaluate to results which satisfy the side conditions, in this case to have the required shape (to be a $\lambda$-abstraction).
\item Second premise: if $\eval{\e_1}{\LambdaExp{\x}{\e}}$, and $\eval{\e2}{\val_2}$, then there should be a rule  with $\AppExp{\e_1}{\e_2}$ in the consequence and $\eval{\e_1}{\LambdaExp{\x}{\e}}$, $\eval{\e2}{\val}$ as first two premises. This is trivial since the meta-variable $\val_2$ can be freely instantiated in the meta-rule. 
\end{enumerate}
\rn{succ}: Assuming $\HasType{\es}{\SuccExp{\e}}{\T}$, again by \refToLemItem{ilr}{5} we get $\HasType{\es}{\e}{\natType}$. If $\eval{\e}{\val}$ is derivable, there should be a rule with $\SuccExp{\e}$ in the consequence and $\eval{\e}{\val}$ as first premise. Indeed, by preservation (\refToLem{sr}) and \refToLemItem{cf}{2}, $\val$ has shape $n$. For the second premise, if $\eval{n+1}{\val}$ is derivable, then $\val$ is necessarily $n+1$.\\
\rn{choice}: Trivial since the meta-variable $\val$ can be freely instantiated.

\medskip

An interesting remark is that, differently from the standard approach, there is  \emph{no induction}  in the proof: everything is \emph{by cases}. This is a consequence of the fact that, as discussed in \refToSect{sc}, the three conditions are \emph{local}, that is, they are conditions on single rules. Induction is ``hidden'' in the proof that  those  three conditions  are sufficient to  ensure soundness.

 If we drop in \refToFigure{example} rule \rn{succ}, then condition \refToSound{progress-ex} fails, since  there is no longer a rule for the well-typed non-result configuration $\SuccExp{\const}$. If we add the \rn{fool} rule $\HasType {}{\AppExp{0}0}{\natType}$, then condition \refToSound{progress-all} fails for rule \rn{app}, since $\eval{0}{0}$ is derivable, but there is no rule with $\AppExp{0}{0}$ in the conclusion and $\eval{0}{0}$ as first premise.


\subsection{$\MiniFJLambda$}\label{sect:fjl}
In this example, the language is a subset of $\FJLambda$ \cite{BettiniBDGV18}, a calculus extending Featherweight Java ($\FJ$) with {$\lambda$-abstraction}s and intersection types, introduced in Java 8.  To keep the example small, we do not consider intersections and focus  on one key typing feature: {$\lambda$-abstraction}s can only be typed when occurring in a context requiring a given type (called  the  \emph{target type}). In a small-step semantics, this poses a problem: reduction can move {$\lambda$-abstraction}s into arbitrary contexts, leading to intermediate terms which would be ill-typed. To  maintain  subject reduction, in \cite{BettiniBDGV18} $\lambda$-abstractions are decorated with their initial target type. In a big-step semantics, there is no need of intermediate terms and annotations.

The syntax is given in the first part of \refToFigure{FJ-lambda-syntax}. 
We assume sets of \emph{variables} $\x$, \emph{class names} $\CC$, \emph{interface names} $\II$, {$\JJ$,} \emph{field names} $\f$, and \emph{method names} $\m$. 
Interfaces which have \emph{exactly} one method (dubbed {\em functional interfaces}) can be used as target 
types.
Expressions are those of $\FJ$, plus {$\lambda$-abstraction}s, and types are class and interface names. In $\LambdaExp{\xs}{\e}$ we assume that $\xs$ is not empty and  $\e$ {is not a $\lambda$-}abstraction.
For simplicity, we only consider \emph{upcasts}, which have no runtime effect, but are important to allow the programmer to use {$\lambda$-abstraction}s,
as exemplified in discussing typing rules.

To be concise, the class table is abstractly modelled  as follows: 
\begin{itemize}
\item $\fields{\CC}$ gives the sequence of  field declarations $\Field{\T_1}{\f_1}..\Field{\T_n}{\f_n}$ for class $\CC$  
\item $\mtype{\T}{\m}$ gives, for each method $\m$ in class or interface $\T$, the pair $\funtype{\T_1\ldots\T_n}{\T'}$ 
consisting of the parameter types and return type
\item $\mbody{\CC}{\m}$ gives, for each method $\m$ in class $\CC$, the pair %
$\Pair{\x_1\ldots\x_n}{\e}$ 
consisting of the parameters and body
\item $\leqfj$ is the reflexive and transitive closure of the union of the $\aux{extends}$ and $\aux{implements}$ relations
\item $\umtype{\II}$ gives, for each \emph{functional} interface $\II$, $\mtype{\II}{\m}$,  where $\m$ is the only method of $\II$. 
\end{itemize}

The big-step semantics is given in the last part of \refToFigure{FJ-lambda-big-step}. $\MiniFJLambda$ shows an example of instantiation of the framework where configurations include an auxiliary structure, rather than being just language terms.  In this case, the structure is an \emph{environment} $\env$ {(a finite map from variables to values)} modelling the current stack frame. Results are values, which are either \emph{objects}, of shape $\object{\CC}{\vals}$, or {$\lambda$-abstraction}s.

\begin{figure}[h]
\begin{small}
\begin{grammatica}
\produzione{\e}{\x\mid\FieldAccess{\e}{\f}\mid\ConstrCall{\CC}{\e_1,\ldots,\e_n}\mid
\MethCall{\e}{\m}{\e_1,\ldots,\e_n}\mid\LambdaExp{\xs}{\e}\mid\Cast{\T}{\e}}{\qquad expression}\\
\produzione{\xs}{\x_1\ldots\x_n}{\qquad variable list}\\
\produzione{\T}{\CC\mid\II}{\qquad type}
\end{grammatica}
\HSep
\begin{grammatica}
\produzione{\conf}{\Conf{\env}{\e}\mid\val}{\qquad configuration}\\
\produzione{\val}{\object{\CC}{\vals}\mid\LambdaExp{\xs}{\e}
}{\qquad result (value)}\\
\produzione{\vals}{\val_1,\ldots,\val_n}{\qquad value list}
\end{grammatica}

\HSep
\begin{math}
\begin{array}{l}
\MetaRule{var}{}{\eval{\Conf{\env}{\x}}{\val}}{\env(\x)=\val}
\\[4ex]
\MetaRule{field-access}{\eval{\Conf{\env}{\e}}{\object{\CC}{\val_1,\ldots,\val_n}}}{
\eval{\Conf{\env}{\FieldAccess{\e}{\f_i}}}{\val_i}}
{\begin{array}{l}
\fields{\CC}=\Field{\T_1}{\f_1}\ldots\Field{\T_n}{\f_n}\\
i\in 1..n\\
\end{array}
}\\[4ex]
\MetaRule{new}{
\eval{\Conf{\env}{\e_i}}{\val_i}\Space \forall i\in 1..n}{\eval{\Conf{\env}{\ConstrCall{\CC}{\e_1,\ldots,\e_n}}}{\object{\CC}{\val_1,\ldots,\val_n}}}
{
}\\[4ex]
\MetaRule{invk}{
\begin{array}{l}
\eval{\Conf{\env}{\e_0}}{\object{\CC}{{\vals}}}\\
\eval{\Conf{\env}{\e_i}}{\val_i}\Space \forall i\in 1..n\\
\eval{\Conf{\x_1{:}\val_1,\ldots,\x_n{:}\val_n,\this{:}\object{\CC}{\vals}}{\e}}{\val}
\end{array}}{
\eval{\Conf{\env}{\MethCall{\e_0}{\m}{\e_1,\ldots,\e_n}}}
{\val}}
{\begin{array}{l}
\mbody{\CC}{\m}=\Pair{\x_1\ldots\x_n}{\e}\\\
\end{array}
}\\[4ex]
\MetaRule{$\lambda$-invk}{
\begin{array}{l}
\eval{\Conf{\env}{\e_0}}{\LambdaExp{\xs}{\e}}\\
\eval{\Conf{\env}{\e_i}}{\val_i}\Space \forall i\in 1..n\\
\eval{\Conf{
\x_1{:}\val_1,\ldots,\x_n{:}\val_n}{\e}}{\val}
\end{array}
}{
\eval{\Conf{\env}{\MethCall{\e_0}{\m}{\e_1,\ldots,\e_n}}}{\val}}
{\begin{array}{l}
\end{array}
}
\BigSpace
\MetaRule{upcast}{
\eval{\Conf{\env}{\e}}{\val}
}{\eval{\Conf{\env}{\Cast{\T}{\e}}}{\val}}
{}
\end{array}
\end{math}

\HSep
\begin{math}
\begin{array}{l}
 \metainlinerule{var}
{\epsilon}
{\eval{\val}\val}
{\Conf{\env}{\x}}\BigSpace\env(\x)=\val\\
\metainlinerule{field-access}
{\eval{\Conf{\env}{\e}}
{\object{\CC}{\val_1,\ldots,\val_n}}}
{\eval{\val_i}{\val_i}}
{\Conf{\env}{\FieldAccess{\e}{\f_i}}}\qquad
\hfill \fields{\CC}=\Field{\T_1}{\f_1}\ldots\Field{\T_n}{\f_n}
\Space
i\in 1..n\\
\metainlinerule{new}
{\eval{\Conf{\env}{\e s}}{\vals}}
{\eval{\object{\CC}{\val_1,\ldots,\val_n}}{\object{\CC}{\val_1,\ldots,\val_n}}}
{\ConstrCall{\CC}{\e_1,\ldots,\e_n}}\\
\metainlinerule{invk}{\eval{\Conf{\env}{\e_0}}{\object{\CC}{\vals'}}, \eval{\Conf{\env}{\e s}}{\vals}}{\eval{\Conf{\env'}{\e}}{\val}}{\Conf{\env}{\MethCall{\e_0}{\m}{\e_1,\ldots,\e_n}}}\\
\hfill \env'=\x_1{:}\val_1,\ldots,\x_n{:}\val_n,\this{:}\object{\CC}{\vals}\qquad\mbody{\CC}{\m}=\Pair{\x_1\ldots\x_n}{\e}\\
\metainlinerule{$\lambda$-invk}
{\eval{\Conf{\env}{\e_0}}{\LambdaExp{\xs}{\e}}, \eval{\Conf{\env}{\e s}}{\vals}}
{\eval{\Conf{\env'}{\e}}{\val}}
{\Conf{\env}{\MethCall{\e_0}{\m}{\e_1,\ldots,\e_n}}}\\
\hfill\env'=\x_1{:}\val_1,\ldots,\x_n{:}\val_n\\
\metainlinerule{upcast}{\epsilon}{ \eval{\Conf{\env}{\e}}{\val} }{\Cast{\T}{\e}}\\
\text{where }\eval{\Conf{\env}{ \e s}}{\vals}\text{ is short for }\eval{\Conf{\env}{\e_1}}{\val_1}, \ldots, \eval{\Conf{\env}{\e_n}}{\val_n}
\end{array}
\end{math}
\end{small}
\caption{$\MiniFJLambda$: syntax  and  big-step semantics}\label{fig:FJ-lambda-big-step}\label{fig:FJ-lambda-syntax}
\end{figure}
Rules for $\FJ$ constructs are straightforward. Note that, since we only consider upcasts, casts have no runtime effect. Indeed, they are guaranteed to succeed on well-typed expressions.
Rule \rn{$\lambda$-invk} shows that, when the receiver of a method is a {$\lambda$-abstraction}, the method name is not significant at runtime, and the effect is that the body of the function is evaluated as in  the  usual application.

The type system is given in \refToFigure{FJ-lambda-typesystem}. 
Method bodies are expected to be well-typed with respect to method types. Formally,
 $\mbody{\CC}{\m}$ and $\mtype{\CC}{\m}$ are either both defined or both undefined: in the first case 
$\mbody{\CC}{\m}=\Pair{\x_1\ldots\x_n}{\e}$,  
$\mtype{\CC}{\m}=\funtype{\T_1\ldots\T_n}{\T}$, and $\HasType{\x_1{:}\T_1,\ldots,\x_n{:}\T_n,\this{:}\CC}{\e}{\T}$. 
Moreover, we assume other standard $\FJ$ constraints on the class table, such as no field hiding, no method overloading, the same parameter  and return  types 
in overriding.
\begin{figure}
\begin{small}
\begin{math}
\begin{array}{l}
\MetaRule{t-conf}{\HasType{\es}{\val_i}{\T_i}\Space \forall i\in 1..n\Space\Space \HasType{\x_1{:}\T'_1,\ldots,\x_n{:}\T'_n}{\e}{\T}}{\HasType{\es}{\Conf{\x_1{:}\val_1,\ldots,\x_n{:}\val_n}{\e}}{\T}}
{\T_i\leqfj\T'_i\Space \forall i\in 1..n }
\\[4ex]
\MetaRule{t-var}{}{\HasType{\Gamma}{\x}{\T}}{\Gamma(\x)=\T}
\BigSpace
\MetaRule{t-field-access}{\HasType{\Gamma}{\e}{\CC}}{\HasType{\Gamma}{\FieldAccess{\e}{\f}}{\T_i}}
{\begin{array}{l}
\fields{\CC}=\Field{\T_1}{\f_1}\ldots\Field{\T_n}{\f_n}\\
i\in 1..n
\end{array}
}\\[4ex]
\MetaRule{t-new}{\HasType{\Gamma}{\e_i}{\T_i}\Space \forall i\in 1..n}{\HasType{\Gamma}{\ConstrCall{\CC}{\e_1,\ldots,\e_n}}{\CC}}
{\begin{array}{l}
\fields{\CC}=\Field{\T_1}{\f_1}\ldots\Field{\T_n}{\f_n}\\
\end{array}
}
\\[4ex]
\MetaRule{t-invk}{\HasType{\Gamma}{\e_i}{\T_i}\Space \forall i\in 0..n}{\HasType{\Gamma}{\MethCall{\e_0}{\m}{\e_1,\ldots,\e_n}}{\T}}
{\begin{array}{l}
\e_0\ \mbox{not of shape}\ \LambdaExp{\xs}{\e}\\
\mtype{\T_0}{\m}=\funtype{\T_1\ldots\T_n}{\T}\\\
\end{array}
}
\\[4ex]
\MetaRule{t-$\lambda$}{\HasType{
\x_1{:}\T_1,\ldots,\x_n{:}\T_n}{\e}{\T}}{\HasType{\Gamma}{\LambdaExp{\xs}{\e}}{\II}}
{\begin{array}{l}
\umtype{\II}=\funtype{\T_1\ldots\T_n}{\T}
\end{array}
}
\\[4ex]
\MetaRule{t-upcast}{\HasType{\Gamma}{\e}{\T}}{\HasType{\Gamma}{\Cast{\T}{\e}}{\T}}{
}\BigSpace
\MetaRule{t-object}{\HasType{\Gamma}{\val_i}{\T'_i}\Space\forall i\in 1..n}{\HasType{\Gamma}{\object{\CC}{\val_1,\ldots,\val_n}}{\CC}{}}
{\begin{array}{l}
\fields{\CC}=\Field{\T_1}{\f_1}\ldots\Field{\T_n}{\f_n}\\
\T'_i\leqfj\T_i\Space \forall i\in 1..n
\end{array}}
\\[4ex]
\MetaRule{t-sub}{\HasType{\Gamma}{\e}{\T}}{\HasType{\Gamma}{\e}{\T'}}{
\begin{array}{l}
\e\ \mbox{not of shape}\ \LambdaExp{\xs}{\e}\\
\T\leqfj\T'
\end{array}}
\end{array}
\end{math}
\end{small}
\caption{$\MiniFJLambda$: type system}\label{fig:FJ-lambda-typesystem}
\end{figure}

Besides the standard typing features of $\FJ$, the $\MiniFJLambda$ type system ensures the following. 
\begin{itemize}
\item A functional interface $\II$ can be assigned as type to a {$\lambda$-abstraction} which has the 
functional type of the method,
see rule \rn{t-$\lambda$}.
\item A {$\lambda$-abstraction} should have a \emph{target type} determined by the context where the $\lambda$-abstraction occurs. 
More precisely, see \cite{GoslingEtAl14} page 602, a $\lambda$-abstraction in our calculus can only occur as  return expression of a method or argument of constructor, method call or cast.  
Then, in some contexts a $\lambda$-abstraction cannot be typed,
in our calculus when  occurring as receiver in field access or method invocation, hence these cases should be prevented.
This is implicit in rule \rn{t-field-access}, since the type of the receiver should be a class name, whereas it is explicitly forbidden in rule \rn{t-invk}. For the same reason, a {$\lambda$-abstraction cannot be the main expression to be evaluated.}
\item A $\lambda$-abstraction with a given target type $\JJ$ should have type \emph{exactly} $\JJ$: a subtype $\II$ of $\JJ$ is not enough. Consider, for instance, the following program:
\begin{small}
\begin{verbatim}
interface J {}
interface I extends J { A m(A x); }
class C {
  C m(I y) { return new C().n(y); }
  C n(J y) { return new C(); }
}
\end{verbatim}
\end{small}
\noindent and  the  main expression $\MethCall{\ConstrCall{\CC}{}}{\n}{\lambda\x.\x}$. Here, the $\lambda$-abstraction has target type $\JJ$, which is \emph{not} a functional interface, hence the expression is ill-typed in Java (the compiler has no functional type against which to typecheck the {$\lambda$-abstraction}). On the other hand, in the body of method $\m$, the parameter $y$ of type $\II$ can be passed, as usual, to method $\n$ expecting a supertype. For instance, the main expression $\MethCall{\ConstrCall{\CC}{}}{\m}{\lambda\x.\x}$ is well-typed, since the {$\lambda$-abstraction} has target type $\II$, and can be safely passed to method $\n$, since it is not used as function there. To formalise this behaviour, it is forbidden to apply subsumption to {$\lambda$-abstraction}s, see rule \rn{t-sub}.
\item However, $\lambda$-abstractions occurring as results rather than  in source code (that is, in the environment and as fields of objects) are allowed to have a subtype of the required type, see the explicit side condition in rules \rn{t-conf} and \rn{t-object}. For instance, if $\CC$ is a class with one field $\JJ\, \f$, the expression 
$\ConstrCall{\CC}{\Cast{\II}{\lambda x.x}}$ is well-typed, whereas $\ConstrCall{\CC}{\lambda x.x}$ is ill typed, since rule \rn{t-sub} cannot be applied to {$\lambda$-abstraction}s. When the expression is evaluated, the result is $\object{\CC}{\lambda x.x}$, which is well-typed.  
\end{itemize}
As mentioned at the beginning, the obvious small-step semantics would produce not typable expressions. {In the above example, } we get\\ 
\centerline{$\ConstrCall{\CC}{\Cast{\II}{\lambda x.x}}\longrightarrow \ConstrCall{\CC}{\lambda x.x}\longrightarrow\object{\CC}{\lambda x.x}$}
and $\ConstrCall{\CC}{\lambda x.x}$ has no type, while $\ConstrCall{\CC}{\Cast{\II}{\lambda x.x}}$ and $\object{\CC}{\lambda x.x}$ have type $\CC$.

\bigskip

As expected to show soundness (\refToThm{sdfjl}) lemmas of inversion and canonical forms are handy: they can be easily proved as usual. Instead we do not need a substitution lemma, since environments associate variables to values. We write $\HasType{\Gamma}{\e}{\leqfj\T}$ as short for $\HasType{\Gamma}{\e}{\T'}$ and $\T'\leqfj\T$ for some $\T'$.  

\begin{lemma}[Inversion]\label{lem:il}
\begin{enumerate}
\item\label{lem:il:1} If $\HasType{\es}{\Conf{\x_1{:}\val_1,\ldots,\x_n{:}\val_n}{\e}}{\T}$, then $\HasType{\es}{\val_i}{\leqfj\T_i}$ for all $i\in 1..n$ and $\HasType{\x_1{:}\T_1,\ldots,\x_n{:}\T_n}{\e}{\T}$.
\item\label{lem:il:3} If  $\HasType{\Gamma}{\x}{\T}$, then $\Gamma(\x)\leqfj\T$.
\item\label{lem:il:4} If  $\HasType{\Gamma}{\FieldAccess{\e}{\f_i}}{\T}$, then $\HasType{\Gamma}{\e}{\mbox{\em\CC}}$ and {\em $\fields{\CC}=\Field{\T_1}{\f_1}\ldots\Field{\T_n}{\f_n}$} and $\T_i\leqfj T$ where 
$i\in 1..n$.
\item\label{lem:il:5} If  $\HasType{\Gamma}{\ConstrCall{\CC}{\e_1,\ldots,\e_n}}{\T}$, then $\mbox{\em\CC}\leqfj\T$ and {\em $\fields{\CC}=\Field{\T_1}{\f_1}\ldots\Field{\T_n}{\f_n}$} and $\HasType{\Gamma}{\e_i}{\T_i}$ for all $i\in 1..n$.
\item\label{lem:il:6} If  $\HasType{\Gamma}{\MethCall{\e_0}{\m}{\e_1,\ldots,\e_n}}{\T}$, then $\e_0$ not of shape $\LambdaExp{\xs}{\e}$ and 
$\HasType{\Gamma}{\e_i}{\T_i}$ for all $i\in 0..n$ and {\em $\mtype{\T_0}{\m}=\funtype{\T_1\ldots\T_n}{\T'}$} with $\T'\leqfj\T$.
\item\label{lem:il:7} If  $\HasType{\Gamma}{\LambdaExp{\xs}{\e}}{\T}$, then $\T={\mbox{\em\II}}$ and {\em $\umtype{\II}=\funtype{\T_1\ldots\T_n}{\T'}$} and \mbox{$\HasType{\x_1{:}\T_1,\ldots,\x_n{:}\T_n}{\e}{\T'}$.}
\item\label{lem:il:8} If  $\HasType{\Gamma}{\Cast{\T'}{\e}}{\T}$, then $\HasType{\Gamma}{\e}{\T'}$ and $\T'\leqfj \T$.
\item\label{lem:il:9} If  $\HasType{\Gamma}{\object{\text{\em\CC}}{\val_1,\ldots,\val_n}}{\T}{}$, 
then $\mbox{\em\CC}\leqfj\T$ and {\em $\fields{\CC}=\Field{\T_1}{\f_1}\ldots\Field{\T_n}{\f_n}$} and $\HasType{\Gamma}{\val_i}{\leqfj\T_i}$  for all $i\in 1..n$.
\end{enumerate}
\end{lemma}

\begin{lemma}[Canonical Forms]\label{lem:cfj}\
\begin{enumerate}
\item \label{lem:cfj:1}
If $\HasType\es \val {\text{\em\CC}}$, then $\val=\object{\text{\em\DD}}{\vals}$ and $\text{\em\DD}\leqfj\text{\em\CC}$.
\item \label{lem:cfj:2}
If $\HasType\es \val \text{\em\II}$, then either $\val=\object{\text{\em\CC}}{\vals}$ and $\text{\em\CC}\leqfj\text{\em\II}$ or $\val=\LambdaExp{\xs}{\e}$ and $\text{\em\II}$ is a functional interface.
\end{enumerate}
\end{lemma}

In order to prove soundness, set $\RuleSet_2$ the big-step semantics defined in \refToFigure{FJ-lambda-big-step}, and let $\Pred2_\T(\Conf{\env}{\e})$ hold if $\HasType{\es}{\Conf{\env}{\e}}{\leqfj\T}$, $\Pred2_T(\val)$ if $\HasType{\es}{\val}{\leqfj\T}$,  for $\T$ defined  in \refToFigure{FJ-lambda-syntax}.  

 To read this and the following soundness proofs of examples, it is convenient to refer to the reduction rules in inline format, where the dummy continuation $\eval\res\res$, if any, is made explicit. 

 \begin{theorem}[Soundness]\label{thm:sdfjl}
The big-step semantics $\RuleSet_2$ and the indexed predicate $\Pred2$ satisfy the conditions {\em \refToSound{preservation}}, {\em \refToSound{progress-ex}} and {\em \refToSound{progress-all}} of \refToSect{sc}.
 \end{theorem}
 \begin{proof} Condition \refToSound{preservation}.
The proof is by cases on instantiations of meta-rules. In considering a rule with typed consequence $\Conf{y_1{:}\hat\val_1,\ldots,y_p{:}\hat\val_p}\e$ \refToLemItem{il}{1} implies $\HasType\es{\hat\val_\ell}{\leqfj\hat\T_\ell}$ for all $\ell\in1\ldots p$ and $\HasType{y_1{:}\hat\T_1,\ldots,y_p{:}\hat\T_p}{\e}{\T}$ for some $\hat\T_1,\ldots,\hat\T_p$.\\
Rule \rn{var}. \refToLemItem{il}{1}  gives $\HasType{\es}{\env(\x)}{\leqfj\T'}$ and $\HasType{\x{:}\T'}{\x}{\T}$. \refToLemItem{il}{3} implies $\T'\leqfj\T$, so we conclude $\HasType{\es}{\env(\x)}{\leqfj\T}$ by transitivity of $\leqfj$.\\
Rule \rn{field-access}. \refToLemItem{il}{4} applied to $\HasType{\Gamma}{\FieldAccess{\e}{\f_i}}{\T}$ implies $\HasType{\Gamma}{\e}{\DD}$ and $\fields{\DD}=\Field{\T_1}{\f_1}\ldots\Field{\T_m}{\f_m}$ and $\T_i\leqfj T$ where 
$i\in 1..m$.  Since $\eval{\Conf{\env}\e}{\object{\CC}{\val_1,\ldots,\val_n}}$ is a premise we assume $\HasType{}{\object{\CC}{\val_1,\ldots,\val_n}}{\leqfj\DD}$, which implies $\CC\leqfj\DD$ and $\fields{\CC}=\Field{\T'_1}{\f'_1}\ldots\Field{\T'_n}{\f'_n}$ and $\HasType{\Gamma}{\val_j}{\leqfj\T'_j}$ for all $j\in 1..n$ by \refToLemItem{il}{9}. From $\CC\leqfj\DD$ we have $m\leq n$ and $\T_j=\T'_j$ and $\f_j=\f'_j$ for all $j\in 1..m$.  We conclude $\HasType{}{\val_i}{\leqfj\T}$.\\
Rule \rn{new}. \refToLemItem{il}{5} applied to $\HasType{\Gamma}{\ConstrCall{\CC}{\e_1,\ldots,\e_n}}{\T}$ implies $\CC\leqfj\T$ and $\fields{\CC}=\Field{\T_1}{\f_1}\ldots\Field{\T_n}{\f_n}$ and $\HasType{\Gamma}{\e_i}{\T_i}$ for all $i\in 1..n$. Since $\eval{\Conf{\env}{\e_i}}{\val_i}$ is a premise we assume $\HasType{}{\val_i}{\leqfj\T_i}$  for all $i\in 1..n$. Using rule \rn{t-object} we derive $\HasType{}{\object{\CC}{\val_1,\ldots,\val_n}}{\leqfj\T}$.\\
Rule \rn{invk}. \refToLemItem{il}{6} applied to $\HasType{\Gamma}{\MethCall{\e_0}{\m}{\e_1,\ldots,\e_n}}{\T}$ implies $\e_0$ not of shape $\LambdaExp{\xs}{\e}$ and 
$\HasType{\Gamma}{\e_i}{\T_i}$ for all $i\in 0..n$ and $\mtype{\T_0}{\m}=\funtype{\T_1\ldots\T_n}{\T'}$ with $\T'\leqfj\T$. Since  $\eval{\Conf{\env}{\e_0}}{\object{\CC}{\vals'}}$ is a premise we assume $\HasType{}{\object{\CC}{\vals'}}{\leqfj\T_0}$, which implies $\CC\leqfj\T_0$ by \refToLemItem{il}{9}. Since $\eval{\Conf{\env}{\e_i}}{\val_i}$ is a premise we assume 
$\HasType{}{\val_i}{\leqfj\T_i}$ for all $i\in 1..n$. We have $\mtype{\CC}{\m}=\funtype{\T_1\ldots\T_n}{\T'}$ since $\mtype{\T_0}{\m}=\funtype{\T_1\ldots\T_n}{\T'}$ and $\CC\leqfj\T_0$. The typing conditions on the class table imply   $\HasType{\Gamma_0}{\e}{\leqfj\T'}$ where $\Gamma_0=\set{\x_1{:}\T_1,\ldots,\x_n{:}\T_n,\this{:}\CC}$. Therefore   using rule \rn{t-conf}  we derive $\HasType\es{\Conf{\x_1{:}\val_1,\ldots,\x_n{:}\val_n,\this{:}\object{\CC}{\vals'}}{\e}}{\leqfj\T}$.\\
Rule \rn{lambda-invk}. \refToLemItem{il}{6} applied to $\HasType{\Gamma}{\MethCall{\e_0}{\m}{\e_1,\ldots,\e_n}}{\T}$ implies $\e_0$ not of shape $\LambdaExp{\xs}{\e'}$ and 
$\HasType{\Gamma}{\e_i}{\T_i}$ for all $i\in 0..n$ and $\mtype{\T_0}{\m}=\funtype{\T_1\ldots\T_n}{\T'}$ with $\T'\leqfj\T$. Since $\eval{\Conf{\env}{\e_0}}{\LambdaExp{\xs}{\e}}$ is a premise we assume $\HasType{}{\LambdaExp{\xs}{\e}}{\leqfj\T_0}$, which implies $\II\leqfj\T_0$ and $\umtype{\II}=\funtype{\T_1\ldots\T_n}{\T'}$ and $\HasType{\x_1{:}\T_1,\ldots,\x_n{:}\T_n}{\e}{\T'}$ by \refToLemItem{il}{7}. Since  $\eval{\Conf{\env}{\e_i}}{\val_i}$ is a premise we assume 
$\HasType{}{\val_i}{\leqfj\T_i}$ for all $i\in 1..n$.  Therefore  using rule \rn{t-conf}  we derive $\HasType\es{\Conf{\x_1{:}\val_1,\ldots,\x_n{:}\val_n}{\e}}{\leqfj\T}$.\\
Rule \rn{upcast}. \refToLemItem{il}{8} applied to $\HasType{\Gamma}{\Cast{\T'}{\e}}{\T}$ implies $\HasType{\Gamma}{\e}{\leqfj \T}$. From $\eval{\Conf{\env}\e}\val$ we conclude $\HasType\es\val{\leqfj\T}$.\\
Condition \refToSound{progress-ex}. It is easy to verify that if $\e$ is generated by the grammar of \refToFigure{FJ-lambda-syntax}, then there is a rule in  \refToFigure{FJ-lambda-big-step} whose conclusion is $\Conf\env\e$.  In particular, for a configuration of shape $\Conf\env\x$, rule \rn{var}  can be applied, since $\HasType\es{\Conf\env\x}\T$ implies that $\x$ is in the domain of $\env$ by Lemmas \ref{lem:il}~(\ref{lem:il:1}) and (\ref{lem:il:3}).

Condition \refToSound{progress-all}.  Rule \rn{var} requires that $\val$ reduces to $\val$, and this is the only derivable judgment for $\val$. 
Rule \rn{field-access} requires that $\Conf{\env}{\e}$ reduces to $\val=\object{\DD}{\val_1,\ldots,\val_m}$ such that $\fields{\DD}=\Field{\T_1}{\f_1}\ldots\Field{\T_m}{\f_m}$, and $i\in 1..m$. Typing rule \rn{t-field-access} prescribes for the expression $\e$ a class type $\CC$ such that $\fields{\CC}=\Field{\T_1}{\f_1}\ldots\Field{\T_n}{\f_n}$, and $i\in 1..n$. The validity of condition \refToSound{preservation} (which assures type preservation by \refToLem{sr}), and  \refToLemItem{cfj}{1}, imply that $\val$ is an object of a subclass $\DD$ of $\CC$, and the well-formedness of the class table implies that $n\leq m$, hence $i\in 1..m$.\\
For a method call $\MethCall{\e_0}{\m}{\e_1,\ldots,\e_n}$, the configuration $\Conf{\env}{\e_0}$ can reduce either to an object $\object{\CC}{\vals}$ or to a lambda-expression. In the first case we can apply rule \rn{invk}, and in the second case rule \rn{$\lambda$-invk}.  In the first case, typing rule \rn{t-invk} prescribes, for the expression $\e_0$,  a type $\T_0$ such that $\mtype{\T_0}{\m}=\funtype{\T_1\ldots\T_n}{\T}$. The validity of condition \refToSound{preservation} (which assures type preservation by \refToLem{sr}), and  \refToLemItem{cfj}{1}, imply that $\CC\leqfj\T_0$, and the well-formedness of the class table implies that $\mtype{\CC}{\m}=\funtype{\T_1\ldots\T_n}{\T}$. 
Other meta-variables for values can be freely instantiated.  
 In rule \rn{up-cast} the meta-variable $\val$ can be freely instantiated. 
\end{proof}

\subsection{Intersection and union types}\label{sect:iut}
We enrich the type system of \refToFigure{lambda-typesystem} by adding intersection and union type constructors and the corresponding typing rules, see \refToFigure{iutypes}. As usual we require an infinite number of arrows in each infinite path for the trees representing types. 
 Intersection types for the $\lambda$-calculus have been widely studied \cite{BDS13}. Union types naturally model conditionals \cite{Grudzinski00} and non-deterministic choice \cite{DdLP98}. 
\begin{figure}
\begin{center}
\begin{small}
\begin{math}
\begin{array}{rcll}
\production{\tA}{\natType\mid \tA_1\to \tA_2\mid \tA_1\wedge \tA_2\mid \tA_1\vee \tA_2}{type}
\end{array}
\end{math}
\HSep
\begin{math}
\begin{array}{c}
\MetaRule
{$\wedge$ I}
{\HasType{\Gamma}{\e}{\T}\Space\HasType{\Gamma}{\e}{\Ts}}
{\HasType{\Gamma}{\e}{\T\wedge\Ts}}{}
\BigSpace
\MetaRule{$\wedge$ E}{\HasType{\Gamma}{\e}{\T\wedge\Ts}}{\HasType{\Gamma}{\e}{\T}}{}\BigSpace
\MetaRule{$\wedge$ E}{\HasType{\Gamma}{\e}{\T\wedge\Ts}}{\HasType{\Gamma}{\e}{\Ts}}{}\\[3ex]
\MetaRule{$\vee$ I}{\HasType{\Gamma}{\e}{\T}}{\HasType{\Gamma}{\e}{\T\vee\Ts}}{}\BigSpace
\MetaRule{$\vee$ I}{\HasType{\Gamma}{\e}{\Ts}}{\HasType{\Gamma}{\e}{\T\vee\Ts}}{}
\end{array}
\end{math}
\end{small}
\end{center}
\caption{Intersection and union types: syntax and typing rules}\label{fig:iutypes}
\end{figure}

The typing rules for the introduction and the elimination of intersection and union are standard,  except  for the absence of the union elimination rule:\\
 \centerline{$ 
\MetaRule{$\vee E$} {\HasType{\SubstFun\Gamma\tA\x}
\e\tC\Space \HasType{\SubstFun\Gamma\tB\x} \e\tC\Space
\HasType\Gamma {\e'}{\tA\vee \tB}}{\HasType\Gamma 
{\subst\e{\e'}\x}\tC}
$} 
As a matter of fact rule \rn{$\vee E$} is unsound for $\oplus$. 
For example, let split the type $\natType$ into $\even$ and $\odd$ and add the expected typings for natural numbers. The prefix addition $\mathtt+$ has type\\ \centerline{$(\funtype{\funtype\even\even}\even)\wedge(\funtype{\funtype\odd\odd}\even)$} and we  derive\\
\[\footnotesize
\prooftree
\HasTypeNarrow{x{:}\even}{+\,x\,x}{\even}
~~
\HasTypeNarrow{x{:}\odd}{+\,x\,x}{\even}
~~
\prooftree
\prooftree
\HasType{}1\odd
\justifies
\HasType{}1\even\vee\odd
\using \rn{$\vee$ I}
\endprooftree
~~
\prooftree
\HasType{}2\even
\justifies
\HasType{}2\even\vee\odd
\using \rn{$\vee$ I}
\endprooftree
\justifies
\HasType{}{(1\oplus2)}{\even\vee\odd}
\using\rn{$\oplus$}
\endprooftree
\justifies
\HasType{}{\!\mathtt{+(1\oplus2)(1\oplus2)}}\even
\using \rn{$\vee$ E}
\endprooftree\]

\noindent
 We cannot assign the type $\even$ to $3$, which is a possible result, so strong soundness is lost. In the small-step approach, we cannot assign $\even$ to the intermediate term $+\,1\,2$, so subject reduction fails.  
In the big-step approach, there is no such intermediate term; however, condition \refToSound{preservation} fails for the reduction rule for $+$. Indeed, considering the following instantiation of the rule:\\
\[\footnotesize
\MetaRule{$+$}{\eval{1\oplus2}{1}\Space \eval{1\oplus2}{2}\Space \eval{3}{3}}{\eval{\!\mathtt{+(1\oplus2)(1\oplus2)}}{3}}{}
\]
and the type $\even$ for the consequence, we cannot assign this type to the (configuration in) last premise (continuation). 


Intersection types allow to derive meaningful types also for expressions containing variables applied to themselves, for example we can derive\\ \centerline{$\HasType{}{\LambdaExp{\x}{\AppExp{\x}{\x}}}{\funType{(\funType{\tA} {\tB})\wedge\tA}}\tB$} With union types all non-deterministic choices between typable expressions can be typed too, since we can derive $\HasType\Gamma {\e_1\oplus\e_2}{\T_1\vee\T_2}$ from $\HasType\Gamma {\e_1}{\T_1}$ and  $\HasType\Gamma {\e_2}{\T_2}$.

In order to prove that the reduction rules satisfy the soundness conditions for the typing system, standard lemmas are handy. We first define the {\em subtyping} relation $\tA\leq \tB$ as the smallest preorder such that:
\begin{itemize}
\item $\tA_1\leq \tB$ and $\tA_2\leq \tB$ imply $\tA_1\wedge \tA_2\leq \tB$;
\item $\tA\wedge \tB\leq \tA$ and $\tA\wedge \tB\leq \tB$;
\item $\tA\leq \tA \vee \tB$ and $\tA\leq \tB\vee \tA$.
\end{itemize}
It is easy to verify that $\tA\leq\tB$ iff $\Gamma, \e:\tA\vdash \e:\tB$ for an arbitrary $\e$ using rules  $(\wedge I)$, $(\wedge E)$ and $(\vee I)$.

\begin{lemma}[Inversion]\label{lem:iliu}\
\begin{enumerate}
\item \label{lem:iliu:1}If $\HasType\Gamma\x\tA$, then $\Gamma(\x)\leq \tA$.
\item \label{lem:iliu:2} If $\HasType{\Gamma}{\const}\T$, then $\natType\leq\T$.
\item \label{lem:iliu:3} If $\HasType\Gamma{\LambdaExp{\x}{\e}}\T$, then $\HasType{\SubstFun\Gamma{\tB_i}\x}\e {\tC_i}$ for $1\leq i\leq m$ and\\ \centerline{$\bigwedge_{1\leq i\leq m}(\funType{\tB_i} {\tC_i})\leq \tA$.}
\item \label{lem:iliu:4} If $\HasType\Gamma{\e_1\appop\e_2}\tA$, then $\HasType\Gamma  {\e_1} {\funType{\tB_i}{\tC_i}}$  and $\HasType\Gamma{\e_2}{\tB_i}$    for $1\leq i\leq m$ and $\bigwedge_{1\leq i\leq m}\tC_i\leq \tA$.
\item \label{lem:iliu:5} If $\HasType{\Gamma}{\SuccExp{\e}}\T$, then $\natType\leq\T$ and $\HasType{\Gamma}{\e}{\natType}$.
\item \label{lem:iliu:6} If $\HasType\Gamma {\e_1\oplus\e_2}\tA$, then $\HasType\Gamma{\e_i}\tA$ with $i\in1,2$.
\end{enumerate}
\end{lemma}
\begin{proof}By induction on derivations and by cases on the last applied typing rule.\\
\refToItem{iliu}{3}. If the last applied rule is \rn{$\wedge I$}, then $\T=\T_1\wedge\T_2$ and $\HasType\Gamma{\LambdaExp{\x}{\e}}{\T_j}$ with $j\in 1,2$. By IH $\HasType{\SubstFun\Gamma{\tB_i^{(j)}}\x}\e {\tC_i^{(j)}}$ for $1\leq i\leq m_j$ and $\bigwedge_{1\leq i\leq m_j}(\funType{\tB_i^{(j)}}{ \tC_i^{(j)}})\leq \tA_j$ with $j\in 1,2$. Then we conclude\\ \centerline{$\bigwedge_{1\leq i\leq m_1}(\funType{\tB_i^{(1)}}{ \tC_i^{(1)}})\wedge
\bigwedge_{1\leq i\leq m_2}(\funType{\tB_i^{(2)}}{ \tC_i^{(2)}})\leq \tA_1\wedge\tA_2$.}
\refToItem{iliu}{4}. If the last applied rule is \rn{$\wedge I$}, then $\T=\T_1\wedge\T_2$ and $\HasType\Gamma{\e_1\appop\e_2}\tA_j$ with $j\in 1,2$. By IH $\HasType\Gamma  {\e_1} {\funType{\tB_i^{(j)}}{\tC_i^{(j)}}}$,  and $\HasType\Gamma{\e_2}{\tB_i^{(j)}}$,  and  for $1\leq i\leq m_j$ and $\bigwedge_{1\leq i\leq m_j}\tC_i^{(j)}\leq \tA_j$ with $j\in 1,2$. Then we conclude\\ \centerline{ $\bigwedge_{1\leq i\leq m_1}\tC_i^{(1)}\wedge \bigwedge_{1\leq i\leq m_2}\tC_i^{(2)}\leq \tA_1\wedge\tA_2$.}
\end{proof}

\begin{lemma}[Substitution]\label{lem:siu}
If $\HasType{\SubstFun\Gamma{\tA'}\x} \e {\tA}$ and $\HasType\Gamma{\e'} {\tA'}$, then $\HasType\Gamma{\subst\e{\e'}\x} {\tA}$.
\end{lemma}

\begin{lemma}[Canonical Forms]\label{lem:cfiu}\
\begin{enumerate}
\item \label{lem:cfiu:1}
If $\HasType\es \val {\funType{\tA'}{\tA}}$, then $\val=\LambdaExp{\x}{\e}$.
\item \label{lem:cfiu:2}
If $\HasType\es \val \natType$, then $\val=\const$.
\end{enumerate}
\end{lemma}

In order to prove soundness,  let
 $\Pred3_T(\e)$  be  
 $\HasType{\es}{\e}{\T}$,  for $\T$  defined in \refToFigure{iutypes}. 

\begin{theorem}[Soundness]\label{thm:sdiu}
The big-step semantics $\RuleSet_1$ and the indexed predicate $\Pred3$ satisfy the conditions {\em \refToSound{preservation}}, {\em \refToSound{progress-ex}} and {\em \refToSound{progress-all}} of \refToSect{sc}.
 \end{theorem}
 \begin{proof}
Condition \refToSound{preservation}. The proof is by cases on instantiations of  meta-rules. For rule \rn{app} \refToLemItem{iliu}{4}  applied to $\HasType\es{\e_1\appop\e_2}\tA$ implies $\HasType{\es}  {\e_1} {\funType{\tB_i}{\tC_i}}$  and $\HasType{\es}{\e_2}{\tB_i}$  for $1\leq i\leq m$ and $\bigwedge_{1\leq i\leq m}\tC_i\leq \tA$. As in the proof of  \refToThm{sd} we get 
$\HasType{\es}{\LambdaExp{\x}{\e}}{\funType{\tB_i}{\tC_i}}$ and $\HasType{\es}{\val_2} {\tB_i}$ for $1\leq i\leq m$.  \refToLemItem{iliu}{3} implies $\HasType{\x {\ : \ } \tB_i} \e {\tC_i}$, so by \refToLem{siu} we have $\HasType{\es}{\subst\e{\val_2}\x} {\tC_i}$ for $1\leq i\leq m$. We can derive $\HasType{\es}{\subst\e{\val_2}\x} {\tA}$ using rules $(\wedge I)$, $(\wedge E)$ and $(\vee I)$.\\
Condition \refToSound{progress-ex}. The proof is as in  \refToThm{sd}.\\
Condition \refToSound{progress-all}. The proof is by cases on instantiations of  meta-rules. For rule \rn{app} \refToLemItem{iliu}{4}  applied to $\HasType{\es}{\e_1\appop\e_2}\tA$ implies $\HasType{\es}  {\e_1} {\funType{\tB_i}{\tC_i}}$  for $1\leq i\leq m$. If $\eval{\e_1}\val$ we get $\HasType{\es}  {\val} {\funType{\tB_i}{\tC_i}}$  for $1\leq i\leq m$ as in the proof of  \refToThm{sd}. \refToLemItem{cfiu}{1} applied to $\HasType{\es}{\val} {\funType{\tB_i}{\tC_i}}$ implies $\val=\LambdaExp{\x}{\e}$ and therefore the premises of the rule can be satisfied. 
\end{proof}

 
\subsection{$\MiniFJOS$}\label{sect:fjos}
A well-known example in  which  
proving soundness with respect to small-step semantics is extremely challenging is  the standard type system with intersection and union types \cite{BDL95}
w.r.t. the pure $\lambda$-calculus with full reduction.  
Indeed, the standard subject reduction technique fails\footnote{For this reason, in \cite{BDL95} soundness is proved by an  ad-hoc  
technique, that is, by considering parallel reduction and an equivalent type system \`a la Gentzen, which enjoys the cut elimination property. }, 
since, for instance, we can derive  the type $(\tA\to \tA\to \tC)\wedge (\tB\to \tB\to \tC)\to (\tD\to \tA \vee \tB)\to \tD \to \tC$ for both 
$\lambda x. \lambda y. \lambda z. x((\LambdaExp{t}{t})(\AppExp{y}{z}))((\LambdaExp{t}{t})(\AppExp{y}{z}))$ 
and
$\lambda x. \lambda y. \lambda z. x(\AppExp{y}{z})(\AppExp{y}{z})$, 
 but the intermediate  expressions  
 $\lambda x. \lambda y. \lambda z. x((\LambdaExp{t}{t})(\AppExp{y}{z}))(\AppExp{y}{z})$ and $\lambda x. \lambda y. \lambda z. x(\AppExp{y}{z})((\LambdaExp{t}{t})(\AppExp{y}{z}))$ do not have this type.  

As the example shows, the key problem is that rule \rn{$\vee E$} can be applied to  expression  
$\e$ where the same  subexpression  
$\e'$ occurs more than once. In the non-deterministic case, as shown by the example in the previous section, this is unsound, since $\e'$ can reduce to different values. In the deterministic case, instead, this is sound, but cannot be proved by subject reduction. 
Since using big-step semantics there are no intermediate steps to be typed, our approach seems very promising to investigate an alternative 
proof of soundness.
Whereas we leave this challenging problem to future work, here as first step we describe a (hypothetical) calculus with a much simpler version of the problematic feature.

The calculus is a variant of $\FJ$ \cite{IPW01} with intersection and union types.  Methods have intersection types with the same return type and different parameter types, modelling  a form of  {\em overloading}. Union types enhance typability of conditionals. The more interesting feature is the possibility of replacing an arbitrary number of parameters with the same expression having an union type. We dub this calculus $\MiniFJOS$. 
\begin{figure}[t]
\begin{small}
\begin{grammatica}
\produzione{\e}{\x\mid\val\mid \FieldAccess{\e}{\f}\mid
\MethCall{\e}{\m}{\e_1,\ldots,\e_n}\mid\ifte{\e}{\e_1}{\e_2}}{\qquad expression}\\
\produzione{\val}{\ConstrCall{\CC}{\val_1,\ldots,\val_n}\mid\true\mid\false}{\qquad value}\\
\produzione{\T}{\CC\mid\boolType\mid\bigvee_{1\leq i\leq n}\T_i}{\qquad  expression  type}\\
\produzione{\MT}{\bigwedge_{1\leq i\leq m}(\funtype{\CC_1^{(i)}\ldots\CC_n^{(i)}}{\DD})}{\qquad method type}
\end{grammatica}
\HSep
\begin{math}
\begin{array}{l}
\MetaRule{field-access}{\eval{\e}{\ConstrCall{\CC}{\val_1,\ldots,\val_n}}}{
\eval{\FieldAccess{\e}{\f_i}}{\val_i}}
{\begin{array}{l}
\fields{\CC}=\Field{\T_1}{\f_1}\ldots\Field{\T_n}{\f_n}\\
i\in 1..n
\end{array}
}\\[4ex]
\MetaRule{new}{
\eval{\e_i}{\val_i}\Space \forall i\in 1..n}{\eval{\ConstrCall{\CC}{\e_1,\ldots,\e_n}}{\ConstrCall{\CC}{\val_1,\ldots,\val_n}}}
{
}\\[4ex]
\MetaRule{invk}{
\begin{array}{l}
\eval{\e_0}{\ConstrCall{\CC}{\vals'}}\\
\eval{\e_i}{\val_i}\Space \forall i\in 1..n\\
\eval{\subst{\subst{\subst\e{\val_1}{\x_1}\ldots}{\val_n}{\x_n}}{\ConstrCall{\CC}{\vals'}}{\this}}{\val}
\end{array}}{
\eval{\MethCall{\e_0}{\m}{\e_1,\ldots,\e_n}}
{\val}}
{\begin{array}{l}
\mbody{\CC}{\m}=\Pair{\x_1\ldots\x_n}{\e}
\end{array}
}
\end{array}
\end{math}
\HSep
\begin{math}
\begin{array}{l}
\metainlinerule{field-access}
{\eval{\e}{\ConstrCall{\CC}{\val_1,\ldots,\val_n}}}
{\eval{\val_i}{\val_i}}
{\FieldAccess{\e}{\f_i}} \BigSpace\fields{\CC}=\Field{\T_1}{\f_1}\ldots\Field{\T_n}{\f_n}
\Space
i\in 1..n\\
\metainlinerule{new}
{\eval{\mathit{es}}{\vals}}
{\eval{\ConstrCall{\CC}{\val_1,\ldots,\val_n}}{\ConstrCall{\CC}{\val_1,\ldots,\val_n}}}
{\ConstrCall{\CC}{\e_1,\ldots,\e_n}}\\
\metainlinerule{invk}{\eval{\e_0}{\ConstrCall{\CC}{\vals'}}, \eval{\mathit{es}}{\vals}}{\eval{\e'}{\val}}{\MethCall{\e_0}{\m}{\e_1,\ldots,\e_n}}\\
\hfill \e'=\subst{\subst{\subst\e{\val_1}{\x_1}\ldots}{\val_n}{\x_n}}{\ConstrCall{\CC}{\vals'}}{\this}\qquad\mbody{\CC}{\m}=\Pair{\x_1\ldots\x_n}{\e}\\
\text{where }\eval{\mathit{es}}{\vals}\text{ is short for }\eval{\e_1}{\val_1}, \ldots, \eval{\e_n}{\val_n}
\end{array}
\end{math}
\HSep
\begin{math}
\begin{array}{l}
\MetaRule{t-var}{}{\HasType{\Gamma}{\x}{\T}}{\Gamma(\x)=\T}
\BigSpace
\MetaRule{t-field-access}{\HasType{\Gamma}{\e}{\CC}}{\HasType{\Gamma}{\FieldAccess{\e}{\f_i}}{\CC_i}}
{\begin{array}{l}
\fields{\CC}=\Field{\T_1}{\f_1}\ldots\Field{\T_n}{\f_n}\\
i\in 1..n
\end{array}
}\\[4ex]
\MetaRule{t-new}{\HasType{\Gamma}{\e_i}{\CC_i}\Space \forall i\in 1..n}{\HasType{\Gamma}{\ConstrCall{\CC}{\e_1,\ldots,\e_n}}{\CC}}
{\begin{array}{l}
\fields{\CC}=\Field{\T_1}{\f_1}\ldots\Field{\T_n}{\f_n}\\
\end{array}
}
\\[4ex]
\MetaRule{t-invk}{\HasType{\Gamma}{\e_i}{\CC_i}\Space \forall i\in 0..n\Space\HasType{\Gamma}{\e}{\bigvee_{1\leq i\leq m}\DD_i}}{\HasType{\Gamma}{\MethCall{\e_0}{\m}{\e_1,\ldots,\e_n,\underbrace{\e,\ldots,\e}_{p}}}{\CC}}
{\begin{array}{l}
 \mtype{\CC_0}{\m}\leqfj\\
 \bigwedge_{1\leq i\leq m}(\funtype{\CC_1\ldots\CC_n\underbrace{\DD_i\ldots\DD_i}_{p}}{\CC})\\
\end{array}
}
\\[4ex]
\MetaRule{t-if}{\HasType{\Gamma}{\e}{\boolType}\Space\HasType{\Gamma}{\e_1}{\T}\Space\HasType{\Gamma}{\e_2}{\T}}{\HasType{\Gamma}{\ifte\e{\e_1}{\e_2}}{\T}}{}\BigSpace
\MetaRule{t-sub}{\HasType{\Gamma}{\e}{\T}}{\HasType{\Gamma}{\e}{\T'}}{
\begin{array}{l}
\T\leqfj\T'
\end{array}}
\end{array}
\end{math}
\end{small}
\caption{$\MiniFJOS$: syntax, big-step semantics and type system}\label{fig:FJOS-typesystem}
\end{figure}

\refToFigure{FJOS-typesystem} gives the syntax, big-step semantics and typing rules of $\MiniFJOS$.  We omit the standard big-step rule for conditional, and typing rules for boolean constants.  

 The subtyping relation $\leqfj$ is the reflexive and transitive closure of the union of the $\aux{extends}$ relation and the standard rules for union:\\
\centerline{$\T_1\leqfj\T_1\vee\T_2\Space\qquad\T_1\leqfj\T_2\vee\T_1$}
On the other hand, \emph{method types} (results of the \aux{mtype} function) are now \emph{intersection types}, and the subtyping relation on them is the reflexive and transitive closure of the standard rules for intersection:\\
\centerline{$\MT_1\wedge\MT_2\leqfj\MT_1\Space\qquad\MT_1\wedge\MT_2\leqfj\MT_2$}

The functions \aux{fields} and \aux{mbody} are defined as for $\MiniFJLambda$.\\
 Instead 
$\mtype{\CC}{\m}$ gives, for each method $\m$ in class $\CC$, an intersection type.  
We assume $\mbody{\CC}{\m}$ and $\mtype{\CC}{\m}$ either both defined or both undefined: in the first case\\ 
\centerline{$\begin{array}{c}
\mbody{\CC}{\m}{=}\Pair{\x_1\ldots\x_n}{\e},~ \mtype{\CC}{\m}{=}\bigwedge_{1\leq i\leq m}(\funtype{\CC_1^{(i)}\ldots\CC_n^{(i)}}{\DD})\\ 
\text{and }\HasType{\x_1{:}\CC_1^{(i)},\ldots,\x_n{:}\CC_n^{(i)},\this{:}\CC}{\e}{\DD}\text{ for }i\in 1..m
\end{array}$}

Clearly rule \rn{t-invk} is inspired by rule \rn{$\vee E$}, but the restriction to method calls  endows  a standard inversion lemma. 
 The subtyping in this rule allows to choose the types for the method best fitting the types of the arguments.  Not surprisingly, subject reduction fails for the expected small-step semantics. For example,
let class $\CC$ have a field $\aux{point}$ which contains cartesian coordinates and class $\DD$ have a field $\aux{point}$ which contains polar coordinates. The method $\aux{eq}$ takes two objects and compares their $\aux{point}$ fields returning a boolean value.   A type for this  
method is 
\mbox{$(\funtype{\CC\,\CC}\boolType)\wedge(\funtype{\DD\,\DD}\boolType)$} and we can type  $\aux{eq}(\e,\e )$,  where
\begin{quote} 
$\e=\ifte{\false}{\ConstrCall\CC\ldots}{\ConstrCall\DD\ldots}$
\end{quote}
 In fact $\e$ has type $\CC\vee\DD$. Notice that in a standard small-step semantics
\begin{quote}
$\begin{array}{l}\aux{eq}(\e,\e) \longrightarrow  \aux{eq}(\ConstrCall\DD\ldots ,\ifte{\false}{\ConstrCall\CC\ldots}{\ConstrCall\DD\ldots})\end{array}$
\end{quote}
and this last expression cannot be typed.

As in previous examples the soundness proof uses an inversion lemma and a substitution lemma,  whereas the canonical form lemma is trivial, notably the only values of type $\CC$ are objects (constructor calls with values as arguments) of a subclass. We need instead a lemma (dubbed key) which assures  that a value typed by a union of classes can also be typed by one of theses classes. The proof of this lemma is straightforward, since values are new constructors.

\begin{lemma}[Inversion]\label{lem:ilos}\
\begin{enumerate}
\item\label{lem:ilos:1} If  $\HasType{\Gamma}{\x}{\T}$, then $\Gamma(\x)\leqfj\T$.
\item\label{lem:ilos:2} If  $\HasType{\Gamma}{\FieldAccess{\e}{\f_i}}{\T}$, then $\HasType{\Gamma}{\e}{\mbox{\em\CC}}$ and {\em $\fields{\CC}=\Field{\CC_1}{\f_1}\ldots\Field{\CC_n}{\f_n}$} and $\mbox{\em\CC}_i\leqfj T$ where 
$i\in 1..n$.
\item\label{lem:ilos:3} If  $\HasType{\Gamma}{\mbox{\em {\sf new}~\CC}(\e_1,\ldots,\e_n)}{\T}$, then $\mbox{\em\CC}\leqfj\T$ and {\em $\fields{\CC}=\Field{\mbox{\em\CC}_1}{\f_1}\ldots\Field{\CC_n}{\f_n}$} and $\HasType{\Gamma}{\e_i}{\mbox{\em\CC}_i}$ for all $i\in 1..n$.
\item\label{lem:ilos:4} If  $\HasType{\Gamma}{\MethCall{\e_0}{\m}{\e_1,\ldots,\e_n}}{\T}$, then  $n=q+p$ and
$\HasType{\Gamma}{\e_i}{\mbox{\em\CC}_i}$ for all $i\in 0..q$ and $\e_{q+1}=\ldots=\e_n=\e$ and $\HasType{\Gamma}{\e}{\bigvee_{1\leq i\leq m}{\mbox{\em\DD}}_i}$ and {\em $\mtype{\CC_0}{\m}\leqfj\bigwedge_{1\leq i\leq m}(\funtype{\CC_1\ldots\CC_p\underbrace{\DD_i\ldots\DD_i}_ p}{\CC}) $} with $\mbox{\em\CC}\leqfj\T$.
\item\label{lem:ilos:5} If   $\HasType{\Gamma}{\mbox{\em \sf if}~\e~\mbox{\em \sf then}~\e_1~\mbox{\em \sf else}~\e_2}{\T}$, then $\HasType{\Gamma}{\e}{\boolType}$  and  $\HasType{\Gamma}{\e_1}{\T}$  and $\HasType{\Gamma}{\e_2}{\T}$.
\end{enumerate}
\end{lemma}

\begin{lemma}[Substitution]\label{lem:sos}
If $\HasType{\SubstFun\Gamma{\tA'}\x} \e {\tA}$ and $\HasType\Gamma{\e'} {\tA'}$, then $\HasType\Gamma{\subst\e{\e'}\x} {\tA'}$.
\end{lemma}

\begin{lemma}[Key]\label{lem:key}
If $\HasType{\Gamma}{\val}{\bigvee_{1\leq i\leq n}\CC_i}$, then $\HasType{\Gamma}{\val}{\CC_i}$ for some $i\in 1\ldots n$. 
\end{lemma}


In order to prove soundness,  
 let $\RuleSet_4$ be  the big-step semantics defined in \refToFigure{FJOS-typesystem}, and let $\Pred4_T(\e)$ hold if $\HasType{\es}{\e}{\T}$, for $\T$  defined  in \refToFigure{FJOS-typesystem}. 

\begin{theorem}[Soundness]\label{thm:sdos}
The big-step semantics $\RuleSet_4$ and the indexed predicate $\Pred4$ satisfy the conditions {\em \refToSound{preservation}}, {\em \refToSound{progress-ex}} and {\em \refToSound{progress-all}} of \refToSect{sc}.
 \end{theorem}
 \begin{proof}
Condition \refToSound{preservation}. The proof is by cases on  instantiations  of meta-rules. For rule \rn{invk} \refToLemItem{ilos}{4}  applied to $\HasType\es{\MethCall{\e_0}{\m}{\e_1,\ldots,\e_n}}{\T}$ implies $n=q+p$ and
$\HasType\es{\e_i}{\CC_i}$ for all $i\in 0..q$ and $\e_{q+1}=\ldots=\e_n=\e$ and $\HasType\es{\e}{\bigvee_{1\leq j\leq m}\DD_j}$ and\\ \centerline{$\mtype{\CC_0}{\m}\leqfj\bigwedge_{1\leq j\leq m}(\funtype{\CC_1\ldots\CC_p\underbrace{\DD_j\ldots\DD_j}_p}{\CC'}) $} with $\CC'\leqfj\T$. Then we get 
$\HasType{\es}{\ConstrCall{\CC}{\vals'}}{\CC_0}$ and $\HasType{\es}{\val_i} {\CC_i}$ for $i\in 1\ldots q$ and $\val_{q+1}=\ldots=\val_n=\val'$ and $\HasType\es{\val'}{\bigvee_{1\leq j\leq m}\DD_j}$, which implies $\HasType\es{\val'}{\DD_j}$ for some $j\in 1\ldots m$ by \refToLem{key}. The typing of the class table implies $\HasType{\x_1{:}\CC_1,\ldots,\x_q{:}\CC_q,\x_{q+1}:\DD_j, \ldots, \x_n:\DD_j,\this{:}\CC}{\e}{\CC'}$ for all $j\in 1\ldots m$.
\refToLem{sos} gives $\HasType{\es}{\subst{\subst{\subst\e{\val_1}{\x_1}\ldots}{\val_n}{\x_n}}{\ConstrCall{\CC}{\vals'}}{\this}}{\CC'}$. We can conclude $\HasType{\es}{\subst{\subst{\subst\e{\val_1}{\x_1}\ldots}{\val_n}{\x_n}}{\ConstrCall{\CC}{\vals'}}{\this}} {\T}$ using rule \rn{t-sub}.

Condition \refToSound{progress-ex}. All the closed expressions which are not values appear as conclusions in the reduction rules.

Condition \refToSound{progress-all}.  Rules \rn{field-access} and \rn{invk} require that the expression in the first premise reduces to an object for which the side-condition holds, and this can be proved exactly as in the corresponding cases in \refToThm{sdfjl}, by using the typing rules \rn{t-field-access}, and \rn{t-invk}, respectively, the validity of condition \refToSound{preservation} (which assures type preservation by \refToLem{sr}), the fact that canonical forms  of type $\CC$ are objects of a subclass, and the well-formedness of the class table. Other meta-variables for values can be freely instantiated. 

\end{proof}

\subsection{Imperative $\FJ$}\label{sect:ifj}

In \refToFigure{FJ-imp} and \refToFigure{FJ-impT} we show a minimal imperative extension of $\FJ$.  
We assume a well-typed class table and we use the notations introduced in \refToSect{fjl}. Expressions are enriched with field assignment and \emph{object identifiers} $\oid$, which only occur in runtime expressions.  A \emph{memory} $\mu$ maps object identifiers to \emph{object states}, which are expressions of shape $\ConstrCall{\CC}{\oid_1,\ldots\oid_n}$. Results are configurations of shape $\Conf{\mu}{\oid}$. We denote by $\UpdateMem{\mu}{\oid}{i}{\oid'}$ the memory obtained from $\mu$ by replacing by $\iota'$ the $i$-th field of the object state associated to $\oid$. The {\em type assignment} $\Sigma$ maps object identifiers into types (class names). We write $\HasType{\Sigma}{\e}{\CC}$ for $\HasTypeMem{\emptyset}{\Sigma}{\e}{\CC}$.   The subtyping relation $\leqfj$ is the reflexive and transitive closure of the $\aux{extends}$ relation.

\begin{figure}
\begin{small}
\begin{grammatica}
\produzione{\e}{\x\mid\FieldAccess{\e}{\f}\mid\ConstrCall{\CC}{\e_1,\ldots,\e_n}\mid
\MethCall{\e}{\m}{\e_1,\ldots,\e_n}\mid\FieldAssign{\e}{\f}{\e'}\mid\oid}{\qquad expression}\\
\produzione{\conf}{\Conf{\mu}{\e}}{\qquad configuration}\\
\end{grammatica}

\HSep
\begin{math}
\begin{array}{l}
\\
\MetaRule{field-access}{\eval{\Conf{\mu}{\e}}{\Conf{\mu'}{\oid}}}{
\eval{\Conf{\mu}{\FieldAccess{\e}{\f_i}}}{\Conf{\mu'}{\oid_i}}} \qquad 
\begin{array}{l}
\mu'(\oid)=\ConstrCall{\CC}{\oid_1,\ldots,\oid_n}\\
\fields{\CC}=\Field{\CC_1}{\f_1}\ldots\Field{\CC_n}{\f_n}\\
i\in 1..n\\
\end{array}
\\[2ex]
\MetaRule{new}{
\eval{\Conf{\mu_i}{\e_i}}{\Conf{\mu_{i+1}}{\oid_i}}\Space \forall i\in 1..n}{\eval{\Conf{\mu}{\ConstrCall{\CC}{\e_1,\ldots,\e_n}}}{\Conf{\mu'}{\oid}}}\qquad 
{\begin{array}{l}
\mu_1=\mu\\
\mu'=\SubstFun{\mu_{n+1}}{\ConstrCall{\CC}{\oid_1,\ldots,\oid_n}}{\oid}\\
\oid\ \mbox{fresh}
\end{array}
}\\[4ex]
\MetaRule{invk}{
\begin{array}{l}
\eval{\Conf{\mu_i}{\e_i}}{\Conf{\mu_{i+1}}{\oid_i}}\Space \forall i\in 0..n\\
\eval{\Conf{\mu_{n+1}}{\subst{\subst{\subst{\e}{\oid_1}{\x_1}\ldots}{\oid_n}{\x_n}}{\oid_0}{\this}}}{\Conf{\mu'}{\oid}}
\end{array}}{
\eval{\Conf{\mu}{\MethCall{\e_0}{\m}{\e_1,\ldots,\e_n}}}
{\Conf{\mu'}{\oid}}}\qquad 
{\begin{array}{l}
\mu_0=\mu\\
 \mu_1(\oid_0) =\ConstrCall{\CC}{\_}\\
\mbody{\CC}{\m}=\Pair{\x_1\ldots\x_n}{\e}\\
\mu'=\mu_{n+1}
\end{array}
}\\[5ex]
\MetaRule{field-assign}{\eval{\Conf{\mu}{\e}}{\Conf{\mu'}{\oid}}\Space\eval{\Conf{\mu'}{\e'}}{\Conf{\mu''}{\oid'}}}{
\eval{\Conf{\mu}{\FieldAssign{\e}{\f_i}{\e'}}}{\Conf{\UpdateMem{\mu''}{\oid}{i}{\oid'}}{\oid'}}}\qquad 
{\begin{array}{l}
\mu(\oid)=\ConstrCall{\CC}{\oid_1,\ldots,\oid_n}\\
\fields{\CC}=\Field{\CC_1}{\f_1}\ldots\Field{\CC_n}{\f_n}\\
i\in 1..n\\\\
\end{array}}
\end{array}
\end{math}
\HSep
\begin{math}
\begin{array}{l}
\metainlinerule{field-access}
{\eval{\Conf{\mu}{\e}}{\Conf{\mu'}{\oid}}}
{\eval{\Conf{\mu'}{\oid_i}}{\Conf{\mu'}{\oid_i}}}
{\Conf{\mu}{\FieldAccess{\e}{\f_i}}}\\
\hfill\begin{array}{l}
\mu'(\oid)=\ConstrCall{\CC}{\oid_1,\ldots,\oid_n}\\
\fields{\CC}=\Field{\CC_1}{\f_1}\ldots\Field{\CC_n}{\f_n}\\
i\in 1..n\\
\end{array}
\\[2ex]
\metainlinerule{new}
{\eval{\Conf{\mu_i}{\e_i}}{\Conf{\mu_{i+1}}{\oid_i}}\Space \forall i\in 1..n}
{\eval{\Conf{\mu'}{\oid}}{\Conf{\mu'}{\oid}}}
{\Conf{\mu}{\ConstrCall{\CC}{\e_1,\ldots,\e_n}}}\\
\hfill
\begin{array}{l}
\mu_1=\mu\\
\mu'=\SubstFun{\mu_{n+1}}{\ConstrCall{\CC}{\oid_1,\ldots,\oid_n}}{\oid}\\
\oid\ \mbox{fresh}
\end{array}
\\[5ex]
\metainlinerule{invk}
{\eval{\Conf{\mu_i}{\e_i}}{\Conf{\mu_{i+1}}{\oid_i}}\Space \forall i\in 0..n}
{\eval{\Conf{\mu_{n+1}}{\e'}}{\Conf{\mu'}{\oid}}}
{\Conf{\mu}{\MethCall{\e_0}{\m}{\e_1,\ldots,\e_n}}}\\
\hfill
\begin{array}{l}
\mu_0=\mu\\
 \mu_1(\oid_0)=\ConstrCall{\CC}{\_}\\
\mbody{\CC}{\m}=\Pair{\x_1\ldots\x_n}{\e}\\
\e'=\subst{\subst{\subst{\e}{\oid_1}{\x_1}\ldots}{\oid_n}{\x_n}}{\oid_0}{\this}\\
\mu'=\mu_{n+1}
\end{array}
\\[5ex]
\metainlinerule{field-assign}
{\eval{\Conf{\mu}{\e}}{\Conf{\mu'}{\oid}}\Space\eval{\Conf{\mu'}{\e'}}{\Conf{\mu''}{\oid'}}}
{\eval{\Conf{\mu'''}{\oid'}}{\Conf{\mu'''}{\oid'}}}
{\Conf{\mu}{\FieldAssign{\e}{\f_i}{\e'}}}\\
\hfill
\begin{array}{l}
\mu(\oid)=\ConstrCall{\CC}{\oid_1,\ldots,\oid_n}\\
\fields{\CC}=\Field{\CC_1}{\f_1}\ldots\Field{\CC_n}{\f_n}\\
i\in 1..n\\
\mu'''=\UpdateMem{\mu''}{\oid}{i}{\oid'}\\\\
\end{array}
\end{array}
\end{math}
\end{small}
\caption{Imperative $\FJ$: syntax and big-step semantics}\label{fig:FJ-imp}
\end{figure}
\begin{figure}
\begin{small}
\begin{math}
\begin{array}{l}
\\
\MetaRule{t-conf}{\HasType{\Sigma}{\mu(\oid)}{\Sigma(\oid)}\ \forall\oid\in\dom(\mu)\Space \HasType{\Sigma}{\e}{\CC}}{\HasType{\Sigma}{\Conf{\mu}{\e}}{\CC}}
{\dom(\Sigma)=\dom(\mu) }
\\[3ex]
\MetaRule{t-var}{}{\HasTypeMem{\Gamma}{\Sigma}{\x}{\CC}}{\Gamma(\x)=\CC}
\\[3ex]
\MetaRule{t-field-access}{\HasTypeMem{\Gamma}{\Sigma}{\e}{\CC}}{\HasTypeMem{\Gamma}{\Sigma}{\FieldAccess{\e}{\f_i}}{\CC_i}}
{\begin{array}{l}
\fields{\CC}=\Field{\CC_1}{\f_1}\ldots\Field{\CC_n}{\f_n}\\
i\in 1..n
\end{array}
}\\[4ex]
\MetaRule{t-new}{\HasTypeMem{\Gamma}{\Sigma}{\e_i}{\CC_i}\Space \forall i\in 1..n}{\HasTypeMem{\Gamma}{\Sigma}{\ConstrCall{\CC}{\e_1,\ldots,\e_n}}{\CC}}
{\fields{\CC}=\Field{\CC_1}{\f_1}\ldots\Field{\CC_n}{\f_n}
}
\\[3ex]
\MetaRule{t-invk}{\HasTypeMem{\Gamma}{\Sigma}{\e_i}{\CC_i}\Space \forall i\in 0..n}{\HasTypeMem{\Gamma}{\Sigma}{\MethCall{\e_0}{\m}{\e_1,\ldots,\e_n}}{\CC}}
{\mtype{\CC_0}{\m}=\funtype{\CC_1\ldots\CC_n}{\CC}
}
\\[3ex]
\MetaRule{t-field-assign}{
\begin{array}{l}
\HasTypeMem{\Gamma}{\Sigma}{\e}{\CC}\\
\HasTypeMem{\Gamma}{\Sigma}{\e'}{\CC_i}
\end{array}}
{\HasTypeMem{\Gamma}{\Sigma}{\FieldAssign{\e}{\f_i}{\e'}}{\CC_i}}
{\begin{array}{l}
\fields{\CC}=\Field{\CC_1}{\f_1}\ldots\Field{\CC_n}{\f_n}\\
i\in 1..n
\end{array}
}\\[4ex]
\MetaRule{t-oid}{}{\HasTypeMem{\Gamma}{\Sigma}{\oid}{\CC}}{\Sigma(\oid)=\CC}\BigSpace\MetaRule{t-sub}{\HasTypeMem{\Gamma}{\Sigma}{\e}{\CC}}{\HasTypeMem{\Gamma}{\Sigma}{\e}{\CC'}}{
\CC\leqfj\CC'}
\end{array}
\end{math}
\end{small}
\caption{Imperative $\FJ$: typing rules}\label{fig:FJ-impT}
\end{figure}
\begin{lemma}[Inversion]\label{lem:ilimp}\
\begin{enumerate}
\item\label{lem:ilimp:0} If $\HasTypeMem{\Gamma}{\Sigma}{\Conf{\mu}{\e}}{\mbox{\em\CC}}$, then $\HasTypeMem{\Gamma}{\Sigma}{\mu(\oid)}{\Sigma(\oid)}$ for all $\oid\in\dom(\mu)$ and  $\HasType{\Sigma}{\e}{\mbox{\em\CC}}$ and $\dom(\Sigma)=\dom(\mu)$.
\item\label{lem:ilimp:1} If  $\HasTypeMem{\Gamma}{\Sigma}{\x}{\mbox{\em\CC}}$, then $\Gamma(\x)\leqfj{\mbox{\em\CC}}$.
\item\label{lem:ilimp:2} If  $\HasTypeMem{\Gamma}{\Sigma}{\FieldAccess{\e}{\f_i}}{\mbox{\em\CC}}$, then $\HasTypeMem{\Gamma}{\Sigma}{\e}{\mbox{\em\DD}}$ and {\em $\fields{\DD}=\Field{\CC_1}{\f_1}\ldots\Field{\CC_n}{\f_n}$} and ${\mbox{\em\CC}}_i\leqfj {\mbox{\em\CC}}$ where 
$i\in 1..n$.
\item\label{lem:ilimp:3} If  $\HasTypeMem{\Gamma}{\Sigma}{\sf new~{\mbox{\em\CC}}(\e_1,\ldots,\e_n)}{\mbox{\em\DD}}$, then ${\mbox{\em\CC}}\leqfj{\mbox{\em\DD}}$ and {\em $\fields{\CC}=\Field{\CC_1}{\f_1}\ldots\Field{\CC_n}{\f_n}$} and $\HasTypeMem{\Gamma}{\Sigma}{\e_i}{{\mbox{\em\CC}}_i}$ for all $i\in 1..n$.
\item\label{lem:ilimp:4} If  $\HasTypeMem{\Gamma}{\Sigma}{\MethCall{\e_0}{\m}{\e_1,\ldots,\e_n}}{\mbox{\em\CC}}$, then  
$\HasTypeMem{\Gamma}{\Sigma}{\e_i}{{\mbox{\em\CC}}_i}$ for all $i\in 0..n$ and  {\em $\mtype{\CC_0}{\m}=\funtype{\CC_1\ldots\CC_n}{\DD}$} with ${\mbox{\em\DD}}\leqfj{\mbox{\em\CC}}$.
\item\label{lem:ilimp:5} If  $\HasTypeMem{\Gamma}{\Sigma}{\FieldAssign{\e}{\f_i}{\e'}}{\mbox{\em\CC}}$, then $\HasTypeMem{\Gamma}{\Sigma}{\e}{\mbox{\em\DD}}$ and {\em $\fields{\DD}=\Field{\CC_1}{\f_1}\ldots\Field{\CC_n}{\f_n}$} and $\HasTypeMem{\Gamma}{\Sigma}{\e'}{{\mbox{\em\CC}}_i}$  and ${\mbox{\em\CC}}_i\leqfj{\mbox{\em\CC}}$.
\item\label{lem:ilimp:6} If  $\HasTypeMem{\Gamma}{\Sigma}{\oid}{\mbox{\em\CC}}$, then $\Sigma(\oid)\leqfj{\mbox{\em\CC}}$.

\end{enumerate}
\end{lemma}

\begin{lemma}[Substitution]\label{lem:sosimp}
If $\HasTypeMem{\SubstFun\Gamma{{\mbox{\em\CC}}'}\x}{\Sigma} \e {\mbox{\em\CC}}$ and $\HasTypeMem{\Gamma}{\Sigma}{\e'} {{\mbox{\em\CC}}'}$, then $\HasTypeMem{\Gamma}{\Sigma}{\subst\e{\e'}\x} {\mbox{\em\CC}}$.
\end{lemma}

\noindent
We can prove the soundness of the indexed predicate $\Pred_\CC$ defined by: $\Pred_\Pair{\Sigma}{\CC}(\Conf{\mu}{\e})$ holds if $\HasType{\Sigma'}{\Conf{\mu}{\e}}{\CC}$ for some $\Sigma'$ such that $\Sigma\subseteq\Sigma'$.  The type assignment $\Sigma'$
is needed, since memory can grow during evaluation. 

\begin{theorem}[Soundness]\label{thm:sdimp}
The big-step semantics of \refToFigure{FJ-imp} and the indexed predicate $\Pred_\Pair{\Sigma}{\CC}$ satisfy the conditions {\em \refToSound{preservation}}, {\em \refToSound{progress-ex}} and {\em \refToSound{progress-all}} of \refToSect{sc}.
 \end{theorem}
\begin{proof}Condition \refToSound{preservation}.
The proof is by cases on instantiations of meta-rules. 
\\
Rule \rn{field-access}. \refToLemItem{ilimp}{0} applied to  $\HasType{\Sigma}{\Conf{\mu}{\FieldAccess{\e}{\f_i}}}{\CC}$ implies $\HasType{\Sigma}{\mu(\oid)}{\Sigma(\oid)}$ for all $\oid\in\dom(\mu)$ and  $\HasType{\Sigma}{\FieldAccess{\e}{\f_i}}{\CC}$ and $\dom(\Sigma)=\dom(\mu)$.
\refToLemItem{ilimp}{2} applied to $\HasType{\Sigma}{\FieldAccess{\e}{\f_i}}{\CC}$ implies $\HasType{\Sigma}{\e}{\DD}$ and $\fields{\DD}=\Field{\CC_1}{\f_1}\ldots\Field{\CC_n}{\f_n}$ and $\CC_i\leqfj \CC$ where 
$i\in 1..n$.  Since $\eval{\Conf{\mu}{\e}}{\Conf{\mu'}{\oid}}$ is a premise we assume $\HasType{\Sigma'}{\Conf{\mu'}{\oid}}\DD$ with $\Sigma\subseteq\Sigma'$. 
\refToLemItem{ilimp}{0} and \refToLemItem{ilimp}{6} imply
$\Sigma'(\oid)\leqfj\DD$. \refToLemItem{ilimp}{3} allows us to get 
 $\mu'(\oid)=\ConstrCall{\CC'}{\oid_1,\ldots\oid_m}$ with $n\leq m$ and $\CC'\leqfj\DD$ and  $\HasType{\Sigma'}{\oid_i}{\CC_i}$. So we conclude $\HasType{\Sigma'}{\Conf{\mu'}{\oid_i}}\CC$ by rules \rn{t-sub} and \rn{t-conf}.\\
Rule \rn{new}. \refToLemItem{ilimp}{0} applied to  $\HasType{\Sigma}{\Conf{\mu}{\ConstrCall{\CC}{\e_1,\ldots,\e_n}}}{\DD}$ implies $\HasType{\Sigma}{\mu(\oid)}{\Sigma(\oid)}$ for all $\oid\in\dom(\mu)$ and  $\HasType{\Sigma}{\ConstrCall{\CC}{\e_1,\ldots,\e_n}}{\DD}$ and $\dom(\Sigma)=\dom(\mu)$.
\refToLemItem{ilimp}{3} applied to $\HasType{\Sigma}{\ConstrCall{\CC}{\e_1,\ldots,\e_n}}{\DD}$ implies $\CC\leqfj\DD$ and $\fields{\CC}=\Field{\CC_1}{\f_1}\ldots\Field{\CC_n}{\f_n}$ and $\HasType{\Sigma}{\e_i}{\CC_i}$ for all $i\in 1..n$. Since $\eval{\Conf{\mu}{\e_i}}{\Conf{\mu_{i+1}}{\oid_i}}$ is a premise we assume 
$\HasType{\Sigma_i}{\Conf{\mu_{i+1}}{\oid_i}}{\CC_i}$ for all $i\in 1..n$ with $\Sigma\subseteq\Sigma_1\subseteq\cdots\subseteq\Sigma_n$.
\refToLemItem{ilimp}{0} and \refToLemItem{ilimp}{6} imply
$\Sigma_i(\oid_i)\leqfj\CC_i$  for all $i\in 1..n$. 
Using rules \rn{t-oid}, \rn{t-new} and \rn{t-sub} we derive $\HasType{\Sigma_n}{\ConstrCall{\CC}{\oid_1,\ldots,\oid_n}}{\DD}$. We then conclude $\HasType{\Sigma_n,\oid:\DD}{\Conf{\mu_{n+1}}\oid}{\DD}$ by rules \rn{t-oid}and \rn{t-conf}.\\
Rule \rn{invk}. \refToLemItem{ilimp}{0} applied to  $\HasType{\Sigma_0}{\Conf{\mu_0}{\MethCall{\e_0}{\m}{\e_1,\ldots,\e_n}}}{\CC}$ implies $\HasType{\Sigma_0}{\mu_0(\oid)}{\Sigma_0(\oid)}$ for all $\oid\in\dom(\mu_0)$ and  $\HasType{\Sigma_0}{\MethCall{\e_0}{\m}{\e_1,\ldots,\e_n}}{\CC}$ and $\dom(\Sigma_0)=\dom(\mu_0)$.
\refToLemItem{ilimp}{4} applied to $\HasType{\Sigma_0}{\MethCall{\e_0}{\m}{\e_1,\ldots,\e_n}}{\CC}$ implies  
$\HasType{\Sigma_i}{\e_i}{\CC_i}$ for all $i\in 0..n$ and  $\mtype{\CC_0}{\m}=\funtype{\CC_1\ldots\CC_n}{\DD}$ with $\DD\leqfj\CC$. Since $\eval{\Conf{\mu_i}{\e_i}}{\Conf{\mu_{i+1}}{\oid_i}}$ is a premise we assume $\HasType{\Sigma_i}{\Conf{\mu_{i+1}}{\oid_i}}{\CC_i}$ for all $i\in 0..n$ with $\Sigma_0\subseteq\cdots\subseteq\Sigma_n$. \refToLemItem{ilimp}{0} gives $\HasType{\Sigma_i}{\oid_i}{\CC_i}$ for all $i\in 0..n$. 
The typing of the class table implies $\HasType{\x_1{:}\CC_1,\ldots,\x_n{:}\CC_n,\this{:}\CC_0}{\e}{\DD}$.
\refToLem{sosimp} gives $\HasType{\Sigma_n}{\e'}\DD$ where $\e'=\subst{\subst{\subst{\e}{\oid_1}{\x_1}\ldots}{\oid_n}{\x_n}}{\oid_0}{\this}$. Using rules \rn{t-sub} and 
\rn{t-conf} we derive $\HasType{\Sigma_n}{\Conf{\mu_{n+1}}{\e'}}\CC$.
Since $\eval{\Conf{\mu_{n+1}}{\e'}}{\Conf{\mu'}{\oid}}$ is  a premise we conclude $\HasType{\Sigma'}{\Conf{\mu'}{\oid}}\CC$ with $\Sigma_n\subseteq\Sigma'$.\\
Rule \rn{field-assign}. \refToLemItem{ilimp}{0} applied to  $\HasType{\Sigma}{\Conf{\mu}{\FieldAssign{\e}{\f_i}{\e'}}}{\CC}$ implies $\HasType{\Sigma}{\mu(\oid)}{\Sigma(\oid)}$ for all $\oid\in\dom(\mu)$ and  $\HasType{\Sigma}{\FieldAssign{\e}{\f_i}{\e'}}{\CC}$ and $\dom(\Sigma)=\dom(\mu)$. \refToLemItem{ilimp}{5} applied to $\HasType{\Sigma}{\FieldAssign{\e}{\f_i}{\e'}}{\CC}$ implies $\HasType{\Sigma}{\e}{\DD}$ and $\fields{\DD}=\Field{\CC_1}{\f_1}\ldots\Field{\CC_n}{\f_n}$ and $\HasType{\Sigma}{\e'}{\CC_i}$  and $\CC_i\leqfj\CC$. Since $\eval{\Conf{\mu}{\e}}{\Conf{\mu'}{\oid}}$ and $\eval{\Conf{\mu'}{\e'}}{\Conf{\mu''}{\oid'}}$ are premises we assume  $\HasType{\Sigma'}{\Conf{\mu'}{\oid}}\DD$  and $\HasType{\Sigma''}{\Conf{\mu''}{\oid'}}\CC_i$, with $\Sigma\subseteq\Sigma'\subseteq\Sigma''$. 
Notice that $\mu''(\oid)$ and $\UpdateMem{\mu''}{\oid}{i}{\oid'}(\oid)$ have the same types for all $\oid$ by construction. We conclude $\HasType{\Sigma''}{\Conf{\UpdateMem{\mu''}{\oid}{i}{\oid'}}{\oid'}}\CC_i$.

Condition \refToSound{progress-ex}. All the closed expressions which are not values appear as conclusions in the reduction rules.

Condition \refToSound{progress-all}.  Rule \rn{field-access} requires that $\Conf{\mu}{\e}$ reduces to $\Conf{\mu'}{\oid}$ such that $\mu'(\oid)=\ConstrCall{\DD}{\oid_1,\ldots,\oid_m}$, $\fields{\DD}=\Field{\CC_1}{\f_1}\ldots\Field{\CC_m}{\f_m}$, and $i\in 1..m$. Since the configuration in the consequence is well-typed in $\Sigma$, by \refToLemItem{ilimp}{0} and \refToLemItem{ilimp}{2} we have $\HasType{\Sigma}{\e}{\CC}$ and $\fields{\CC}=\Field{\CC_1}{\f_1}\ldots\Field{\CC_n}{\f_n}$, and 
$i\in 1..n$.  The validity of condition \refToSound{preservation} (which assures type preservation by \refToLem{sr}) implies that $\HasType{\Sigma'}{\Conf{\mu'}{\oid}}{\DD}$ with $\Sigma\subseteq\Sigma'$ and $\DD$ subclass of $\CC$, hence, by rules \rn{t-conf} and \rn{t-oid}, $\HasType{\Sigma'}{\mu'(\oid)}{\DD}$, and the well-formedness of the class table implies that $\fields{\DD}=\Field{\CC_1}{\f_1}\ldots\Field{\CC_m}{\f_m}$ with $n\leq m$, hence $i\in 1..m$.\\
Rules \rn{invk} and \rn{field-assign} require that the expression in the first premise reduces to an object identifier for which the side-conditions hold, and this can be proved analogously. Other meta-variables for results can be freely instantiated.\\ 

\end{proof}

\section{The partial evaluation construction} \label{sect:partial-eval}\label{sect:PEV}
In this section, our aim is to provide a \emph{formal} justification that the constructions in \refToSect{constructions} are correct. 
For instance, for the $\Wrong$ semantics we would like to be sure that \emph{all} the cases are covered. 
To this end, we define  a \emph{third construction}, dubbed $\PEval$ for ``partial evaluation'', which makes explicit the \emph{computations} of a big-step semantics, intended as the sequences of execution steps of the naturally associated evaluation algorithm. 
Formally, we obtain a reduction relation on approximated proof trees, so non-termination and stuck computation are distinguished, and both soundness-must and soundness-may can be expressed.

 To this end, first of all we introduce a special result $\Unknown$, so that a judgment $\eval{\conf}{\Unknown}$ (called \emph{incomplete}, whereas a judgment in $\RuleSet$ is \emph{complete})  means that  the evaluation of $\conf$ is not completed yet.  
Analogously to the previous constructions, we define an augmented set of rules $\UnRuleSet$ for the judgment extended with $\Unknown$:

\begin{description}
\item [$\Unknown$ introduction rules] These rules derive $\Unknown$ whenever a rule is partially applied:
for each rule $\rho \equiv \inlinerule{\judg_1\ldots\judg_n}{\judg_{n+1}}{\conf}$ in $\RuleSet$, index $i \in 1..n+1$, and result $\res \in \ResSet$,  we define the rule $\unknownrule{\rho}{i}{\res}$ as\\
\centerline{$
\Rule{\judg_1 \Space \ldots \Space \judg_{i-1} \Space \eval{\ConfSet(\judg_i)}{\res} }{ \eval{\conf}{\Unknown} }
$}

\noindent
We also add an axiom $\Rule{}{\eval{\conf}{\Unknown}}$ for each configuration $\conf \in \ConfSet$.

\item [$\Unknown$ propagation rules] These rules propagate $\Unknown$ analogously to those for divergence and $\Wrong$ propagation:
for each $\rho \equiv \inlinerule{\judg_1\ldots\judg_n}{\judg_{n+1}}{\conf}$ in $\RuleSet$, and index $i \in 1..n+1$, we add the rule $\proprule{\rho}{i}{\Unknown}$ as follows:\\
\centerline{$
\Rule{
	\judg_1 \Space \ldots \Space \judg_{i-1}
	\Space 
	\eval{\ConfSet(\judg_i)}{\Unknown} 
}{ \eval{\conf}{\Unknown} }
$}
\end{description}
Finally, we consider the set $\TreeSet$ of the (finite) proof trees $\tr$ in $\UnRuleSet$. 
Each $\tr$ can be thought as a \emph{partial proof} or \emph{partial evaluation} of the root configuration. 
In particular, we say it is \emph{complete} if it is a proof tree in $\RuleSet$ (that is, it only contains complete judgments), \emph{incomplete} otherwise. 
We define a reduction relation $\ruleAr$   on $\TreeSet$ 
such that, starting from the initial proof tree $\Rule{}{\eval{\conf}{\Unknown}}$, we derive a sequence where, intuitively, at each step we detail the proof (evaluation). 
In this way, a sequence ending with a complete tree  $\Rule{\ldots}{\eval{\conf}{\res}}$ models terminating computation, whereas an infinite sequence (tending to an infinite proof tree) models divergence, and a stuck sequence models a stuck computation. 

The one-step reduction relation $\ruleAr$ on $\TreeSet$ is inductively defined by the rules in \refToFigure{?-rules}.  In this figure  $\#\rho$ denotes the number of premises of $\rho$, and $\rt{\tr}$ the root of $\tr$. We set $\ResSet_\Unknown(\eval\conf u)=u$ where $u\in\ResSet\cup\{\Unknown\}$. Finally, $\sim_i$ is the \emph{equivalence up-to an index} of rules, introduced at the beginning of \refToSect{wrong}.
\begin{figure}[t]
\begin{small}
\begin{math}
\begin{array}{l}
\MetaRuleTree{\res_\Unknown}{\Rule{}{\eval{\res}{\Unknown}}}{\res}{\Rule{}{\eval{\res}{\res}}}{}
\BigSpace
\MetaRuleTree{\conf_\Unknown}{\Rule{}{\eval{\conf}{\Unknown}}}{\proprule{\rho}{1}{\Unknown}}{\Rule{\eval{\conf'}{\Unknown}}{\eval{\conf}{\Unknown}}}{
\ConfSet(\rho) = \conf \\
 \ConfSet(\rho, 1) = \conf'
} \\[3ex]
\MetaRuleTree{\unknownrule{\rho}{i}{\res}}{\Rule{\tr_1\ \ldots\ \tr_i}{\eval{\conf}{\Unknown}}}{\rho'}{\Rule{\tr_1\ \ldots\ \tr_i}{\eval{\conf}{\res}}}{
\rho' \sim_i \rho  \\
\ResSet(\rho', i) = \res \\
\#\rho' = i}\\[3ex]
\MetaRuleTree{\unknownrule{\rho}{i}{\res}}{\Rule{\tr_1\ \ldots\ \tr_i}{\eval{\conf}{\Unknown}}}{\proprule{\rho'}{i+1}{\Unknown}}{\Rule{\tr_1\ \ldots \ \tr_i\ \eval{\conf'}{\Unknown}}{\eval{\conf}{\Unknown}}}{
\rho' \sim_i \rho \\
\ResSet(\rho', i) = \res \\
 \ConfSet(\rho', i+1) = \conf' 
} \\[3ex]
\MetaRuleTree{\proprule{\rho}{i}{\Unknown}}{\Rule{\tr_1\ \ldots\ \tr_i}{\eval{\conf}{\Unknown}}}{\proprule{\rho}{i}{\Unknown}}{\Rule{\tr_1\ \ldots\ \tr_{i-1}\ \tr'_i}{\eval{\conf}{\Unknown}}}{
\tr_i \ruleAr \tr'_i \\
\ResSet_\Unknown(\rt{\tr'_i}) = \Unknown
} \\[3ex]
\MetaRuleTree{\proprule{\rho}{i}{\Unknown}}{\Rule{\tr_1\ \ldots\ \tr_i}{\eval{\conf}{\Unknown}}}{\unknownrule{\rho}{i}{\res}}{\Rule{\tr_1\ \ldots\ \tr_{i-1}\ \tr'_i}{\eval{\conf}{\Unknown}}}{
\tr_i \ruleAr \tr'_i \\
\ResSet_\Unknown(\rt{\tr'_i}) = \res
}
\end{array}
\end{math}
\end{small}
\caption{Reduction relation on $\TreeSet$}\label{fig:?-rules}
\end{figure}
As said above, each reduction step makes ``less incomplete'' the proof tree. Notably, reduction rules apply to nodes with consequence $\eval{\conf}{\Unknown}$, whereas subtrees with root $\eval{\conf}{\res}$ represent terminated evaluation.  
In detail:
\begin{itemize}
\item If the last applied rule is an axiom, and the configuration is a result $\res$, then we can evaluate $\res$ to itself.   
Otherwise, we have to find a rule $\rho$ with $\conf$ in the consequence and start evaluating the first premise of such rule.
\item If the last applied rule is $\unknownrule{\rho}{i}{\res}$, then all subtrees are complete, hence, to continue the evaluation, we have to find another rule $\rho'$, having, for each $k \in 1..i$, as $k$-th premise the root of $\tr_k$. 
Then there are two possibilities: if there is an $i+1$-th premise, we start evaluating it, otherwise, we propagate to the conclusion the result $\res$ of $\tr_i$.
\item If the last applied rule is a propagation rule $\proprule{\rho}{i}{\Unknown}$, then we simply propagate the step made by $\tr_i$. 
\end{itemize}
In \refToFigure{ex-pev-reduction} we report an example of $\PEval$ reduction. 
\begin{figure}
\begin{small}
\begin{align*}
\Rule{}{\eval{(\Lam{\x}\x)\appop n}{\Unknown}}  
&\ruleAr \Rule{\eval{\Lam{\x}\x}{\Unknown}}{\eval{(\Lam{\x}\x)\appop n}{\Unknown}}
 \ruleAr \Rule{\eval{\Lam{\x}\x}{\Lam{\x}\x}}{\eval{(\Lam{\x}\x)\appop n}{\Unknown}}
 \ruleAr \Rule{\eval{\Lam{\x}\x}{\Lam{\x}\x}\quad \eval n \Unknown}{\eval{(\Lam{\x}\x)\appop n}{\Unknown}} \\
&\ruleAr \Rule{\eval{\Lam{\x}\x}{\Lam{\x}\x}\quad \eval n n}{\eval{(\Lam{\x}\x)\appop n}{\Unknown}}
 \ruleAr \Rule{\eval{\Lam{\x}\x}{\Lam{\x}\x}\quad \eval n n \quad \eval n \Unknown}{\eval{(\Lam{\x}\x)\appop n}{\Unknown}} \\
&\ruleAr \Rule{\eval{\Lam{\x}\x}{\Lam{\x}\x}\quad \eval n n \quad \eval n n}{\eval{(\Lam{\x}\x)\appop n}{\Unknown}}
 \ruleAr \Rule{\eval{\Lam{\x}\x}{\Lam{\x}\x}\quad \eval n n \quad \eval n n}{\eval{(\Lam{\x}\x)\appop n}{n}}
\end{align*}
\end{small}
\caption{The evaluation in $\PEval$ of $(\Lam{\x}\x) \appop n$. } \label{fig:ex-pev-reduction} 
\end{figure}

The remaining of this section is split into three parts. In \refToSect{ppt} we show properties  of proof trees starting from a formal account of them. The coherence of our approach through the equivalence of the three constructions is the content of  \refToSect{etwp}. Lastly \refToSect{spev} discusses soundness of $\PEval$ semantics.

\subsection{Properties of proof trees}\label{sect:ppt}

We give a formal account of proof trees, which is useful to state and to prove following technical results. The account
follows \cite{Courcelle83,Dagnino19}, but it is adjusted to our specific setting.

Set $\NPos$ the set of positive natural numbers and $\LabelSet$ a set of labels. 
A \emph{tree} labelled in $\LabelSet$ is a partial function $\fun{\tr}{\List{\NPos}}{\LabelSet}$ such that 
$\dom(\tr)$ is not empty, and, for each $\alpha \in \List{\NPos}$ and $n \in \NPos$, 
if $\alpha n \in \dom(\tr)$ then 
$\alpha \in \dom(\tr)$ and, for all $k \le n$, $\alpha k \in \dom(\tr)$. 
Given a tree $\tr$ and $\alpha \in \dom(\tr)$, set $\trbr_\tr(\alpha)=\max\{n \in \N \mid \alpha n \in \dom(\tr)\}$ the \emph{branching} of $\tr$ at $\alpha$, and 
 $\subtree{\tr}{\alpha}$ the \emph{subtree} of $\tr$ rooted at $\alpha$, that is, $\subtree{\tr}{\alpha}(\beta) = \tr(\alpha\beta)$. In particular, 
 $\tr(\EList)=\rt\tr $
 is the \mbox{\emph{root} of $\tr$. }
 Finally, we write $\Rule{\tr_1\Space \ldots \Space \tr_n}{\judg}$ for the tree $\tr$ defined by $\tr(\EList)=\judg$, and $\tr(i\alpha) = \tr_{i}(\alpha)$  \mbox{for all $i \in 1..n$.}
 
Assume now that labels in $\LabelSet$ are the judgments of an inference system $\GenRuleSet$ (where premises of rules are sorted). 
Then, a tree labelled in $\LabelSet$ is a proof tree in $\GenRuleSet$
if, for each $\alpha \in \dom(\tr)$, there is a rule $\Rule{\tr(\alpha 1)\Space\ldots\Space\tr(\alpha\trbr_\tr(\alpha))}{\tr(\alpha)} \in \GenRuleSet$.

The following proposition assures two key properties of proof trees in $\UnRuleSet$. First, if there is some $\Unknown$, then it is propagated to ancestor nodes. 
Second, for each level of the tree there is at most one $\Unknown$. We set $\Len{\alpha}$ the length of $\alpha \in \List{\NPos}$.
\begin{proposition} \label{prop:tree-unknown}
Let $\tr$ be a proof tree in $\UnRuleSet$, then the following hold:
\begin{enumerate}
\item\label{prop:tree-unknown:1} for all $\alpha n \in \dom(\tr)$, if $\ResSet_\Unknown(\tr(\alpha n)) = \Unknown$ then $\ResSet_\Unknown(\tr(\alpha)) = \Unknown$.
\item\label{prop:tree-unknown:2} for all $n \in \N$, there is at most one $\alpha \in \dom(\tr)$ with $\Len{\alpha} = n$ such that $\ResSet_\Unknown(\tr(\alpha)) = \Unknown$.
\end{enumerate}
\end{proposition}
\begin{proof}
To prove 1, it is enough to note that the only rules having a premise $\judg$ with $\ResSet_\Unknown(\judg) = \Unknown$ are $\Unknown$-propagation rules, which also have conclusion $\judg'$ with $\ResSet_\Unknown(\judg') = \Unknown$; hence the thesis is immediate.
To prove 2, we proceed by induction on $n$. 
For $n = 0$, there is only one  $\alpha \in \List{\NPos}$ with $\Len{\alpha} = 0$, hence the thesis is trivial.
Consider $\alpha = \alpha' k \in \dom(\tr)$ with $\Len{\alpha} = n+1$; 
if $\ResSet_\Unknown(\tr(\alpha)) = \Unknown$, then, by point 1, $\ResSet_\Unknown(\tr(\alpha')) = \Unknown$, and $\alpha'$ is unique by induction hypothesis. 
Therefore, another node $\beta \in \dom(\tr)$, with $\Len{\beta} = n+1$ and $\ResSet_\Unknown(\tr(\beta)) = \Unknown$, must satisfy $\beta = \alpha' h$ for some $h \in \NPos$;
hence, since $\tr$ is a proof tree, $\tr(\alpha)$ and $\tr(\beta)$ are two premises of the same rule with $\Unknown$ as result, thus they must coincide, since the rules have at most one premise \mbox{with $\Unknown$.}
\end{proof}

\begin{corollary} \label{cor:complete-tree}
Let $\tr$ be a finite proof tree in $\UnRuleSet$, then 
$\ResSet_\Unknown(\rt\tr) \in \ResSet$ implies $\tr$ is complete. 
\end{corollary}

As said above, the definition of $\ruleAr$ given in \refToFigure{?-rules} nicely models a ``small-step'' version of an interpreter driven by the big-step rules. In other words, the one-step reduction relation on $\TreeSet$ specifies an algorithm of incremental proof/ evaluation.\footnote{Non-determinism can only be caused by intrinsic non-determinism of the big-step semantics, if any.} However, to carry out some proofs on \PEval\ semantics, it is convenient to consider a more abstract relation.

The relation $\TreeOrd$ on (finite or infinite) trees\footnote{A slight variation of analogous relations is considered in \cite{Courcelle83,Dagnino19}.} labelled by semantic judgements is defined by:
\begin{quote}
\begin{tabular}{ll}
$\tr \TreeOrd \tr'$ if 
& $\dom(\tr) \subseteq \dom(\tr')$\\
&for all $\alpha \in \dom(\tr)$, $\ConfSet(\tr(\alpha)) = \ConfSet(\tr'(\alpha))$ and $\ResSet_\Unknown(\tr(\alpha)) \in \ResSet$
implies $\subtree{\tr}{\alpha} = \subtree{\tr'}{\alpha}$.
\end{tabular}
\end{quote}
Intuitively, $\tr\TreeOrd\tr'$ means that $\tr'$ can be obtained from $\tr$ by adding new branches or replacing some $\Unknown$s with results. 
We use $\StrTreeOrd$ for the strict version of $\TreeOrd$. 
It is easy to check that $\TreeOrd$ is a partial order and, if $\tr\TreeOrd\tr'$, then, for all $\alpha \in \dom(\tr)$, $\subtree{\tr}{\alpha} \TreeOrd \subtree{\tr'}{\alpha}$. 
{The following proposition shows some, less trivial, properties of $\TreeOrd$. }

\begin{proposition} \label{prop:tree-ord}
The following properties hold:
\begin{enumerate}
\item\label{prop:tree-ord:1} for all trees $\tr$ and $\tr'$, if $\tr \TreeOrd \tr'$ and $\ResSet_\Unknown(\rt\tr) \in \ResSet$, then $\tr = \tr'$
\item\label{prop:tree-ord:2} for each increasing sequence $(\tr_i)_{i \in \N}$ of trees, there is a least upper bound $\tr = \Treelub \tr_n$.
\end{enumerate}
\end{proposition}
\begin{proof}
Point 1 is immediate by definition of $\TreeOrd$.
To prove point 2, first note that, since for all $n \in \N$, $\tr_n \TreeOrd \tr_{n+1}$, 
for all $\alpha \in \List{\NPos}$ we have that, for all $n \in \N$, if $\tr_n(\alpha)$ is defined, then, 
for all $k \ge n$, $\ConfSet(\tr_k(\alpha)) = \ConfSet(\tr_n(\alpha))$, and, 
if $\ResSet_\Unknown(\tr_n(\alpha)) \in \ResSet$, then $\tr_k(\alpha) = \tr_n(\alpha)$. 
Hence, for all $n \in \N$, there are only three possibilities for $\tr_n(\alpha)$: it is either undefined, or equal to $\eval{\conf}{\Unknown}$, or equal to $\eval{\conf}{\res}$, where the configuration is the same. 
Let us denote by $k_\alpha$ the least index where $\tr_n(\alpha)$ is most defined, hence, for all $n \ge k_\alpha$, we have that $\tr_n(\alpha) = \tr_{k_\alpha}(\alpha)$.
Then, consider a tree $\tr$ defined by  $\tr(\alpha) = \tr_{k_\alpha}(\alpha)$. 
It is easy to check that $\dom(\tr) = \bigcup_{n \in \N} \dom(\tr_n)$. 
We now check that, for all $n \in \N$, $\tr_n \TreeOrd \tr$. 
For all $\alpha \in \dom(\tr_n)$, we have $\alpha \in \dom(\tr)$ and we distinguish two cases:
\begin{itemize}
\item if $\tr_n(\alpha) = \eval{\conf}{\Unknown}$, then $k_\alpha \ge n$, hence, since $\tr_n \TreeOrd \tr_{k_\alpha}$, we get $\ConfSet(\tr(\alpha))=\ConfSet(\tr_{k_\alpha}(\alpha)) = \ConfSet(\tr_n(\alpha)) = \conf$
\item if $\tr_n(\alpha) = \eval{\conf}{\res}$, then $k_\alpha \le n$, hence, since $\tr_{k_\alpha} \TreeOrd \tr_n$, we get $\ConfSet(\tr(\alpha)) = \ConfSet(\tr_{k_\alpha}(\alpha))  = \ConfSet(\tr_n(\alpha)) = \conf$, thus we have only to check that $\subtree{\tr_n}{\alpha} = \subtree{\tr}{\alpha}$, that is easy,
because, for all $\beta \in \dom(\subtree{\tr_n}{\alpha})$, we have $\subtree{\tr_n}{\alpha}(\beta) = \tr_n(\alpha\beta) = \eval{\conf'}{\res'}$, hence $k_{\alpha\beta} \ge n$, hence $\subtree{\tr}{\alpha}(\beta) = \tr(\alpha\beta) = \tr_{k_{\alpha\beta}}(\alpha\beta) = \tr_n(\alpha\beta)$, as needed. 
\end{itemize}
This proves that $\tr$ is an upper bound of the sequence, we have still to prove that it is the least one. 
To this end, let $\tr'$ be an upper bound of the sequence: we have to show that $\tr \TreeOrd \tr'$. 
Since $\tr'$ is an upper bound, for all $n \in \N$ we have $\dom(\tr_n) \subseteq \dom(\tr')$, hence $\dom(\tr) \subseteq \dom(\tr')$, and, especially, for all $\alpha \in \List{\NPos}$ we have $\tr_{k_\alpha} \TreeOrd \tr'$. 
Hence, for all $\alpha \in \dom(\tr)$, we have $\ConfSet(\tr(\alpha)) = \ConfSet(\tr_{k_\alpha}(\alpha)) = \ConfSet(\tr'(\alpha))$, and, if $\ResSet_\Unknown(\tr(\alpha)) = \res$, since $\tr_{k_\alpha} \TreeOrd \tr$ and $\tr_{k_\alpha} \TreeOrd \tr'$, we have $\subtree{\tr_{k_\alpha}}{\alpha}  = \subtree{\tr}{\alpha}$ and $\subtree{\tr_{k_\alpha}}{\alpha} =  \subtree{\tr'}{\alpha}$, hence $\subtree{\tr}{\alpha} = \subtree{\tr'}{\alpha}$ as needed.  
\end{proof}

Finally, the next proposition formally proves that $\TreeOrd$ is an abstraction of $\ruleAr$. 
\begin{proposition} \label{prop:tree-ord-arrow}
For all $\tr,\tr'\in\TreeSet$, the following hold:
\begin{enumerate}
\item\label{prop:tree-ord-arrow:1} if $\tr\ruleAr \tr'$ then $\tr \StrTreeOrd \tr'$
\item\label{prop:tree-ord-arrow:2} if $\tr \TreeOrd \tr'$ then $\tr \ruleArStar \tr'$.
\end{enumerate}
\end{proposition}
\begin{proof}
Point 1 can be easily proved by induction on the definition of $\ruleAr$. 
To prove point 2, we proceed by induction on $\tr'$.
We can assume $\ResSet_\Unknown(\rt\tr) = \Unknown$, since in the other case, by \refToProp{tree-ord}(\ref{prop:tree-ord:1}), we have $\tr = \tr'$, hence the thesis is trivial. 
We can also assume $\ResSet_\Unknown(\rt{\tr'}) = \Unknown$, since, if $\tr' = \Rule{\tr'_1\Space\ldots\Space\tr'_n}{\eval{\conf}{\res}}$, then we always have $\tr'' = \Rule{\tr'_1\Space\ldots\Space\tr'_n}{\eval{\conf}{\Unknown}} \ruleAr \tr'$ and $\tr \TreeOrd \tr''$. 
Now, if $\tr' = \Rule{}{\eval{\conf}{\Unknown}}$ (base case), then, since $\dom(\tr) \subseteq \dom(\tr')$ and $\ConfSet(\rt\tr) = \ConfSet(\rt{\tr'})$ by definition of $\TreeOrd$, we have $\tr = \tr'$, hence the thesis is trivial. 

Let us assume $\tr = \Rule{\tr_1\Space\ldots\tr_k}{\eval{\conf}{\Unknown}}$ and $\tr' = \Rule{\tr'_1\Space\ldots\Space\tr'_i}{\eval{\conf}{\Unknown}}$, with, necessarily,  $k \le i$ by definition of $\TreeOrd$. 
By \refToPropItem{tree-unknown}{2}, at most $\tr_k$ is incomplete, hence, for all $h < k$, $\tr_h$ is complete, hence $\ResSet_\Unknown(\rt{\tr_h}) \in \ResSet$, thus, by definition of $\TreeOrd$, we have $\tr_h = \tr'_h$. 
We now show, concluding the proof, by induction on $i-k$, that $\tr\ruleArStar \tr'$. 
Since $\tr_k \TreeOrd \tr'_k$, by $HI$, 
we get $\tr_k \ruleArStar \tr'_k$, hence $\tr \ruleArStar \tr'' = \Rule{\tr'_1\Space\ldots\Space\tr'_k}{\eval{\conf}{\Unknown}}$.
If $i-k = 0$, hence $i = k$, we have $\tr'' = \tr'$, hence the thesis is immediate. 
If $i-k > 0$, hence $i > k$, again by \refToProp{tree-unknown}(\ref{prop:tree-unknown:2}), we have $\ResSet_\Unknown(\rt{\tr_k}) \in \ResSet$, hence $\tr'' \ruleAr \Rule{\tr'_1\Space\ldots\Space\tr'_k\Space\eval{\conf'}{\Unknown}}{\eval{\conf}{\Unknown}} = \hat{\tr}$, where $\conf' = \ConfSet(\rt{\tr'_{k+1}})$.
Finally, by $HI$, we get $\hat{\tr} \ruleArStar \tr'$, \mbox{as needed. }
\end{proof}

\subsection{Equivalence of traces, $\Wrong$ and $\PEval$ semantics}\label{sect:etwp}

We prove  the three constructions to be equivalent to each other, thus providing a coherency result of the approach.
 In particular, first we show that $\PEval$ is conservative with respect to $\RuleSet$, and this ensures the three constructions are equivalent for finite computations. 
Then, we prove traces and $\Wrong$ constructions to be equivalent to $\PEval$ for diverging and stuck computations, respectively, and this ensures they  
 cover all possible cases. 
 
 \begin{theorem} \label{thm:eq-finite_2} 
$\valid{\RuleSet}{\eval{\conf}{\res}}$ iff $\Rule{}{\eval{\conf}{\Unknown}} \ruleArStar \tr$, where $\rt\tr = \eval{\conf}{\res}$.
\end{theorem}
\begin{proof}
$\valid{\RuleSet}{\eval{\conf}{\res}}$ implies $\Rule{}{\eval{\conf}{\Unknown}} \ruleArStar \tr$ where $\rt\tr = \eval{\conf}{\res}$. By definition, if $\valid{\RuleSet}{\eval{\conf}{\res}}$ holds, then there is a finite proof tree $\tr$ in $\RuleSet$ such that $\rt\tr = \eval{\conf}{\res}$. 
Since $\RuleSet \subseteq \UnRuleSet$, $\tr$ is a proof tree in $\UnRuleSet$ as well;
furthermore, $\Rule{}{\eval{\conf}{\Unknown}} \TreeOrd \tr$, hence by \refToPropItem{tree-ord-arrow}{2} we get the thesis. \\
$\Rule{}{\eval{\conf}{\Unknown}} \ruleArStar \tr$ where $\rt\tr = \eval{\conf}{\res}$ implies $\valid{\RuleSet}{\eval{\conf}{\res}}$. Since $\rt\tr = \eval{\conf}{\res}$, by \refToCor{complete-tree}, $\tr$ is complete, hence, 
it is a proof tree in $\RuleSet$, thus $\valid{\RuleSet}{\eval{\conf}{\res}}$ holds.
\end{proof}
 
 To relate trace semantics with $\PEval$, first we show that, in the $\PEval$ semantics, proof trees obtained as limits of infinite sequences of $\ruleAr$ steps can be characterised as the infinite proof trees which are \emph{well-formed}, in the sense that there is a unique infinite path, entirely labelled by incomplete judgments.
Formally, $\tr$ is \emph{well-formed}, if, for all $n \in \N$, there is $\alpha \in \dom(\tr)$ such that $\Len{\alpha} = n$ and $\tr(\alpha) = \eval{\conf}{\Unknown}$ for some $\conf \in \ConfSet$, and, 
for all $\alpha \in \dom(\tr)$, if $\ResSet_\Unknown(\tr(\alpha)) \in \ResSet$, then $\subtree{\tr}{\alpha}$ is finite. 

\begin{proposition} \label{prop:pt-seq}
The following properties hold:
\begin{enumerate}
\item\label{prop:pt-seq:1} for each increasing sequence $(\tr_n)_{n \in \N}$ of proof trees, the least upper bound $\Treelub \tr_n$ is a \mbox{proof tree}
\item\label{prop:pt-seq:2} for each strictly increasing sequence $(\tr_n)_{n \in \N}$ of finite proof trees, the least upper bound $\Treelub \tr_n$ is infinite and well-formed
\item\label{prop:pt-seq:3} for each well-formed infinite proof tree $\tr$, there is a strictly increasing sequence $(\tr_n)_{n \in \N}$ of finite proof trees such that $\tr = \Treelub \tr_n$.
\end{enumerate}
\end{proposition}
\begin{proof}
To prove point 1, set $\tr = \Treelub \tr_n$ and recall from \refToPropItem{tree-ord}{2} that $\tr(\alpha) = \tr_{k_\alpha}(\alpha)$, where $k_\alpha \in \N$ is the least index $n$ where $\tr_n(\alpha)$ is most defined. 
Note that, for all $\alpha \in \dom(\tr)$, $\trbr_\tr(\alpha)$ is finite, since, by definition of $\tr$, we have $\trbr_\tr(\alpha) = \max \{ \trbr_{\tr_n}(\alpha) \mid \alpha \in \dom(\tr_n) \}$, and this value is bounded because $\trbr_{\tr_n}(\alpha)$ is the number of premises of a rule, which is bounded by definition, see condition {\sf BP} at page \pageref{bp}.
Then, since $\trbr_\tr(\alpha)$ is finite, there is an index $n \in \N$ such that $\trbr_\tr(\alpha) = \trbr_{\tr_n}(\alpha)$, and, especially, this holds for $n = \max\{k_\alpha, k_{\alpha \trbr_\tr(\alpha)}\}$.
Therefore, we have that $\Rule{\tr(\alpha 1)\Space\ldots\Space\tr(\alpha \trbr_\tr(\alpha))}{\tr(\alpha)} = \Rule{\tr_n(\alpha 1)\Space\ldots\Space\tr_n(\alpha\trbr_\tr(\alpha))}{\tr_n(\alpha)} \in \UnRuleSet$, since $\tr_n$ is a proof tree in $\UnRuleSet$. 

To prove point 2, set $\tr = \Treelub \tr_n$, then, by point 1, we have that $\tr$ is a proof tree, hence we have only to check it is infinite and well-formed. 
Since the sequence is strictly increasing, we have that for all $n \in \N$ there is $h > n$ such that $\dom(\tr_n) \subset \dom(\tr_h)$. 
This can be proved by induction on the number of $\Unknown$ in $\tr$, which is finite since $\tr_n$ is finite, noting that, if $\dom(\tr_n) = \dom(\tr_{n+1})$, since $\tr_n \StrTreeOrd \tr_{n+1}$, there is at least one node $\alpha \in \dom(\tr_n)$ such that $\ResSet_\Unknown(\tr_n(\alpha)) = \Unknown$ and $\ResSet_\Unknown(\tr_{n+1}(\alpha)) = \res$. 
Therefore, $\dom(\tr) = \bigcup_{n \in \N} \dom(\tr_n)$ is infinite, that is, $\tr$ is infinite. 
To show that $\tr$ is well-formed, first note that for all $\alpha \in \dom(\tr)$ such that $\tr(\alpha) = \eval{\conf}{\res}$, since $\tr_{k_\alpha} \TreeOrd \tr$ and $\tr_{k_\alpha}(\alpha) = \tr(\alpha)$, by definition of $\TreeOrd$, we get $\subtree{\tr_{k_\alpha}}{\alpha} = \subtree{\tr}{\alpha}$, hence, $\subtree{\tr}{\alpha}$ is finite. 
Then, we still have only to prove that, for each $n \in \N$, there is $\alpha \in \dom(\tr)$ such that $\Len{\alpha} = n$ and  $\ResSet_\Unknown(\tr(\alpha)) = \Unknown$. 
We proceed by induction on $n$. 
For $n = 0$, we have $\ResSet_\Unknown(\rt\tr) = \Unknown$, since, otherwise, we would have $\ResSet_\Unknown(\rt{\tr_{k_\EList}}) = \res$, hence, by \refToPropItem{tree-ord}{1}, we would get $\tr_{k_\EList} = \tr_{k_\EList + 1}$ which is not possible, because the sequence is strictly increasing. 
Now, by induction hypothesis, we know there is $\alpha \in \dom(\tr)$ such that $\Len{\alpha} = n$ and $\ResSet_\Unknown(\tr(\alpha)) = \Unknown$. 
By \refToPropItem{tree-unknown}{1}, we also know that, if there is $\beta \in \dom(\tr)$ such that $\Len{\beta} = n+1$ and $\ResSet_\Unknown(\tr(\beta)) = \Unknown$, then $\beta = \alpha h$ for some $h \in \NPos$. 
If such $h$ did not exist, then, for all $k \in \NPos$ such that $\alpha k \in \dom(\tr)$, we would have $\ResSet_\Unknown(\tr(\alpha k)) \in \ResSet$, hence, as we have just proved, $\subtree{\tr}{\alpha k}$  would be finite, and this would imply that $\tr$ is finite, which is not possible. 
Hence, $\tr$ is well-formed as needed. 

To prove point 3, for all $n \in \N$, consider the proof tree $\tr_n$ defined as follows:
let $\alpha_n \in \dom(\tr)$ be the (unique thanks to \refToPropItem{tree-unknown}{2}) node such that $\Len{\alpha_n} = n$ and $\ResSet_\Unknown(\tr(\alpha_n)) = \Unknown$, 
then define $\tr_n(\beta) = \tr(\beta)$ for all $\beta \ne \alpha_n \beta'$, with $\beta' \in \NEList{\NPos}$, and undefined otherwise.  
We have $\tr_n \TreeOrd \tr_{n+1}$, since, by \refToPropItem{tree-unknown}{1}, $\alpha_{n+1} = \alpha_n i$ for some $i \in \NPos$. 
Finally, by construction, we have $\tr = \Treelub \tr_n$, as needed. 
\end{proof}

Then, we define a function $\mathsf{erase}$ that transforms a proof tree in $\TrRuleSet$ in one in $\UnRuleSet$, by essentially erasing traces.
The definition is given coinductively by the following equations:

\begin{quote}
\begin{small}
$\EraseTr{\DecoratedTree{\res}{\Rule{}{\evaltr{\res}{\res}}}}=\DecoratedTree{\res}{\Rule{}{\eval{\res}{\res}}}$\\[2ex]
$\EraseTr{\DecoratedTree{\tracerule{\rho}{t_1,\,\ldots,\,t_{n+1}}}{\Rule{\tr_1\Space\ldots\Space\tr_{n+1}}{\evaltr{\conf}{t'\cdot \res}}}}=\DecoratedTree{\rho}{\Rule{\EraseTr{\tr_1}\Space\ldots\Space\EraseTr{\tr_{n+1}}}{\eval{\conf}{\res}}}$\\[2ex]
$\EraseTr{\DecoratedTree{\divtracerule{\rho}{i}{t_1,\,\ldots,\,t_{i-1}}{t}}{\Rule{\tr_1\Space\ldots\Space\tr_i}{\evaltr{\conf}{t'}}}}=\DecoratedTree{\proprule{\rho}{i}{\Unknown}}{\Rule{\EraseTr{\tr_1}\Space\ldots\Space\EraseTr{\tr_i}}{\eval{\conf}{\Unknown}}}$
\end{small}
\end{quote}

By construction, $\EraseTr{\tr}$ is a proof tree in $\UnRuleSet$ and it is infinite and well-formed iff $\tr$ \mbox{is infinite.}

\begin{lemma} \label{lem:trace-to-unknown}
If $\valid{\TrRuleSet}{\evaltr{\conf}{t}}$ holds by an infinite proof tree $\tr^\trlb$, then there is a sequence $(\tr'_n)_{n \in \N}$ such that $\tr'_n \ruleAr \tr'_{n+1}$ for all $n\in \N$, $\tr'_0 = \Rule{}{\eval{\conf}{\Unknown}}$,  and $\Treelub \tr'_n = \EraseTr{\tr^\trlb}$.
\end{lemma} 
\begin{proof}
Since $\tr^\trlb$ is infinite,  $\EraseTr{\tr^\trlb} = \tr$ is a well-formed infinite proof treee in $\UnRuleSet$ and, by \refToPropItem{pt-seq}{3}, there is a strictly increasing sequence $(\tr_n)_{n \in \N}$ of finite proof trees in $\UnRuleSet$ such that $\Treelub \tr_n = \tr$ and $\tr_0 = \Rule{}{\eval{\conf}{\Unknown}}$. 
By \refToPropItem{tree-ord-arrow}{2},  since for all $n \in \N$ we have $\tr_n \StrTreeOrd \tr_{n+1}$, we get $\tr_n \ruleArStar \tr_{n+1}$, and, since $\tr_n \ne \tr_{n+1}$, this sequence of steps is not empty. 
Hence, we can construct a sequence $(\tr'_n)_{n \in \N}$ such that $\tr'_0 = \Rule{}{\eval{\conf}{\Unknown}}$, $\tr'_n \ruleAr \tr'_{n+1}$ and $\Treelub \tr'_n = \tr$, as needed. 
\end{proof}

\begin{lemma} \label{lem:unknown-to-trace}
If $\tr$ is a well-formed infinite proof tree in $\UnRuleSet$ with $\rt\tr = \eval{\conf}{\Unknown}$, then there is an infinite trace $t_\tr \in \InfList{\ConfSet}$ and an infinite proof tree $\tr^\trlb$ in $\TrRuleSet$ such that $\rt{\tr^\trlb} = \evaltr{\conf}{t_\tr}$. 
\end{lemma}
\begin{proof}
Let us denote by $\alpha_n$ the (unique thanks to \refToPropItem{tree-unknown}{2}) node in $\tr$ such that $\Len{\alpha_n} = n$ and $\tr(\alpha_n) = \eval{\conf_n}{\Unknown}$, 
and denote by $\hat{\rho}_n \equiv \proprule{\rho_n}{k_n}{\Unknown}$ the rule applied at $\alpha_n$, where $\rho_n \equiv \inlinerule{\judg_1^n\ldots\judg_{h_n}^n}{\judg_{h_n+1}^n}{\conf_n}$ and $k_n \le h_n +1$. 
Since $\tr$ is well-formed, for all $k < k_n$, we have $\subtree{\tr}{\alpha_n k}$ is finite, hence it is a valid proof tree in $\RuleSet$, thus $\valid{\RuleSet}{\judg_k^n}$ holds, and $\ConfSet(\judg_{k_n}^n) = \conf_{n+1}$ by \refToProp{tree-unknown}. 
Then, the set $\Spec = \{ \conf_n \mid  n \in \N\}$ with the rules $\rho_n$ and indexes $k_n$ satisfies the hypothesis of \refToLem{div-consistency}, hence we get that, for all $n \in \N$, there is an infinite trace $t_n \in \InfList{\ConfSet}$ such that $\valid{\TrRuleSet}{\evaltr{\conf_n}{t_n}}$. 
Now, set $t = t_0$, since $\conf = \conf_0$ by construction, we have just proved that $\evaltr{\conf}{t}$ is derivable in $\TrRuleSet$ by a proof tree $\tr^\trlb$ and, since $t$ is infinite, $\tr^\trlb$ is infinite as well. 
\end{proof}

\begin{theorem}  \label{thm:eq-trace-unknown} 
$\valid{\TrRuleSet}{\evaltr{\conf}{t}}$ for some $t \in \ConfSet^\omega$
iff  $\Rule{}{\eval{\conf}{\Unknown}} \ruleArInf$.
\end{theorem}
\begin{proof}
$\valid{\TrRuleSet}{\evaltr{\conf}{t}}$ for some $t \in \ConfSet^\omega$ implies $  \Rule{}{\eval{\conf}{\Unknown}} \ruleArInf$. Since $\valid{\TrRuleSet}{\evaltr{\conf}{t}}$ holds and $t$ is infinite, by (a consequence of) \refToProp{evaltr-fin-tr}, there is an infinite proof tree $\tr^\trlb$ in $\TrRuleSet$ such that $\rt{\tr^\trlb} = \evaltr{\conf}{t}$. 
Then by \refToLem{trace-to-unknown} we get the thesis. \\
$\Rule{}{\eval{\conf}{\Unknown}} \ruleArInf$ implies $ \valid{\TrRuleSet}{\evaltr{\conf}{t}}$ for some $t \in \ConfSet^\omega$.  By definition of $\ruleArInf$, there is an infinite sequence $(\tr_n)_{n \in \N}$ such that $\tr_0 = \Rule{}{\eval{\conf}{\Unknown}}$ and, for all $n \in \N$, $\tr_n \ruleAr \tr_{n+1}$, hence, by \refToPropItem{tree-ord-arrow}{1}, we get $\tr_n \StrTreeOrd \tr_{n+1}$. 
By \refToPropItem{pt-seq}{2}, we have that $\tr = \Treelub \tr_n$ is a well-formed infinite proof tree, hence, by \refToLem{unknown-to-trace}, we get the thesis. 
\end{proof}

We give now a lemma useful to prove the equivalence for wrong computations.
We say that a (finite) proof tree $\tr$ in $\UnRuleSet$ is \emph{irreducible} if there is no $\tr'$ such that $\tr \ruleAr \tr'$, and it is \emph{stuck} if it is irreducible and $\ResSet_\Unknown(\rt\tr) = \Unknown$. 
Note that, by \refToPropItem{tree-ord}{1} and \refToPropItem{tree-ord-arrow}{1}, a complete proof tree $\tr$  is irreducible. 

\begin{lemma} \label{lem:stuck-to-wrong}
If $\tr$ is a stuck proof tree with $\rt\tr = \eval{\conf}{\Unknown}$, then {\em $\valid{\WrRuleSet}{\eval{\conf}{\Wrong}}$} holds.
\end{lemma}
\begin{proof}
We proceed by induction on $\tr$, analyzing cases on the last applied rule. 
There are three cases:
\begin{description}
\item [axiom] If an axiom is applied, then, since $\tr$ is stuck, there is no rule $\rho \in \RuleSet$ such that $\ConfSet(\rho) = \conf$, hence $\valid{\WrRuleSet}{\eval{\conf}{\Wrong}}$ holds, by applying the axiom.
\item [$\Unknown$-introduction] If $\unknownrule{\rho}{i}{\res}$ is applied, then, since $\tr$ is stuck, there is no rule $\rho' \sim_i \rho$ with $\ResSet(\rho', i) = \res$, hence $\wrongrule{\rho}{i}{\res} \in \WrRuleSet$, and applying this rule we get $\valid{\WrRuleSet}{\eval{\conf}{\Wrong}}$. 
\item [$\Unknown$-propagation] If $\proprule{\rho}{i}{\Unknown}$ is applied, set $\conf_i = \ConfSet(\rho, i)$, then, since $\tr$ is stuck, the subtree $\subtree{\tr}{i}$ is stuck as well and $\rt{\subtree{\tr}{i}} = \eval{\conf_i}{\Unknown}$;
hence, by induction hypothesis, we get $\valid{\WrRuleSet}{\eval{\conf_i}{\Wrong}}$ holds, thus, applying the rule $\proprule{\rho}{i}{\Wrong}$, we get the thesis. 
\end{description}
\end{proof}

\begin{theorem}   \label{thm:eq-wrong-unknown}
{\em $\valid{\WrRuleSet}{\eval{\conf}{\Wrong}}$}
iff $\Rule{}{\eval{\conf}{\Unknown}}\ruleArStar \tr$, where $\tr$ is stuck.
\end{theorem}
\begin{proof}
$\valid{\WrRuleSet}{\eval{\conf}{\Wrong}}$ implies $ \Rule{}{\eval{\conf}{\Unknown}}\ruleArStar \tr$ where $\tr$ is stuck. We prove that there is a stuck tree $\tr$ with $\rt\tr = \eval{\conf}{\Unknown}$, then the thesis follows immediately from \refToPropItem{tree-ord-arrow}{2}. 
The proof is by induction on rules.
It is enough to consider only rules with $\Wrong$ in the conclusion, hence we have the following three cases: 
\begin{description}
\item [axiom] By definition, there is no rule $\rho \in \RuleSet$ such that $\ConfSet(\rho) = \conf$, hence $\Rule{}{\eval{\conf}{\Unknown}}$ is stuck.
\item [$\Wrong$-introduction] By definition of $\wrongrule{\rho}{i}{\res}$, with $\rho \equiv \inlinerule{\judg_1\ldots\judg_n}{\judg_{n+1}}{\conf}$, there is no rule $\rho' \sim_i \rho$ such that $\ResSet(\rho', i) = \res$;
then, by \refToThm{eq-finite_1} for each $\judg_k$, with $k \le i$, there is a finite proof tree $\tr_k$, with $\rt{\tr_k} = \judg_k$, hence by applying the rule $\unknownrule{\rho}{i}{\res}$ we get a proof tree which is stuck, by definition of $\ruleAr$. 
\item [$\Wrong$-propagation] For a rule $\proprule{\rho}{i}{\Unknown}$ with $\rho \equiv \inlinerule{\judg_1\ldots\judg_n}{\judg_{n+1}}{\conf}$ and $\conf_i = \ConfSet(\judg_i)$, we have, by induction hypothesis, that there is a stuck tree $\tr'$ such that $\rt{\tr'} = \eval{\conf_i}{\Unknown}$; 
then, by \refToThm{eq-finite_1}, for each $k<i$, there is a proof tree $\tr_k$ such that $\rt{\tr_k} = \judg_k$, hence, by applying $\proprule{\rho}{i}{\Unknown}$ to $\tr_1, \ldots, \tr_{i-1}, \tr'$ we get a stuck tree. 
\end{description}
$\Rule{}{\eval{\conf}{\Unknown}}\ruleArStar \tr$ where $\tr$ is stuck implies $\valid{\WrRuleSet}{\eval{\conf}{\Wrong}}$. It follows immediately from \refToLem{stuck-to-wrong}, since $\rt\tr = \eval{\conf}{\Unknown}$ by hypothesis. 
\end{proof}

\subsection{Soundness with respect to partial evaluation semantics}\label{sect:spev}

 $\PEval$ semantics enjoys both soundness-must and soundness-may properties, giving a way to establish an explicit link between the previous two constructions. 
The statements are the following: 
\begin{description}
\item[soundness-must] $\PEval$ If $\conf \in \Pred$ and $\Rule{}{\eval{\conf}{\Unknown}}\ruleArStar \tr$, then 
either $\tr$ is complete or there is $\tr'$ such that $\tr \ruleAr \tr'$. 
\item[soundness-may] $\PEval$ If $\conf \in \Pred$, then either $\Rule{}{\eval{\conf}{\Unknown}}\ruleArStar \tr$ 
where $\tr$ is complete or $\Rule{}{\eval{\conf}{\Unknown}} \ruleArInf$. 
\end{description}
Hence, we could also prove the correctness of the proposed proof techniques using the $\PEval$ approach. 
Here we report the proof for soundness-must, as it is useful to show where, in the evaluation process, the three conditions come into play. 

Recall that $\UnRuleSet$ is the extension of $\RuleSet$ with incomplete judgements $\eval{\conf}{\Unknown}$. 
In this approach, the semantics is modelled by a reduction relation on finite proof trees in $\UnRuleSet$. 
 We extend  the indexed predicate $(\Pred_\idx)_{\idx\in \IdxSet}$ to finite proof trees $\tr$  and  
we will write $\TreeSat{\tr}{\Pred}$ if 
$\ConfSet(\rt\tr) \in \Pred$.

Soundness-must  with respect to  the $\PEval$ semantics follows, as usual, from progress and subject reduction. Note that, for the reduction relation on proof trees, the latter is trivial since the configuration at the root never changes in a reduction sequence. 
For the proof of progress we need the following proposition.

\begin{proposition} \label{prop:tree-preservation}
For any {proof tree} $\tr$ in $\UnRuleSet$, if $\TreeSat{\tr}{\Pred}$, then, for all $\alpha \in \dom(\tr)$, $\ConfSet(\tr(\alpha)) \in \Pred$. 
\end{proposition}
\begin{proof}
The proof is by induction on the length of $\alpha$, namely, on the level of the node $\alpha$ in the tree $\tr$. 
If $\Len{\alpha} = 0$, then $\alpha = \EList$, hence $\ConfSet(\tr(\alpha)) = \ConfSet(\rt\tr)$, hence the thesis holds by hypothesis, since $\TreeSat{\tr}{\Pred}$. 
Now, assume the thesis for $\alpha$ and prove it for $\alpha k$ for some $k \in \NPos$. 
Since $\tr$ is a proof tree, there is a rule $\rho^\Unknown$ with conclusion $\tr(\alpha)$ and such that $\tr(\alpha k)$ is its $k$-th premise. 
By construction of rules in $\UnRuleSet$, for all $h < k$, we have $\ResSet_\Unknown(\tr(\alpha h)) \in \ResSet$, hence, by \refToCor{complete-tree}, $\subtree{\tr}{\alpha h}$ is complete,  thus it is a finite proof tree in $\RuleSet$, and so we get $\valid{\RuleSet}{\tr(\alpha h)}$.  
Then, by induction hypothesis, we have $\ConfSet(\tr(\alpha)) \in \Pred$, hence, by \refToProp{preservation}, we get $\ConfSet(\tr(\alpha k)) \in \Pred$, as needed. 
\end{proof}

\begin{lemma}[Progress for $\ruleAr$] \label{lem:progress-ruleAr}
For each finite proof tree $\tr$ in $\UnRuleSet$, if $\TreeSat{\tr}{\Pred}$, then either $\tr$ is complete or there is $\tr'$ such that $\tr \ruleAr \tr'$. 
\end{lemma}
\begin{proof}
The proof is by induction on $\tr$.
We split cases on the last applied rule.
\begin{itemize}
\item If $\tau=\DecoratedTree{\rho}{\Rule{\ldots}{\eval{\conf}{u}}}$, for $\rho \in \RuleSet$, then $u\in \ResSet$, hence, by \refToCor{complete-tree}, $\tr$ is complete. 
\item If $\tr = \DecoratedTree{\res_\Unknown}{\Rule{}{\eval{\res}{\Unknown}}}$, then $\tr \ruleAr \Rule{}{\eval{\res}{\res}}$.
\item If $\tr = \DecoratedTree{\conf_\Unknown}{\Rule{}{\eval{\conf}{\Unknown}}}$ with $\conf \notin \ResSet$,  then, since $\TreeSat{\tr}{\Pred}$, we have $\conf \in \Pred$, hence, by condition \refToSound{progress-ex}, there is $\rho \equiv \inlinerule{\judg_1\ldots\judg_n}{\judg_{n+1}}{\conf}$.
Therefore, we get $\tr \ruleAr \Rule{\eval{\ConfSet(\judg_1)}{\Unknown}}{\eval{\conf}{\Unknown}}$. 
\item If $\tr =\DecoratedTree{\unknownrule{\rho}{i}{\res}}{ \Rule{\tr_1\ \ldots\ \tr_i}{\eval{\conf}{\Unknown}}}$, then $\rt{\tr_i} = \eval{\ConfSet(\rho, i)}{\res}$ and so $\ResSet_\Unknown(\rt{\tr_i}) = \res$. Hence, by \refToCor{complete-tree}, $\tr_i$ is complete, thus we get $\valid{\RuleSet}{\eval{\ConfSet(\rho, i)}{\res}}$.
Then,  by condition \refToSound{progress-all}, there is $\rho' \sim_i \rho$ such that $\ResSet(\rho', i) = \res$ and there are two cases:
\begin{itemize}
\item if $\#
(\rho') = i$, then $\tr \ruleAr \Rule{\tr_1\ \ldots\ \tr_i}{\eval{\conf}{\res}}$
\item if $\#
(\rho') > i$, then $\tr \ruleAr \Rule{\tr_1\ \ldots\ \tr_i\ \eval{\ConfSet(\rho', i+1)}{\Unknown}}{\eval{\conf}{\Unknown}}$. 
\end{itemize}
\item If $\tr =\DecoratedTree{\proprule{\rho}{i}{\Unknown}}{\Rule{\tr_1\ \ldots\ \tr_i}{\eval{\conf}{\Unknown}}}$,  then $\rt{\tr_i} = \eval{\conf'}{\Unknown}$ and, since $\TreeSat{\tr}{\Pred}$, by \refToProp{tree-preservation}, we get $\conf' \in \Pred$, that is, $\TreeSat{\tr_i}{\Pred}$. 
Then, by induction hypothesis, either $\tr_i$ is complete, or $\tr_i \ruleAr \tr'_i$ for some $\tr'_i$;
but $\tr_i$ cannot be complete since $\rt{\tr_i} = \eval{\conf'}{\Unknown}$, hence $\tr_i \ruleAr \tr'_i$, 
and this implies $\tr \ruleAr \Rule{\tr_1\ \ldots\ \tr_{i-1}\ \tr'_i}{\eval{\conf}{\Unknown}}$. 
\end{itemize}
\end{proof}

\begin{theorem} \label{thm:sound-unknown}
If $\TreeSat{\tr}{\Pred}$ and $\tr \ruleArStar \tr'$, then 
either $\tr'$ is complete or there is $\tr''$ such that $\tr' \ruleAr \tr''$. 
\end{theorem}
\begin{proof}
By induction on the number of steps in $\tr \ruleArStar \tr'$: 
if it is equal to $0$, then $\tr = \tr'$ and the thesis follows by \refToLem{progress-ruleAr}, 
otherwise, we have $\tr \ruleAr \tr_1 \ruleArStar \tr'$ and, since $\ConfSet(\rt\tr) = \ConfSet(\rt{\tr_1})$ and $\TreeSat{\tr}{\Pred}$, we have $\TreeSat{\tr_1}{\Pred}$, hence we get the thesis by induction hypothesis. 
\end{proof}

\section{Related work}\label{sect:rw}
\paragraph{Modeling divergence} The issue of modelling divergence in big-step semantics dates back to \citet{CousotCousot92}, where a stratified approach with a separate coinductive judgment for divergence is proposed,
also investigated in \citet{LeroyGrall09}.

In \citet{Ancona12} the authors models divergence by interpreting coinductively standard big-step rules and considering also non-well-founded values. 
In \citet{Chargueraud13} a similar technique is exploited, 
by adding a special result modelling divergence.
Flag-based big-step semantics \cite{PoulsenMosses17} captures divergence
by interpreting the same semantic rules both inductively and coinductively. 
In all these approaches, spurious judgements can be derived  for diverging computations.  

Other proposals \cite{OwensMKT16,AminRompf17}
are inspired by the notion of definitional interpreter  \cite{R98}, 
where a counter limits the number of steps
 of  a computation. 
 Thus, divergence can be modelled on top of an inductive judgement: a program diverges if the timeout is raised for any value of the counter, hence it is not directly modelled in the definition.  
 Instead, \citet{Danielsson12} provides a way to directly model divergence using definitional interpreters, relying on the coinductive partiality monad \cite{Capretta05}. 
  
The trace semantics in \refToSect{traces} has been inspired by \cite{KusmierekB10}. Divergence propagation rules are very similar to those used in \cite{AnconaDZ@OOPSLA17,AnconaDZ@ECOOP18} to define a big-step judgment which directly includes divergence as result. However, this direct definition relies on  a non-standard notion of inference system, allowing \emph{corules}  \cite{AnconaDZ@ESOP17,Dagnino19}, whereas for the trace semantics presented in this work standard coinduction is enough, since all rules are \emph{productive}, that is, they always add an element \mbox{to the trace.} 


Differently from all the previously cited papers which consider specific examples, the work \citet{Ager04} shares with us the aim of providing a \emph{generic construction} to model non-termination, basing on an arbitrary big-step semantics. 
Ager considers a class of big-step semantics identified by a specific shape of rules, and defines, in a small-step style, a proof-search algorithm which follows the big-step rules;
in this way, converging, diverging and stuck computations are distinguished. 
This approach is somehow similar to our \PEval\ semantics, even tough the transition system we propose is directly defined on proof trees.

 There is an extensive body of work on coalgebraic techniques, where the difference between semantics can be simply expressed by a change of functor. In this paper we take a set-theoretic approach, simple and accessible to a large audience. Furthermore, as far as we know \cite{Rutten00}, coalgebras abstract several kinds of transition systems, thus being more similar to a small-step approach.  In our understanding, the coalgebra models a single computation step with possible effects, and from this it is possible to derive a unique morphism into the final coalgebra modelling the ``whole'' semantics. Our trace semantics, being big-step, seems to roughly correspond to directly get this whole semantics. In other words, we do not have a coalgebra structure \mbox{on configurations.}

\paragraph{Proving soundness} As we have discussed, also proving (type) soundness with respect to a big-step semantics is a challenging task, and some approaches have been proposed in the literature.
 In \cite{EOC06}, to show soundness of large steps semantics, they prove a coverage lemma, which ensures that
the rules cover all cases, including error situations. 
In \citet{LeroyGrall09} the authors prove a soundness property similar to \refToThm{sound-traces}, but  by  using a separate judgment to represent divergence, thus avoiding  using traces.
In \citet{Ancona12} there is a proof of soundness of a coinductive type system with respect to a coinductive big-step semantics for a Java-like language, defining a relation between derivations in the type system and in the big-step semantics. 
In \citet{AnconaDZ@OOPSLA17} there is a proof principle, used to show type soundness with respect to a big-step semantics defined by an inference system with  corules \cite{AnconaDZ@ESOP17}. 
In \citet{AminRO14} the proof of type soundness of a calculus formalising path-dependent types relies on a big-step semantics, while 
in \cite{AminRompf17} soundness is shown for the polymorphic type systems $F_{<:}$, and for the DOT calculus, using definitional interpreters to model the semantics. 
In both cases they extend the original semantics adding error and timeout, and adopt inductive proof strategies,  as in \citet{Siek13}. 
A similar approach is followed by \citet{OwensMKT16} to show type soundness of the Core ML language. 

Also \citet{Ancona14} proposes an inductive proof of type soundness  for  the big-step semantics of a Java-like language, but relying  on a notion of approximation of infinite derivation in the \mbox{big-step semantics.}

 Pretty big-step semantics \citet{Chargueraud13} aims at providing an efficient representation of big-step semantics, so that it can be easily extended without duplication of meta-rules. 
In order to define and prove soundness, they propose a generic error rule based on a \emph{progress judgment}, whose definition can be
easily derived manually from the set of evaluation rules. This is partly similar to our $\Wrong$ extension, with two main differences. First, by factorising rules, they introduce intermediate steps as in small-step semantics, hence there are similar problems when intermediate steps are ill-typed (as in \refToSect{fjl}, \refToSect{fjos}). Second, $\Wrong$ introduction is handled by the progress judgment, that is, at the level of side-conditions. 
 Moreover, in \cite{BodinJS15} there is a formalisation of the pretty-big-step rules for performing a generic reasoning on big-step semantics by using abstract interpretation. However, the authors say that they interpret rules inductively, hence non-terminating computations are not modelled.
 
 Finally, some (but not all) infinite trees of our trace semantics can be seen as cyclic proof trees, see end of \refToSect{traces}. Proof systems supporting cyclic proofs can be found, e.g., in \cite{Brotherston05,BrotherstonS11} for classical first order logic with inductive definitions.

\section{Conclusion and future work}\label{sect:conclu}

The most important contribution is a general approach  for reasoning  on soundness  with respect to a  big-step operational semantics. 
Conditions can be proven by a case analysis on the semantic (meta-)rules avoiding small-step-style intermediate configurations. This can be crucial since there are  calculi  where the property to be checked is \emph{not preserved} by such intermediate configurations, whereas it holds for the final result,   as illustrated in \refToSect{examples}. 

In future work, we plan to use the meta-theory in \refToSect{framework} as basis to investigate yet other constructions, notably the approach relying on  corules  \cite{AnconaDZ@OOPSLA17,AnconaDZ@ECOOP18}, and that, adding a counter, based on timeout \cite{OwensMKT16,AminRompf17}.

We also plan to compare our proof technique for proving soundness with the standard one for small-step semantics: 
if a predicate satisfies progress and subject reduction with respect to a small-step semantics, does it satisfy our soundness conditions  with respect to  an equivalent big-step semantics? 
To formally prove such a statement, the first step will be to express equivalence between small-step and big-step  semantics. 
On the other hand, 
the converse does not hold, as  shown 
by the examples in \refToSect{fjl} and \refToSect{fjos}. 

For what concerns significant applications, we plan to use the approach to prove soundness  for  the $\lambda$-calculus with full reduction and intersection/union types \cite{BDL95}. The interest of this example lies in the failure of  the  subject reduction, as discussed in \refToSect{fjos}.
In another direction, we want to enhance $\MiniFJOS$ with {$\lambda$-abstractions} and  allowing everywhere intersection and union types \cite{DGV19}. This will extend typability of shared expressions.  
 We plan to apply our approach to the big-step semantics of the statically typed virtual classes calculus developed in \cite{EOC06}, discussing also the non terminating computations not considered there. 

 With regard to proofs,  we plan to investigate if we can simplify them by means of enhanced conductive techniques. 

 As a proof-of-concept, we provided a mechanisation\footnote{Available at \url{https://github.com/fdgn/soundness-big-step-semantics}.} in Agda of \refToLem{sr}. 
The mechanisations of the other proofs is similar.
However, as future work, we think it would be more interesting  to provide a software  for writing  big-step definitions and  for checking  that the soundness conditions hold.

\bibliography{main}

\end{document}